\def\submission{0}
\def\note{0}

\ifnum\submission=0
\documentclass[11pt,letterpaper]{article}
\else
\documentclass{llncs}
\fi
\usepackage[utf8]{inputenc}
\usepackage{amsmath}
\usepackage{amsfonts}
\usepackage{xcolor}
\usepackage{bm}
\usepackage{soul}
\ifnum\submission=0
\usepackage{amsthm}
\fi
\usepackage{braket}
\usepackage{authblk}
\usepackage[colorlinks=true,allcolors=blue]{hyperref}
\ifnum\submission=0
\usepackage[margin=1in]{geometry}
\fi
\usepackage{xcolor}
\usepackage[capitalize]{cleveref}
\usepackage{algorithm}
\usepackage{mdframed} 
\mdfdefinestyle{figstyle}{ %
  linecolor=black!7, %
  backgroundcolor=black!7, %
  innertopmargin=10pt, %
  innerleftmargin=\textwidth, %
  innerrightmargin=\textwidth, %
  innerbottommargin=5pt %
}

\usepackage{comment}

\ifnum\submission=0
\newtheorem{theorem}{Theorem}[section]

\newtheorem{lemma}[theorem]{Lemma}
\newtheorem{claim}[theorem]{Claim}

\newtheorem{corollary}[theorem]{Corollary}
\newtheorem{definition}[theorem]{Definition}
\newtheorem{observation}{Observation}
\newtheorem{remark}{Remark}

\fi

\def \sample { \overset{\hspace{0.1em}\mathsf{\scriptscriptstyle\$}}{\leftarrow} }

\newcommand{\ra}{\rightarrow}

\newcommand{\pro}{P}
\newcommand{\ver}{V}
\newcommand{\secpar}{\lambda}
\newcommand{\negl}{\mathsf{negl}}
\newcommand{\lang}{L}
\newcommand{\rel}{R}

\newcommand{\bit}{\{0,1\}}
\newcommand{\aux}{{\sf aux}}
\newcommand{\poly}{\mathsf{poly}}
\newcommand{\OUT}{\mathsf{OUT}}
\newcommand{\NP}{\mathbf{NP}}

\newcommand{\BQP}{\mathbf{BQP}}
\newcommand{\compind}{\stackrel{comp}{\approx}}
\newcommand{\statind}{\stackrel{stat}{\approx}}
\newcommand{\hil}{\mathcal{H}}

\newcommand{\experiment}{\mathsf{Exp}}
\newcommand{\querysmall}{\mathsf{Q}_{\leq q}}
\newcommand{\event}{\mathsf{E}_{\regB=1}}
\newcommand{\Bisone}{\mathsf{E}_{\regB=1}}
\newcommand{\Contiszero}{\mathsf{E}_{\regcont=0}}
\newcommand{\Contisone}{\mathsf{E}_{\regcont=1}}
\newcommand{\keyspace}{\mathcal{K}}
\newcommand{\hashfamily}{\mathcal{H}}
\newcommand{\tildehashfamily}{\widetilde{\mathcal{H}}}
\newcommand{\tildeH}{\widetilde{H}}
\newcommand{\A}{\mathcal{A}}
\newcommand{\B}{\mathcal{B}}

\newcommand{\TD}{\mathsf{TD}}
\newcommand{\Tr}{\mathrm{Tr}}

\newcommand*{\reginp}{\mathbf{Inp}}
\newcommand*{\regS}{\mathbf{S}}

\newcommand*{\regR}{\mathbf{R}}
\newcommand*{\regX}{\mathbf{X}}

\newcommand*{\regcount}{\mathbf{Count}}
\newcommand*{\regcont}{\mathbf{Cont}}
\newcommand*{\regH}{\mathbf{H}}

\newcommand{\regW}{\mathbf{W}}

\newcommand*{\regM}{\mathbf{M}}
\newcommand*{\regout}{\mathbf{Out}}

\newcommand*{\regaux}{\mathbf{Aux}}

\newcommand*{\regA}{\mathbf{A}}

\newcommand*{\regB}{\mathbf{B}}
\newcommand{\regother}{\mathbf{Other}}
\newcommand{\regV}{\mathbf{V}}

\newcommand{\siml}{\mathsf{Sim}}
\newcommand{\func}{\mathsf{Func}}
\newcommand{\execution}[2]{\langle #1, #2 \rangle}
\newcommand{\calX}{\mathcal{X}}
\newcommand{\calY}{\mathcal{Y}}
\newcommand{\calZ}{\mathcal{Z}}
\newcommand{\vecx}{\mathbf{x}}
\newcommand{\vecy}{\mathbf{y}}
\newcommand{\inp}{\mathsf{inp}}
\newcommand{\mes}{m}
\newcommand{\vmes}{\bm{m}}
\newcommand{\vchal}{\bm{c}}
\newcommand{\vbeta}{\bm{\beta}}
\newcommand{\vone}{\bm{1}}
\newcommand{\messpace}{\mathcal{M}}
\newcommand{\chalspace}{\mathcal{C}}
\newcommand{\randspace}{\mathcal{R}}

\newcommand{\Acc}{\mathsf{Acc}}
\newcommand{\final}{\mathsf{final}}

\newcommand{\oracle}{\mathcal{O}}
\newcommand{\ora}{\mathcal{O}}
\newcommand{\reprogram}{\mathsf{Reprogram}}
\newcommand{\Hord}{H^{\mathsf{ord}}}
\newcommand{\Aord}{\widetilde{A}^{\mathsf{ord}}}

\newenvironment{boxfig}[2]{\begin{figure}[#1]\fbox{\begin{minipage}{0.97\linewidth}
                        \vspace{0.2em}
                        \makebox[0.025\linewidth]{}
                        \begin{minipage}{0.95\linewidth}
            {{
                        #2 }}
                        \end{minipage}
                        \vspace{0.2em}
                        \end{minipage}}}{\end{figure}}

\ifnum\note=0
\newcommand{\nai}[1]{}
\newcommand{\km}[1]{}
\newcommand{\qipeng}[1]{}
\newcommand{\takashi}[1]{}

\newcommand{\revise}[1]{#1}
\else
\newcommand{\nai}[1]{{\color{purple}[Nai: #1]}}
\newcommand{\km}[1]{{\color{brown}[KM:#1]}}
\newcommand{\qipeng}[1]{{\color{red}[Qipeng: #1]}}
\newcommand{\takashi}[1]{{\color{orange}[Takashi: #1]}}

\newcommand{\revise}[1]{{\color{purple}#1}}
\fi

\ifnum\submission=0
\title{On the Impossibility of Post-Quantum Black-Box Zero-Knowledge in Constant Rounds}
\else
\title{On the Impossibility of Post-Quantum Black-Box Zero-Knowledge in Constant Rounds}
\fi
\begin{document}

\ifnum\submission=0
\newcommand*{\email}[1]{\normalsize\href{mailto:#1}{#1}}
\author[1]{Nai-Hui Chia}
\author[2]{Kai-Min Chung}
\author[3]{Qipeng Liu}
\author[4]{Takashi Yamakawa\thanks{This work was done while the author was visiting Princeton University.}}
\affil[1]{QuICS, University of Maryland}
\affil[1]{Luddy School of Informatics, Computing, and Engineering, Indiana University Bloomington
\email{naichia@iu.edu}}
\affil[2]{Institute of Information Science, Academia Sinica \email{kmchung@iis.sinica.edu.tw}}
\affil[3]{Princeton University \email{qipengl@cs.princeton.edu}}
\affil[4]{NTT Secure Platform Laboratories \email{ takashi.yamakawa.ga@hco.ntt.co.jp}}
\else
\author{\empty}\institute{\empty} 
\fi

\maketitle
\ifnum\submission=0
\vspace{-7mm} 
\fi

\begin{abstract}
We investigate the existence of constant-round post-quantum black-box  zero-knowledge protocols for $\mathbf{NP}$. As a main result, we show that there is no constant-round post-quantum black-box  zero-knowledge argument for $\mathbf{NP}$ unless $\mathbf{NP}\subseteq \mathbf{BQP}$. As constant-round black-box zero-knowledge arguments for $\mathbf{NP}$ exist in the classical setting, our main result points out a fundamental difference between post-quantum and classical zero-knowledge protocols. Combining previous results, we conclude that unless $\mathbf{NP}\subseteq \mathbf{BQP}$,   constant-round post-quantum zero-knowledge protocols for $\mathbf{NP}$ exist if and only if we use non-black-box techniques or relax certain security requirements such as relaxing standard zero-knowledge to $\epsilon$-zero-knowledge.  Additionally, we also prove that three-round and public-coin constant-round post-quantum black-box $\epsilon$-zero-knowledge arguments for $\mathbf{NP}$ do not exist unless $\mathbf{NP}\subseteq \mathbf{BQP}$. 
\end{abstract}
\thispagestyle{empty}
 \newpage
 \setcounter{page}{1}   
 \section{Introduction}


Zero-knowledge (ZK) interactive proof, introduced by Goldwasser, Micali, and Rackoff~\cite{GolMicRac89}, is a fundamental primitive in cryptography. ZK protocols provide privacy to the prover by proving a statement without revealing anything except that the statement is true even though the verifier is malicious. After many decades of study, what languages ZK protocols can express is quite understood. There have been many positive results for ZK protocols for particular languages, including quadratic residuosity \cite{GolMicRac89}, graph isomorphism \cite{JACM:GMW91}, statistical difference problem \cite{SahVad03} etc., and for all $\NP$ languages assuming one-way functions \cite{JACM:GMW91, Blum86}.


In addition to the expressiveness, round complexity is an important complexity measure for ZK protocols. One fascinating question regarding ZK is whether languages in $\mathbf{NP}$ have constant-round ZK protocols. In this aspect, the ZK protocol for 3-coloring~\cite{JACM:GMW91} is not ideal since that requires super-constant number of rounds \revise{if we require negligible soundness error. (In the following, we require negligible soundness error by default.)} Feige and Shamir~\cite{C:FeiSha89} and Brassard et al.~\cite{TCS:BCY91} presented constant-round ZK arguments\footnote{The protocol only guarantees to be computationally sound.} for $\mathbf{NP}$. Then, under reasonable cryptographic assumptions, Goldreich and Kahan \cite{JC:GolKah96} gave the first constant-round ZK proof\footnote{The protocol has statistical soundness.} for $\NP$. 


On the other hand, generalizing above results to obtain ZK protocols against malicious quantum verifiers is nontrivial. Briefly speaking, a protocol is ZK if there exists an efficient simulator such that for all malicious verifiers, the simulator can simulate the view generated by the prover and the malicious verifier. In the quantum setting, the malicious verifier can have quantum auxiliary input and can use quantum algorithms, which gives the verifier additional power to cheat even though the simulator is also quantum. This difference fails the security proofs of previous classical results. Specifically, those security proofs rely on a technique called \emph{rewinding}, enabling the simulator to complete the simulation by using only black-box access to the malicious verifier. This rewinding technique \revise{often} cannot be applied when an adversary is quantum due to the no-cloning theorem.  

Watrous~\cite{SIAM:Watrous09} presented the first classical ZK protocol against malicious quantum verifiers for languages in $\mathbf{NP}$. For simplicity, we call such a protocol \emph{post-quantum ZK protocol}. In particular, he introduced the quantum rewinding lemma and showed that, given black-box access to the malicious quantum verifier, there exists a quantum simulator assuming quantum-secure one-way functions. However, to achieve negligible soundness, Watrous's protocol needs super-constant number of rounds. Therefore, it does not satisfy the constant-round requirement. 

Recently, Bitansky and Shmueli~\cite{STOC:BitShm20} gave the first constant-round post-quantum ZK argument for $\mathbf{NP}$ assuming Quantum Learning with Error (QLWE) and Quantum Fully Homomorphic Encryption (QFHE) assumptions. However, their result relies on a novel technique for non-black-box simulation, i.e., the simulator requires the actual description of the malicious verifier instead of using it in a black-box manner. Along this line, Chia et al.~\cite{CCY20} presented a constant-round black-box $\epsilon$-ZK (BB $\epsilon$-ZK) argument for $\mathbf{NP}$ assuming quantum-secure one-way functions and a constant-round BB $\epsilon$-ZK proof for $\mathbf{NP}$ assuming QLWE (or more generally, the existence of collapsing hash function \cite{EC:Unruh16}). $\epsilon$-ZK is a security notion weaker than standard ZK. Roughly speaking, while standard ZK requires that the simulator can only fail with negligible probability, $\epsilon$-ZK allows the simulator to run in time $\poly(1/\epsilon)$ with failing probability at most $\epsilon$.

Nevertheless, all these results in~\cite{STOC:BitShm20,CCY20,SIAM:Watrous09} cannot achieve constant-round post-quantum black-box ZK (BBZK) for $\mathbf{NP}$. In contrast, constant-round classical ZK protocols can be obtained by black-box simulation~\cite{C:FeiSha89,JC:GolKah96,TCC:PasWee09,TCS:BCY91}. Observing this inconsistency between classical and quantum settings, one may start wondering if non-black-box simulation is necessary for post-quantum ZK or if we really need to sacrifice ZK security for black-box simulation. In this work, we aim to satisfy all these curiosities by answering the following question:  
\begin{itemize}
    \item[] \textsl{Do there exist constant-round post-quantum BBZK protocols for $\mathbf{NP}$?}
\end{itemize}

\paragraph{Classical impossibility results.} In the classical setting, certain constant-round BBZK protocols are unlikely to exist. Goldreich and Krawczyk~\cite{SIAM:GK96} showed that there do not exist three-round  BBZK protocols and public-coin constant-round BBZK protocols for $\mathbf{NP}$ unless $\mathbf{NP} \subseteq \mathbf{BPP}$. Barak and Lindell \cite{STOC:BarLin02} proved that there is no  constant-round BBZK protocol with strict-polynomial-time simulation  unless $\mathbf{NP} \subseteq \mathbf{BPP}$.\footnote{A ZK protocol has strict-polynomial-time simulation if the simulator always runs in a fixed polynomial time. } 
We note that a simulator is allowed to run in expected-polynomial-time in the standard  definition of the ZK property, which we also follow.
Indeed, all known constant-round BBZK protocols for $\NP$ rely on expected-polynomial-time simulation to circumvent the above impossibility result.




\subsection{Our Results}

In this work, we give a negative answer to the above question. In particular, we show that 
\begin{theorem} 
\label{thm:impossibility_main}
There do not exist constant-round post-quantum BBZK protocols for $\mathbf{NP}$ unless $\mathbf{NP} \subseteq \mathbf{BQP}$. 
\end{theorem}

We stress that  Theorem~\ref{thm:impossibility_main} rules out constant-round post-quantum BBZK protocols with \emph{expected-polynomial-time} simulation. 
This indicates a fundamental difference between classical and post-quantum BBZK. 
That is, although there exist constant-round classical BBZK protocols (with expected-polynomial-time simulation) for $\mathbf{NP}$, such a protocol does not exist in the quantum setting unless $\mathbf{NP} \subseteq \mathbf{BQP}$. 

Along this line, to fully understand the feasibility of various constant-round post-quantum ZK protocols for $\mathbf{NP}$, we also prove other impossibility results. 

\begin{theorem} 
\label{thm:impossibility_pbulic_coin}
There do not exist constant-round public-coin post-quantum BB $\epsilon$-ZK protocols for $\mathbf{NP}$ unless $\mathbf{NP} \subseteq \mathbf{BQP}$.
\end{theorem}
\begin{theorem} 
\label{thm:impossibility_3_round}
There do not exist three-round post-quantum BB $\epsilon$-ZK protocols for $\mathbf{NP}$ unless $\mathbf{NP} \subseteq \mathbf{BQP}$.
\end{theorem}

In summary, combining previous works on post-quantum ZK, we are able to give detailed characterizations of the feasibility of constant-round post-quantum BBZK and BB $\epsilon$-ZK protocols for $\mathbf{NP}$. We summarize them as follows:
\begin{enumerate}
    \item For post-quantum ZK, non-black-box simulation is sufficient~\cite{STOC:BitShm20} and necessary (Theorem~\ref{thm:impossibility_main}) for constant round. Otherwise, one has to relax the security of ZK to $\epsilon$-ZK for black-box simulation~\cite{CCY20}.  
    \item For post-quantum BB $\epsilon$-ZK, private coin is sufficient~\cite{CCY20} and necessary (Theorem~\ref{thm:impossibility_pbulic_coin}) for constant round. 
    \item In the last, although it is unlikely to have three-round post-quantum BB $\epsilon$-ZK protocols for $\mathbf{NP}$ by Theorem~\ref{thm:impossibility_3_round}, there exists a five-round post-quantum BB $\epsilon$-ZK protocol when assuming QLWE~\cite{CCY20}. Whether there exists a four-round post-quantum BB $\epsilon$-ZK protocol for $\mathbf{NP}$ is still open. 
\end{enumerate}

\subsection{Technical Overview}
\subsubsection{Impossibility of  Constant-Round ZK}
In this section, we start by recalling the classical impossibility result by Barak and Lindell~\cite{STOC:BarLin02} and provide overviews for our techniques that extend the classical impossibility to the quantum setting as well as expected polynomial time quantum simulator. 

We first fix some notations that will be used in this overview. 
Let $(P, V)$ be a (classical or post-quantum) zero-knowledge proof or argument system for a language $L$ with a BB simulator $\siml$, with \revise{negligible soundness error} and perfect completeness\footnote{
\revise{Though our main theorem also rules out protocols with negligible completeness error,  
we assume  perfect completeness in this overview for simplicity.}}. 
The number of rounds is a constant number $2 k - 1$, where the first message is sent by $P$.  We can assume all messages sent by the prover $P$ or the verifier $V$ are elements in a classical set $\messpace$.

\paragraph{Barak-Lindell Impossibility Result.} 
As observed by Barak and Lindell~\cite{STOC:BarLin02}, to construct constant-round
zero-knowledge proofs or arguments for languages not in $\mathbf{BPP}$, one has to either allow \emph{expected} polynomial time simulators or \emph{non-black-box} simulators (that make inherent use of the code of the verifier). More precisely, they show that all languages that have constant round zero-knowledge proofs or arguments with strict polynomial-time black-box (BB) simulators must be trivial, i.e. in $\mathbf{BPP}$. We give the proof sketch below and explain the potential barriers for quantizing this proof. 

The simulator $\siml$ with oracle access to a (dishonest) verifier $V^*$ uses  $V^*(x, \aux, \cdot)$ as a black-box routine. Here $V^*(x, \aux, \cdot)$ is the next-message function of $V^*$, which takes a statement $x$, an auxiliary input $\aux$,  a random tape $r$ and a message transcript $\vmes = (m_1, \cdots, m_i)$ (all messages sent by $P$, of length at most $k$), and outputs the next message sent to the prover.

To show an algorithm that decides $L$, we first define a random aborting verifier $V^*$. $V^*$ works almost in the same way as the honest verifier $V$, except on each input transcript $\vmes$, it  refuses to answer and aborts with some probability. In other words, let $H$ be a random oracle that independently on each input transcript $\vmes$ of length $i \leq k$, outputs $0$ (aborting) with probability $1-\epsilon_i$ and $1$ (non-aborting) with some non-negligible probability $\epsilon_i$
\footnote{In the actual proof by Barak and Lindell, the probability $\epsilon_i$ is set as $\epsilon^{2^i}$ for some chosen $\epsilon$.}; 
then $V^*$ works in the same way as $V$ if $H(\vmes_{j}) = 1$ for all prefix $\vmes_j$ of $\vmes$ and aborts otherwise. Note that $H$ is treated as the auxiliary input feed to $V^*$.

We are now ready to define an algorithm $\B$ that decides $L$: on input $x$, it samples a random tape $r$ and a random aborting oracle $H$, then runs $\siml(x)$ with oracle access to ${V^*(x, \aux:=(H, r), \cdot)}$; 
$\B$ outputs $1$ (accepts) if and only if the simulator outputs an accepting transcript.

The proof consists of two parts: 
\begin{itemize}
    \item On $x \in L$, $\B$ accepts with non-negligible probability. 
    \item On $x \not\in L$, $\B$ accepts with negligible probability. 
\end{itemize}
The first bullet point is easier, simply invoking zero-knowledge property. For the second bullet point, a more delicate argument is needed. The core of the proof is to turn $\B$ into a cheating prover that tries to prove a statement $x$ which is not in the language $L$. 

\vspace{1em}

For $x \in L$, 
the simulator will output the same transcript distribution as the distribution induced by the interaction between the honest prover $P$ and the random aborting verifier $V^*$. When $V^*$ aborts, it never gives an accepting transcript. When $V^*$ never aborts, the transcript is always an accepting one, by the perfect completeness of the underlying proof system $\Pi = (P, V)$. 
Since in each round $V^*$  aborts with probability $1 - \epsilon_i$, the probability that $V^*$ never aborts in the execution is $\epsilon^* = \prod_i \epsilon_i$. 
By zero-knowledge property, the simulator will output an accepting transcript with probability roughly $\epsilon^*$, which is non-negligible as all $\epsilon_i$ are chosen to be some non-negligible function. Thus, $\B$ on input $x \in L$, accepts with non-negligible probability. 

\vspace{1em}

For $x \not\in L$, we argue that the simulator will almost never output an accepting transcript, thus $\B$ on $x \not\in L$ never outputs $1$.  By a delicate argument\footnote{By choosing each $\epsilon_i$ properly. This is the place where the proof requires the running time of the simulator is strict polynomial, instead of expected polynomial. Otherwise, such $\epsilon_i$ may not exist.}, one can show  that except with small probability, $\siml$ can never make two queries  $\vmes = (m_1, \cdots, m_{i-1}, m_i)$ and $\vmes' = (m_1, \cdots, m_{i-1}, m'_i)$ whose $m_i \ne m'_i$ and $V^*$ does not abort on both $\vmes$ and $\vmes'$. In other words, during the whole execution of $\siml$, it never gets to see two different continuations of the same transcript.   Therefore, the execution of $\siml$ with oracle access to $V^*$ can be roughly simulated by an algorithm with only {interaction} to $V^*$. Further notice that the interaction with $V^*$ is simply an interaction with $V$ plus random aborting,  the execution of $\siml$ can be therefore simulated by an algorithm (cheating prover) $P^*$ with interaction with $V$ (instead of $V^*$).  If the simulator outputs an accepting transcript, then the interaction between $P^*$ and $V$ also outputs  an accepting transcript. Since the statement $x$ is not in $L$, by the soundness of the proof system $\Pi$, any (efficient) prover $P^*$ can not convince $V$. We then conclude that $\B$ on $x\not\in L$ never accepts.

\vspace{1em}

We notice that the first half of the proof (for $x \in L$) relies only on the zero-knowledge property of $\Pi$. This part can be generalized to the quantum setting. The barrier of quantizing the proof is from the second part (for $x \not \in L$). Recall that in the second part of the proof, we need to argue that $\siml$ never sees two different continuations of the same transcript. However, for a quantum simulator, even a single quantum query would completely reveal answers of possibly all transcripts. We resolve the problem in the next section.

\paragraph{Impossibility for Strictly Polynomial-Time Simulator.} We first extend the classical result to the quantum setting, showing that all languages that have constant round post-quantum zero-knowledge proofs or arguments with strict polynomial-time BB simulators must be trivial, i.e. in $\mathbf{BQP}$. 

Let $\siml$ be the quantum strict polynomial-time BB simulator. 
Roughly speaking\footnote{To formally define it, we follow the definition in \cite{EC:Unruh12}, see \Cref{sec:prelim_interactive_proof}.}, the simulator $\siml$ uses  a (dishonest) $V^*(x, \aux, \cdot)$ as a quantum black-box routine. It will be more clear when we define $V^*$ below. 
Similar to the classical proof, we construct a random aborting verifier $V^*$ based on the honest $V$ and we show that there is an efficient quantum algorithm that makes use of the simulator and the random aborting verifier and decides language $L$. As mentioned in the previous section, although the idea follows from \cite{STOC:BarLin02}, we show barriers for lifting the proof and how we overcome them. 

\vspace{1em}

\emph{Random Aborting Verifier.}  A random aborting verifier $V^*$ is similar to that defined in the classical proof, except the aborting probability $\epsilon_1 = \cdots = \epsilon_k = \epsilon$ are the same.
Let $H$ be a random oracle that independently on each input transcript $\vmes$ of length $i \leq k$, outputs $0$ (aborting) with probability $1-\epsilon$ and $1$ (non-aborting) with some non-negligible probability $\epsilon$. 
Thus, a quantum query made by the simulator $\siml$ will become
\begin{align*}
    & \ket {\vmes, 0} \to  \ket {\vmes,  V(x, r, \vmes)}   &  H(\vmes) = 1; \\
    & \ket {\vmes, 0} \to  \ket {\vmes,  0}   &  \text{ otherwise.}
\end{align*}
In the above notation, $V(x, r, \cdot)$ is the next-message function of $V$ corresponding to the statement $x$ and its random tape $r$. Importantly, the quantum oracle access to $V^*(x, r, \cdot)$ can be  simulated by constant number of quantum oracle access to $V(x, r, \cdot)$ and $H$. Later in the construction of our algorithm, $\siml$ chooses a random tape for $V$, we assume $\siml$ has oracle access to  $H$ and can compute $V(x, r, \cdot)$ by itself. 

\vspace{1em}

\emph{First Attempt.}   A natural approach is to consider the following algorithm $\B$: on input $x$, it samples a random tape $r$, a random oracle $H$ and runs the simulator $\siml$ on input $x$ with oracle access to $H$, outputs $1$ if the transcript produced by $\siml$ is an accepting transcript and compatible with $H$. Here ``compatible'' means that conditioned on this transcript, it never aborts. 

For $x \in L$, $\siml$ would output an accepting  transcript  with probability roughly $\epsilon^k$, which is the same probability as that the random aborting verifier accepts a proof. Because $\epsilon$ will be chosen as an inverse polynomial and $k$ is a constant, $\B$ on $x\in L$ accepts with non-negligible probability. 

Although this algorithm works on input $x \in L$, there is an issue for $x\not\in L$.  
As briefly mentioned at the end of the last section, the idea underlines the classical proof is: a strict polynomial-time  BB simulator will not have ``enough time'' to obtain two non-abort  responses from the verifier; thus one would use any execution of $\siml$ that outputs $1$ to convince an honest verifier $V$. 
Such a claim is not trivial in the quantum setting as a single quantum query to $V^*$, even if $V^*$ aborts with very high but still non-negligible probability, reveals exponentially many non-abort responses. 
Thus, we can not conclude that the algorithm does not accept $x \not \in L$.

\vspace{1em}

\emph{Measure-and-Reprogram.} 
A naive solution would be to measure all quantum queries made by the simulator. As long as all queries become classical, we can resolve the issue. Though, this approach works for $x \not \in L$, the modified algorithm may never accept $x \in L$, as measuring all the quantum queries can be easily identified. Our idea is to apply a refined way of measuring and extracting quantum queries -- the ``measure-and-reprogram'' technique that was first introduced for proving the post-quantum security of Fiat-Shamir \cite{C:DFMS19, C:DonFehMaj20}. Very informally, by applying the technique, we obtain the following oracle algorithm $\widetilde{\siml}^H$: 
\begin{itemize}
    \item It picks $k$ queries out of all $q$ queries made by $\siml$, which will be measured later and runs $\siml$ as a subroutine.  
    \item Every time $\siml$ makes a query that is supposed to be measured, $\widetilde{\siml}$ measures the query and reprograms the oracle $H$ on the measured point in a ``certain'' way. 
\end{itemize}
Intuitively, this allows us to exactly measure the transcript that will be outputted by $\siml$ at the end of the execution while still preserving its success probability. 
Followed by the ``measure-and-reprogram'' lemma\footnote{We rely on a variant by \cite{YZ20}, also see  \cite{C:DFMS19, C:DonFehMaj20}.}, 
the probability that the output transcript of $\widetilde{\siml}$ gets measured during  ``measure-and-reprogram'' and it is an accepting transcript  is non-negligible. 

We can then define a new algorithm $\B'$ based on $\widetilde{\siml}$: on input $x$, it samples a random tape $r$, a random oracle $H$ and runs the algorithm $\widetilde{\siml}$ on input $x$ with oracle access to $H$, outputs $1$ if the output transcript is an accepting transcript and compatible with the updated $H$. Here $H$ gets updated in the measure-and-reprogram process.  This algorithm only partially solves the issue. Recall our goal is to show the algorithm can be turn into a cheating prover, thus all its queries can be simulated by only having interaction with a honest verifier $V$. 
It now makes $k$ out of $q$ queries classical, and these $k$ queries are exactly what will be the output transcript. These $k$ queries can then be simulated by the interaction with $V$.  However, for the other $q-k$ queries, they are still quantum queries and may be hard to answer if it only sees an interaction with $V$.

\vspace{1em}

\emph{Finish the Proof.} Actually, it turns out that the other $q-k$ queries can be easily answered, by preparing an empty oracle $H_0$ (which outputs $0$ on every input) instead of a real random oracle $H$. Thus, the algorithm for deciding $L$ is the following: 
\begin{description}
    \item $\widetilde{\B}$: On input $x$, it samples a random tape $r$, an empty oracle $H_0$ and runs the algorithm $\widetilde{\siml}$ on input $x$ with oracle access to $H_0$, outputs $1$ if the output transcript is an accepting transcript and compatible with the updated $H_0$.
\end{description}

The only difference between $\widetilde{\B}$ and $\B'$ is that the underlying $\widetilde{\siml}$ gets either an empty oracle $H_0$ or a sparse oracle $H$ (each output is $1$ with probability $\epsilon$). By {\cite[Lemma 3]{ePrint:HRS}}, the advantage of distinguishing an empty oracle from a sparse oracle by a $q$-quantum-query algorithm is at most $8 q^2 \epsilon$.  Therefore, for any $x\in L$, $\widetilde{\B}$ outputs $1$ with probability at least that of $\B'$ outputs $1$ minus $8 q^2 \epsilon$. By carefully tuning $\epsilon$, we can show that $\widetilde{\B}$ still accepts $x \in L$ with non-negligible probability.

For $x \not \in L$, we want to turn $\widetilde{\B}$ into a cheating prover $P^*$. Because $P^*$ can never convince $V$ on $x \not \in L$, $\widetilde{\B}$ should never accept $x$ unless with negligible probability. Assuming the first quantum query made by $\widetilde{\siml}$ is not going to be measured. In this case, $H_0$ does not get updated and $\widetilde{\siml}$ makes the first quantum query to $V^*(x, r, H_0, \cdot)$. By the definition of the random aborting verifier, it always aborts on any input. Therefore, it can answer the first quantum query by simply always returning $0$, without getting any response from the real verifier $V(x, r,\cdot)$.
If the first quantum query needs to be measured, it will be part of the final output transcript. $P^*$ can obtain and record the response by doing the interaction with $V(x, r, \cdot)$.

Similarly, whenever $\widetilde{\siml}$ makes a quantum query to $V^*(x, r, H_0, \cdot)$ (where $H_0$ is the updated oracle so far), the only non-abort responses come from the input $\vmes$ such that all its prefix $\vmes_j$ satisfying $H_0(\vmes_j) = 1$. Because $H_0$ is initialized as an empty oracle, every input $\vmes$ satisfying $H_0(\vmes) = 1$  must be measured and reprogrammed at certain point in the execution of $\widetilde{\siml}$, its response $V(x, r, \vmes)$ is already known and recorded. For $\vmes$ such that $H(\vmes) = 0$, we do not need to know its response. Overall, $P^*$ can simulate any quantum query in the execution of $\widetilde{\siml}^{H_0}$.   Therefore, $\widetilde{\B}$ never accepts $x \not\in L$ except with negligible probability. 

Thus, post-quantum constant-round zero-knowledge proofs or arguments with strict polynomial-time BB simulators for all languages in $\mathbf{NP}$ do not exist unless $\mathbf{NP} \subseteq \mathbf{BQP}$.

\vspace{1em}

\paragraph{Impossibility for Expected Polynomial-Time Simulator.}

As discussed by Barak and Lindell \cite{STOC:BarLin02}, a natural attempt to extend the classical impossibility proof 
to expected-time BB simulators is by truncating the execution of a simulator (see Section 1.3 of \cite{STOC:BarLin02}); they pointed out that such an attempt would fail. Imagine a expected polynomial-time BB simulator (with expected running time $q/2$) has oracle access to an aborting verifier $V^*$ with a very small \revise{non-aborting} probability, 
say $\epsilon = q^{-10}$. As long as we truncate the execution of the simulator  when the running time is significantly smaller than $q^{10}$, it would never get any non-abort response from the aborting verifier and is not able to produce any accepting transcript.

Another way to interpret the above argument is: 
by the impossibility of strict-poly BB simulation, we can say that a simulator has to ``learn" the aborting probability. 
If a BB simulator \revise{is only allowed} to make a bounded number of queries (which is independent of the aborting probability $\epsilon$), it almost can never learn $\epsilon$, as long as $1/\epsilon$ is significantly larger. 

We show that, informally, if there is an expected polynomial-time BB simulator, then by truncating this simulator, it is still a ``good-enough'' simulator for a specific aborting verifier while it does not measure/learn the aborting probability.  We know that this can not happen for all languages in $\mathbf{NP}$ unless $\mathbf{NP} \subseteq \mathbf{BQP}$.

\revise{
A crucial difference between quantum and classical malicious verifiers is that the quantum malicious verifier can use an auxiliary input qubit to control the aborting probability in ``superposition''. This implies that this auxiliary qubit can somehow ``entangle with the runtime'' of the protocol or the simulation. Therefore, conditioned on accepting, the state of this auxiliary qubit after the interaction between the real prover and the verifier can be far from the state after the simulation since the black-box simulator requires to ``measure'' the aborting probability (and thus measures the control bit).
By combining this observation and our impossibility result on the strict polynomial-time simulation, we can also fail the expected polynomial-time simulation.
To be more specific,} consider the following verifier $\widetilde{V}^*$ that runs a honest verifier $V$ and a random aborting verifier $V^*$ (with non-aborting probability $\epsilon$) in superposition: 
\begin{itemize}
    \item It prepares a control bit $\ket \psi = \frac{1}{\sqrt{2}} (\ket 0 + \ket 1)$ at the beginning. 
    \item It runs $V$ and $V^*$ in superposition: on input $\vmes$, if the control bit is $0$, it never aborts and behaves as $V$; otherwise the control bit is $1$, it aborts with probability $1-\epsilon$  as $V^*$, by querying an internal random oracle. 
    \item Finally, it outputs a classical bit $b$ indicating whether it accepts and a single qubit in the control bit register.  
\end{itemize}
If the control bit is $0$, we know that $V$ always accepts by perfect completeness of $\Pi$. If the control bit is $1$, it accepts with $\epsilon^k$. Thus, if $\widetilde{V}^*$ accepts (the classical output $b = 1$), the output qubit is proportional to $\ket 0 + \sqrt{\epsilon^k} \ket 1$\footnote{In the real execution, $\widetilde{V}^*$ would be entangled with its internal random oracle $H$ and make the final qubit a mixed state. Nonetheless, we show such entanglement can be uncomputed by $\widetilde{V}^*$ and the final qubit is a pure state. For simplicity, we omit the details here.\label{footnote:uncompute}}.  

Let $\siml$ be the expected polynomial-time BB simulator that makes $q/2$ queries in expectation. 
Consider a truncated simulator $\siml_{\sf trunc}$ for $\widetilde{V}^*$ which halts after $\siml$ tries to make the $(q+1)$-th query. By Markov inequality, we know the probability that $\siml$ will halt within the first $q$ queries is at least $1/2$. 
When $\siml$ makes at most $q$ queries, it outputs $b = 1$ with probability at least $1/2$ (because when the control bit is $0$, $V$ always accepts). It is worth noting that this two events are independent. Thus, $\siml_{\sf trunc}$ would output $b = 1$ with probability at least $1/4$. As we know when $b = 1$, $\widetilde{V}^*$ always outputs the qubit $\ket 0 + \sqrt{\epsilon^k} \ket 1$. By zero-knowledge property, $\siml_{\sf trunc}$ should also output \revise{a state close to} $\ket 0 + \sqrt{\epsilon^k} \ket 1$ when the classical output $b = 1$. However, this can not happen for  languages outside $\mathbf{BQP}$. \revise{By our  impossibility result for strict poly-time simulators,} 
such a bounded query BB simulator would essentially need to measure the aborting probability of $\widetilde{V}^*$, which necessarily collapse the control qubit (as the aborting probability for control bit $0$ is $0$, and for control bit $1$ is $\epsilon$). Therefore, we conclude that
post-quantum constant-round zero-knowledge proofs or arguments with (expected) polynomial-time BB simulators for all languages in $\mathbf{NP}$ do not exist unless $\mathbf{NP} \subseteq \mathbf{BQP}$. 
\takashi{The way of using ``first" and ``second" results may be confusing here because our first result may be understood as the impossibility of expected poly simulation for constant round protocols when one reads Our Result subsection.} \qipeng{Fixed}


\paragraph{On the Efficiency of Malicious Verifier.}
In the explanation so far, we considered a malicious verifier that relies on a random oracle. 
For making the verifier efficient, a standard technique is to simulate a random oracle by using a $2q$-wise independent function when the number of queries is at most $q$  \cite{C:Zhandry12}.
Though we show that this works in our setting, it is not as trivial as one would expect due to some technical reasons.\footnote{The reason is related to that we have to uncompute the entanglement between $H$ and $\widetilde{V}^*$'s final qubit as explained in \cref{footnote:uncompute}} 
Therefore, in the main body, we first consider an inefficient malicious verifier that simulates the random oracle by a completely random function, and then we explain how we make the verifier be efficient without affecting the proof.

\subsubsection{Impossibility of  Constant-Round Public-Coin or Three-Round $\epsilon$-ZK}
In the classical setting, Goldreich and Krawczyk  \cite{SIAM:GK96} proved the impossibility of constant-round public-coin or three-round ZK arguments. 
\revise{It is easy to see that their result also rules out $\epsilon$-ZK arguments by essentially the same proof. }
\qipeng{what does it mean? rules out quantum $\epsilon$-ZK?}
\takashi{fixed}
Roughly speaking, we translate their proof into the quantum setting by again relying on the measure-and-reprogram technique \cite{C:DFMS19,C:DonFehMaj20}.  
We give more details of each case below.

\paragraph{Constant-Round Public-Coin Case.}
For a constant-round public-coin protocol $\Pi=(P,V)$ for an $\NP$ language $L$, we consider a malicious verifier $V^*$ that derives its messages by applying a random oracle on the current transcript. 
From the view of the honest prover, $V^*$ is perfectly indistinguishable from the honest verifier $V$.
Thus, if $V^*$  interacts with the honest prover given on common input $x\in L$ and prover's private input $w\in R_L(x)$, it always accepts by the completeness of the protocol. 
Let $\siml$ be a simulator for the $\epsilon$-ZK property. 
By the above observation, when $\siml$ is given oracle access to $V^*$ on input $x\in L$, it should let $V^*$ accept with probability at least $0.9$ since otherwise $V^*$ may notice the difference with a constant advantage, which violates the $\epsilon$-ZK property.\footnote{The choice of the constant is arbitrary.} 
On the other hand, we observe that an accepting transcript between $V^*$ is essentially an accepting proof for the non-interactive argument $\Pi_{\mathsf{ni}}$ obtained by applying Fiat-Shamir transform to the protocol $\Pi$ since the way of deriving the verifier's messages is the same as that in the Fiat-Shamir transform. 
Noting that $V^*$ can be simulated given oracle access to the random oracle, if $\siml$ lets $V^*$ accept on some $x\notin L$ with non-negligible probability, such a simulator can be directly translated into an adversary that breaks the soundness of $\Pi_{\mathsf{ni}}$ in the quantum random oracle model. 
On the other hand, it is shown in \cite{C:DonFehMaj20} that Fiat-Shamir transform preserves soundness up to polynomial security loss for constant-round public-coin protocols, and thus $\Pi_{\mathsf{ni}}$ has negligible soundness error. 
This means that $\siml$ lets $V^*$ accept with negligible probability for any $x\notin L$.
Combining the above, we can decide if $x\in L$ by simulating an interaction between $\siml$ and $V^*$ and then seeing if $V^*$ accepts finally. 
This means $L\in \BQP$.
Therefore, such a protocol for all $\NP$ does not exist unless $\NP\subseteq \BQP$. 

\paragraph{Three-Round Case.}
This case is similar to the constant-round public-coin case except that we apply a random oracle to obtain verifier's private randomness rather than a verifier's message itself.  
Due to this difference, we cannot directly relate the $x\notin L$ case to the soundness of Fiat-Shamir, and we need more careful analysis. 

For a three-round protocol $\Pi=(P,V)$ for an $\NP$ language $L$, we consider a malicious verifier $V^*$ that derives its \emph{private randomness} by applying a random oracle on \emph{prover's first message}. 
From the view of the honest prover, $V^*$ is perfectly indistinguishable from the honest verifier $V$.
Thus, if $V^*$  interacts with the honest prover given on common input $x\in L$ and prover's private input $w\in R_L(x)$, it always accepts by the completeness of the protocol. 
Let $\siml$ be a simulator for the $\epsilon$-ZK property. 
By the above observation, when $\siml$ is given oracle access to $V^*$ on input $x\in L$, it should let $V^*$ accept with probability at least $0.9$ similarly to the constant-round public-coin case. 
However, unlike the constant-round public-coin case above, we cannot directly say that $\siml$  let $V^*$ accept with negligible probability on input $x\notin L$ because $V^*$ derives the private randomness by the random oracle, which is different from the Fiat-Shamir transform. Therefore we need additional ideas. 

$\siml$ can be seen as an algorithm that makes quantum queries to the next-message-generation function $F_{\mathsf{next}}$, which outputs $V^*$'s second-round message taking prover's first-round message on input, and output-decision function $F_{\mathsf{out}}$, which decides if $V^*$ accepts taking a transcript as input. (Note these functions depend on the random oracle.)
Finally, it outputs an accepting transcript with probability at least $0.9$. 
First, we claim that we can assume that $\siml$ only has the oracle  $F_{\mathsf{next}}$ and does not make any query to $F_{\mathsf{out}}$ if we admit a polynomial security loss. Intuitively, this is because if it makes a query on which $F_{\mathsf{out}}$ returns ``accept", then it could have used this query as its final output. 
Though this is trivial in the classical setting, it is not in the quantum setting.
Fortunately, we can prove this by relying on the one-way to hiding lemma \cite{JACM:Unruh15,C:AmbHamUnr19} in the quantum setting as well.
Thus, we think of  $\siml$ as an algorithm that makes quantum queries to $F_{\mathsf{next}}$ and outputs an accepting transcript with probability at least $\frac{1}{\poly(\secpar)}$ when $x\in L$. 

Next, we apply the measure-and-reprogram lemma of \cite{C:DonFehMaj20} to $\siml^{F_{\mathsf{next}}}$.
That is, we consider an experiment $\mathsf{Exp}_{\mathsf{MaR}}(x)$ which roughly works as follows: The experiment simulates $\siml^{F_{\mathsf{next}}}(x)$ except that a randomly chosen $\siml$'s query is measured (let $m_P$ be the outcome), its response is replaced with a freshly sampled message $\mes_V$ independently of $F_{\mathsf{next}}$, and the oracle is updated to be consistent to this response thereafter.
Finally, the experiment outputs the transcript output by $\siml$. 

By using the measure-and-reprogram lemma, the probability (which we denote by $p(x)$ in the following) that $\mathsf{Exp}_{\mathsf{MaR}}(x)$  outputs an accepting transcript whose first and second messages match $m_P$ and $m_V$ 
is $\frac{1}{\poly(\secpar)}$ times the probability that $\siml^{F_{\mathsf{next}}}(x)$ outputs an accepting transcript.
Thus, $p(x)$ is at least  $\frac{1}{\poly(\secpar)}$  for all $x\in L$. 

On the other hand, we can prove that $p(x)$ is negligible for all $x \notin L$ by using soundness of the protocol $\Pi$.
Indeed, we can construct a cheating prover $P^*$ that simulates $\mathcal{S}^{F_{\mathsf{next}}}(x)$ where $F_{\mathsf{next}}$ is simulated according to a random oracle chosen by $P^*$, sends $m_P$ to the external verifier as the first message, embeds verifier's response as $m_V$,  
and sends the third message derived from the output of $\mathcal{S}^{F_{\mathsf{next}}}(x)$ to the external verifier. 
It is easy to see that $P^*$ perfectly simulates  the environment of $\mathsf{Exp}_{\mathsf{MaR}}(x)$  for  $\mathcal{S}^{F_{\mathsf{next}}}(x)$  and thus the probability that the verifier accepts is at least $p(x)$. 
Therefore, by the assumed soundness, $p(x)$ is  negligible. 

By combining above, we can decide if $x\in L$ by simulating $\mathsf{Exp}_{\mathsf{MaR}}(x)$ and then seeing if the output is an accepting transcript whose first and second messages match $m_P$ and $m_V$. (Note that we can efficiently check if the transcript is accepting if we sample $m_V$ by ourselves so that we know the corresponding verifier's private randomness). 
This means $L\in \BQP$.
Therefore, such a protocol for all $\NP$ does not exist unless $\NP\subseteq \BQP$.

\subsection{More Related Work}
Jain et al. \cite{JKMR09} showed that there does not exist constant-round public-coin or three-round post-quantum BBZK \emph{proofs} for $\NP$ unless $\BQP\subseteq \NP$. 
Indeed, they showed that this holds even if the last message in the protocol can be quantum. We believe that our impossibility results can also be extended to this setting, but we focused on classical protocols in the context of post-quantum security for simplicity.
If we focus on classical protocols, our impossibility results are stronger than theirs as we also rule out BB \emph{$\epsilon$-ZK arguments}. 
To the best of our knowledge, the work of \cite{JKMR09} is the only known result on the impossibility of quantum BBZK.  

We review additional related works on lower bounds of ZK protocols in the classical setting. 
Katz \cite{TCC:Katz08a} proved  that  there does not exist four-round BBZK \emph{proofs} for $\NP$ unless $\NP\subseteq \mathbf{coMA}$.
It is interesting to study if we can extend this to rule out four-round post-quantum BB $\epsilon$-ZK proofs for $\NP$ under a reasonable complexity assumption.
We note that there exists  four-round (classical) BBZK \emph{arguments} for $\NP$ under the existence of one-way functions \cite{EC:BelJakYun97}.
It is also interesting to study if we can extend their construction to construct four-round post-quantum BB $\epsilon$-ZK for $\NP$.  (Note that it is necessary  to relax ZK property by \cref{thm:impossibility_main}.)

Kalai, Rothblum, and Rothblum \cite{C:KalRotRot17} proved that there does not exist  constant-round public-coin ZK \emph{proofs} for $\NP$ even with \emph{\revise{non-BB} simulation} under certain  assumptions on obfuscation. 
Fleischhacker, Goyal, and Jain \cite{EC:FleGoyJai18} proved that there does not exist three-round ZK \emph{proofs} for $\NP$ even with \emph{\revise{non-BB} simulation} under the same assumptions.
Though these results are shown in the classical setting,  it might be possible to extend them to the quantum setting by assuming similar assumptions against quantum adversaries.
However, that would result in impossibility for \emph{proofs} whereas our impossibility covers \emph{arguments} though limited to BB simulation.

 \section{Preliminaries}
\paragraph{Basic Notations.}
We denote by $\secpar$ the security parameter throughout the paper.
For a positive integer $n\in\mathbb{N}$, $[n]$ denotes a set $\{1,2,...,n\}$.
For a finite set $\calX$, $x\sample \calX$ means that $x$ is uniformly chosen from $\calX$.
For a finite set $\calX$ and a positive integer $k$, $\calX^{\leq k}$ is defined to be $\bigcup_{i\in [k]}\calX^{i}$. 
For finite sets $\calX$ and $\calY$, $\func(\calX,\calY)$ denotes the set of all functions with domain $\calX$ and range $\calY$.

A function $f:\mathbb{N}\ra [0,1]$ is said to be \emph{negligible} if for all polynomial $p$ and sufficiently large $\secpar \in \mathbb{N}$, we have $f(\secpar)< 1/p(\secpar)$; it is said to be \emph{overwhelming} if $1-f$ is negligible, and said to be \emph{noticeable} if there is a polynomial $p$ such that $f(\secpar)\geq  1/p(\secpar)$ for sufficiently large $\secpar\in \mathbb{N}$.
We denote by $\poly$ an unspecified polynomial and by $\negl$ an unspecified negligible function.

We use PPT and QPT to mean (classical) probabilistic polynomial time and quantum polynomial time, respectively.
For a classical probabilistic or quantum algorithm $\A$, $y\sample \A(x)$ means that $\A$ is run on input $x$ and outputs $y$.
When $\A$ is a classical probabilistic algorithm, we denote by $\A(x;r)$ the execution of $\A$ on input $x$ and randomness $r$.
When $\A$ is a quantum algorithm that takes a quantum advice, we denote by $\A(x;\rho)$ the execution of $\A$ on input $x$ and an advice $\rho$.

We use the bold font (like $\regX$) to denote quantum registers, and $\hil_\regX$ to mean the Hilbert space corresponding to the register $\regX$. 
For a quantum state $\rho$, $M_{\regX}\circ \rho$ means a measurement in the computational basis on the register $\regX$ of $\rho$.
For quantum states $\rho$ and $\rho'$, $\TD(\rho,\rho')$ denotes trace distance between them.
\revise{
We say that $\rho$ is negligibly close to $\rho'$ if $\TD(\rho,\rho')=\negl(\secpar)$.
}

\paragraph{Standard Computational Models.} 
\begin{itemize}
\item A PPT algorithm is a probabilistic polynomial time (classical) Turing machine.
A PPT algorithm is also often seen as a sequence of uniform polynomial-size circuits.
\item A QPT algorithm is a polynomial time quantum Turing machine. 
A QPT algorithm is also often seen as a sequence of uniform polynomial-size quantum circuits.
\item 
An adversary (or malicious party) is modeled as a non-uniform QPT algorithm $\A$ (with quantum advice) that is specified by sequences of polynomial-size quantum circuits $\{\A_\secpar\}_{\secpar\in\mathbb{N}}$ and polynomial-size quantum advice $\{\rho_\secpar\}_{\secpar\in \mathbb{N}}$.
When $\A$ takes an input of $\secpar$-bit, $\A$ runs $\A_{\secpar}$ taking $\rho_\secpar$ as an advice. 

\end{itemize}

\paragraph{Indistinguishability of Quantum States.}
We define computational and statistical indistinguishability of quantum states similarly to \cite{STOC:BitShm20}.

We may consider random variables over bit strings or over quantum states. 
This will be clear from the context. 
For ensembles of random variables $\mathcal{X}=\{X_i\}_{\secpar\in \mathbb{N},i\in I_\secpar}$ and $\mathcal{Y}=\{Y_i\}_{\secpar\in \mathbb{N},i\in I_\secpar}$  over the same set of indices $I=\bigcup_{\secpar\in\mathbb{N}}I_\secpar$ and a function $\delta$,       
we write $\mathcal{X}\compind_{\delta}\mathcal{Y}$ to mean that for any non-uniform QPT algorithm $\A=\{\A_\secpar,\rho_\secpar\}$, there exists a negligible function $\negl$ such that for all $\secpar\in\mathbb{N}$, $i\in I_\secpar$, we have
\[
|\Pr[\A_\secpar(X_i;\rho_\secpar)]-\Pr[\A_\secpar(Y_i;\rho_\secpar)]|\leq \delta(\secpar) + \negl(\secpar).
\]
Especially, when we have the above for $\delta=0$, we say that $\calX$ and $\calY$ are computationally indistinguishable, and simply write $\calX\compind \calY$.

Similarly, we write $\calX\statind_{\delta}\calY$ to mean that for any unbounded time  algorithm $\A$, there exists a negligible function $\negl$ such that for all $\secpar\in\mathbb{N}$, $i\in I_\secpar$, we have 
\[
|\Pr[\A(X_i)]-\Pr[\A(Y_i)]|\leq \delta(\secpar) + \negl(\secpar).\footnote{In other words, $\calX\statind_{\delta}\calY$ means that there exists a negligible function $\negl$ such that the trace distance between $\rho_{X_i}$ and $\rho_{Y_i}$ is at most $\delta(\secpar) + \negl(\secpar)$ for all $\secpar\in \mathbb{N}$ and $i\in I_\secpar$ where $\rho_{X_i}$ and $\rho_{Y_i}$ denote density matrices corresponding to $X_{i}$ and $Y_{i}$.}
\]
Especially, when we have the above for $\delta=0$, we say that $\calX$ and $\calY$ are statistically indistinguishable, and simply write $\calX\statind \calY$.
Moreover, 
we write $\calX \equiv \calY$ to mean
that $X_i$ and $Y_i$ are distributed identically for all $i\in I$

\subsection{Interactive Proof and Argument.} \label{sec:prelim_interactive_proof}
We define interactive proofs and arguments similarly to \cite{STOC:BitShm20,CCY20}. 
\paragraph{Notations.}
For an $\NP$ language $\lang$ and $x\in \lang$, $\rel_{\lang}(x)$ is the set that consists of all (classical) witnesses $w$ such that the verification machine for $L$ accepts $(x,w)$.

A classical interactive protocol is modeled as an interaction between interactive classical polynomial-time machines $\pro$ referred to as a prover and $\ver$ referred to as a verifier. 
We denote by $\execution{\pro(x_{\pro})}{\ver(x_{\ver})}(x)$ an execution of the protocol where $x$ is a common input, $x_\pro$ is $\pro$'s private input, and $x_\ver$ is $\ver$'s private input.
We denote by $\OUT_\ver\execution{\pro(x_{\pro})}{\ver(x_{\ver})}(x)$ the final output of $\ver$ in the execution. 
An honest verifier's output is $\top$ indicating acceptance or $\bot$ indicating rejection.
We say that the protocol is public-coin if the honest verifier $V$ does not use any private randomness, i.e., each message sent from $V$ is a uniform string of a certain length and $V$'s final output is derived by applying an efficiently computable classical function on the transcript.

\begin{definition}[Interactive Proof and Argument for $\NP$]
A classical interactive proof or argument $\Pi$ for an $\NP$ language $\lang$ is an interactive protocol between a PPT prover $\pro$ and a PPT verifier $\ver$ that satisfies the following: 
\paragraph{Completeness.}
For any $x\in L$,  and $w\in R_L(x)$, we have 
\begin{align*}
    \Pr[\OUT_\ver\execution{\pro(w)}{\ver}(x)=\top]\geq \revise{1-\negl(\secpar)}
\end{align*}
\paragraph{Statistical/Computational Soundness.}
We say that an interactive protocol is statistically (resp. computationally) sound if for any unbounded-time (resp. non-uniform QPT) cheating prover $\pro^*$, there exists a negligible function $\negl$ such that for any $\secpar \in \mathbb{N}$ and any $x\in \bit^\secpar\setminus \lang $, we have  
\begin{align*}
    \Pr[\OUT_\ver\execution{\pro^*}{\ver}(x)=\top]\leq \negl(\secpar).
\end{align*}
We call an interactive protocol with statistical (resp. computational) soundness an interactive proof (resp. argument).
\end{definition}


\paragraph{Malicious verifier and black-box simulator.}
For a formal definition of black-box quantum zero-knowledge, we give a model of quantum malicious verifiers against classical interactive protocols. 
A malicious verifier $V^*$ is specified by a sequence of unitary $U^*_\secpar$ over the internal register $\regV_\secpar$ and the message register $\regM_\secpar$ (whose details are explained later) 
and an auxiliary input $\rho_\secpar$ indexed by the security parameter $\secpar\in \mathbb{N}$. 
We say that $V^*$ is non-uniform QPT if the sizes of $U^*_\secpar$ and $\rho_\secpar$  are polynomial in $\secpar$. 
In the rest of this paper, 
$\secpar$ is always set to be the length of the statement $x$ to be proven, and thus we omit $\secpar$ for notational simplicity. 

Its internal register $\regV$ consists of 
the statement register $\regX$,  
auxiliary input register $\regaux$, 
and verifier's working register $\regW$, and
part of $\regV$ 
is designated  as the output register $\regout$.
$V^*$ interacts with an honest prover $P$ of a protocol $\Pi$ on a common input $x$ and $P$'s private input $w\in R_L(x)$ in the following manner:
\begin{enumerate}
    \item $\regX$ is initialized to $x$,  $\regaux$ is initialized to $\rho$, and $\regW$ and $\regM$ are initialized to be $\ket{0}$.
    \item $P$ (with private input $w$) and $V^*$ run the protocol $\Pi$ as follows:
    \begin{enumerate}
        \item On $V^*$'s turn, it  applies the unitary $U^*$, measures $\regM$, and sends the measurement outcome to $P$. 
        \item On $P$'s turn, when it sends a message to the verifier, 
        $\regM$ is overwritten by the message. Note that this can be done since $\regM$ is measured in the previous $V^*$'s turn.    
    \end{enumerate}
    \item After $P$ sends the final message of $\Pi$, $V^*$ applies $U^*$ and outputs the state in $\regout$, tracing out all other registers. 
\end{enumerate}
We denote by $\execution{\pro(w)}{\ver^*(\rho)}(x)$ the above execution and  by $\OUT_{\ver^*}\execution{\pro(w)}{\ver^*(\rho)}(x)$
the final output of $\ver^*$, which is a quantum state over $\regout$.    

A quantum black-box simulator $\siml$ is modeled as a quantum oracle Turing machine (e.g., see \cite{BBBV}).
We say that $\siml$ is  expected-QPT (resp. strict-QPT) if the expected (resp. maximum) number of steps is polynomial in the input length counting an oracle access as a unit step.   
For an input $x$ and a malicious verifier $V^*$ specified by a unitary $U^*$ and auxiliary input $\rho$,  
$\siml$ works over  the input register $\reginp$, verifier's internal register $\regV$, message register $\regM$, and its working register $\regS$ as follows.
$\reginp$ and the sub-register 
$\regX$ of $\regV$ are initialized to $x$, the sub-register 
$\regaux$ of $\regV$ is initialized to $\rho$, and all other registers ($\regW$, $\regM$, and $\regS$) are initialized to $\ket{0}$. 
$\siml$ is given oracle access to $U^*$ and its inverse ${U^*}^{\dagger}$ and can apply any unitary over $\reginp$, $\regM$, and $\regS$, but it is not allowed to directly act on $\regV$ (except for the invocations of $U^*$ or ${U^*}^{\dagger}$).
\takashi{We may add an explanation on why we allow BB access to the unitary following \cite{EC:Unruh12,TCC:Zhandry20} etc.}
We denote by 
$\siml^{\ver^*(x;\rho)}(x)$ the above execution and 
by
$\OUT_{\ver^*}(\siml^{\ver^*(x;\rho)}(x))$ the output of $\ver^*$, i.e., final state in $\regout$ tracing out all other registers after the execution.

Based on the above formalization, we define post-quantum black-box zero-knowledge proof/argument as follows. 
\begin{definition}[Post-Quantum Black-Box Zero-Knowledge Proof and Argument]\label{def:post_quantum_ZK}
A post-quantum black-box  zero-knowledge proof (resp. argument) for an $\NP$ language $\lang$ is a classical interactive proof (resp. argument) for $\lang$ that satisfies the following property in addition to  completeness and statistical (resp. computational) soundness:
\paragraph{Quantum Black-Box Zero-Knowledge.}
There exists an expected-QPT simulator $\siml$ such that for any non-uniform QPT malicious verifier $\ver^*$ with an auxiliary input $\rho$, we have
\[
\{\OUT_{\ver^*}\execution{\pro(w)}{\ver^*(\rho)}(x)\}_{\secpar,x,w}
\compind
\{\OUT_{\ver^*}(\siml^{\ver^*(x;\rho)}(x))\}_{\secpar,x,w}
\]
where $\secpar \in \mathbb{N}$, $x\in \lang\cap \bit^\secpar$, and $w\in \rel_\lang(\secpar)$. 
\end{definition}

\paragraph{Quantum Black-Box Zero-Knowledge for Inefficient Verifiers.}
In the above definition, we restrict a malicious verifier $V^*$ to be non-uniform QPT.
On the other hand, We say that $\siml$ works for inefficient verifiers if the above holds even for all possibly inefficient malicious verifiers $V^*$.
To the best of our knowledge, all known black-box simulation techniques (in classical or quantum settings) work even for inefficient verifiers.
The reason of defining this notion  is that  we first prove the impossibility of quantum black-box simulation for inefficient verifiers (Theorem \ref{thm:main_impossibility_BB_QZK_inefficient}) as it is simpler than that for efficient verifiers (Theorem \ref{thm:main_impossibility_BB_QZK}).
(Remark that the impossibility of quantum black-box simulation for inefficient verifiers is weaker than the impossibility of quantum black-box simulation for efficient verifiers.)
We stress that we finally extend it to prove the impossibility of  quantum black-box simulation for efficient verifiers (Theorem \ref{thm:main_impossibility_BB_QZK}).

\vspace{1em}

We next define a weaker version of  zero-knowledge called \emph{$\epsilon$-zero-knowledge} following \cite{CCY20}.
\begin{definition}[Post-Quantum Black-Box $\epsilon$-Zero-Knowledge Proof and Argument]\label{def:post_quantum_epsilon_ZK}
A post-quantum black-box  $\epsilon$-zero-knowledge proof (resp. argument) for an $\NP$ language $\lang$ is a classical interactive proof (resp. argument) for $\lang$ that satisfies the following property in addition to completeness and statistical (resp. computational) soundness:
\paragraph{Quantum Black-Box $\epsilon$-Zero-Knowledge.}
For any noticeable $\epsilon$, 
there exists a strict-QPT simulator $\siml$ such that for any non-uniform QPT malicious verifier $\ver^*$ with an auxiliary input $\rho$, we have
\[
\{\OUT_{\ver^*}\execution{\pro(w)}{\ver^*(\rho)}(x)\}_{\secpar,x,w}
\compind_\epsilon 
\{\OUT_{\ver^*}(\siml^{\ver^*(x;\rho)}(x))\}_{\secpar,x,w}
\]
where $\secpar \in \mathbb{N}$, $x\in \lang\cap \bit^\secpar$, and $w\in \rel_\lang(\secpar)$. 
Note that the running time of $\siml$ may depend on $1/\epsilon$.
\end{definition}
\begin{remark}\label{rem:strict_poly_epsilon_ZK}
In the definition of quantum black-box $\epsilon$-zero-knowledge, we assume that $\siml$ runs in strict-QPT rather than expected-QPT without loss of generality. 
Indeed, if we have a expected-QPT simulator, then we can consider a truncated version of it that immediately halts if its running time exceeds the expected running time too much.
It is easy to see that this truncated version of the simulator is still good enough for $\epsilon$-zero-knowledge (unlike for the full-fledged zero-knowledge). 
\revise{A similar observation is also given in \cite{STOC:BarLin02}.}
\end{remark}

\subsection{Useful Lemmas}
The following lemma is heavily used throughout the paper.
\begin{lemma}[\cite{C:Zhandry12}]\label{lem:simulation_QRO}
For any sets $\calX$ and $\calY$ of classical strings and $q$-quantum-query algorithm $\A$, we have
\[
\Pr[\A^{H}=1:H\sample \func(\calX,\calY)]= \Pr[\A^{H}=1:H\sample \mathcal{H}_{2q}]
\]
where  $\mathcal{H}_{2q}$ is a family of $2q$-wise independent hash functions from $\calX$ to $\calY$.
\end{lemma}
The following two lemmas are used in \cref{sec:impossible_constant}. 
\begin{lemma}[{\cite[Lemma 3]{ePrint:HRS}}]\label{lem:ind_sparse_and_zero}
Let $\calX$ be a finite set, $\epsilon\in [0,1]$ be a non-negative real number, and 
$\mathcal{H}_\epsilon$ be a distribution over $H:\calX\rightarrow \bit$ such that we have  $\Pr[H(x)=1]=\epsilon$ independently for each $x\in \calX$.
Let $H_0:\calX\rightarrow \bit$ be the function that returns $0$ for all inputs $x\in \calX$. 
Then for any algorithm $\A$ that makes at most $q$ quantum queries, we have 
\begin{align*}
    \left|\Pr\left[\A^{H}=1:H\sample \mathcal{H}_\epsilon\right]
    -
    \Pr\left[\A^{H_0}=1\right]
    \right|\leq 8q^2\epsilon.
\end{align*}
\end{lemma}

\begin{lemma}[SWAP test]\label{lem:SWAP}
There is a QPT algorithm, called the SWAP test, which satisfies the following: the algorithm takes a product state $\rho\otimes \sigma$ as input,  and accepts with probability $\frac{1+\Tr(\rho\sigma)}{2}$.
Especially, when $\rho$ is a pure state $\ket{\phi}\bra{\phi}$, the probability is $\frac{1+\bra{\phi}\sigma\ket{\phi}}{2}$.
\end{lemma}

We will use a corollary of the one-way to hiding lemma \cite{JACM:Unruh15,C:AmbHamUnr19} in \cref{sec:impossible_three}.
First, we introduce a special case of the one-way to hiding lemma in \cite{C:AmbHamUnr19}. 
\begin{lemma}[{A Special Case of One-Way to Hiding Lemma~\cite{C:AmbHamUnr19}}]
\label{lem:o2h}
Let $S \subseteq \calX$ be random.
Let $z$ be a random bit string.
($S$ and $z$ may have an arbitrary joint distribution.)
Let $F_S:\calX \rightarrow \bit$ be the function defined by 
\begin{align*}
F_S(x):=
\begin{cases}
1 &\text{~if~} x\in S \\
0 &\text{~otherwise~} 
\end{cases}.    
\end{align*}
Let $F_{\emptyset}:\calX \rightarrow \bit$ be the function that outputs $0$ on all inputs.
Let $\A$ be an oracle-aided quantum algorithm that makes at most $q$ quantum queries. 
Let $\B$ be an algorithm that on input $z$ chooses $i\sample [q]$, runs $\A^{F_{\emptyset}}(z)$, measures $\A$'s $i$-th query, and outputs the measurement outcome.
Then we have
\[
\left|\Pr[\A^{F_S}(z)\in S]-\Pr[\A^{F_{\emptyset}}(z)\in S]\right|\leq 2\sqrt{(q+1)\Pr[\B(z)\in S]}.
\]
\end{lemma}
We show a simple corollary of the above lemma, \revise{which roughly says that if we can find an element of $S$ by making polynomial number of quantum queries to $F_S$, then we can find an element of $S$ without making any query to $F_S$ with a polynomial reduction loss.}
\begin{corollary}\label{cor:o2h}
Let $S$, $z$, $F_S$, $F_{\emptyset}$, $\A$, and $\B$ be defined as in  \cref{lem:o2h}.
Let $\mathcal{C}$ be an algorithm (without any oracle) that works as follows: On input $z$, it flips a bit $b\sample \bit$. If $b=0$, then it runs $\A^{F_{\emptyset}}(z)$ by simulating $F_{\emptyset} $ by itself and if $b=1$, then it runs $\B(z)$.
Then we have 
\[
\sqrt{\Pr[\mathcal{C}(z)\in S]}\geq \frac{\Pr[\A^{F_S}(z)\in S]}{4\sqrt{q+1}}.
\]
\end{corollary}
\begin{proof}
By \cref{lem:o2h}, we have 
\[
\left|\Pr[\A^{F_S}(z)\in S]-\Pr[\A^{F_{\emptyset}}(z)\in S]\right|\leq 2\sqrt{(q+1)\Pr[\B(z)\in S]}.
\]
which implies 
\begin{align*}
\Pr[\A^{F_S}(z)\in S]&\leq 2\sqrt{(q+1)\Pr[\B(z)\in S]}+\Pr[\A^{F_{\emptyset}}(z)\in S]\\
&\leq 2\sqrt{(q+1)\Pr[\B(z)\in S]}+2\sqrt{(q+1)\Pr[\A^{F_{\emptyset}}(z)\in S]}\\
&\leq 2\sqrt{2(q+1)(\Pr[\B(z)\in S]+\Pr[\A^{F_{\emptyset}}(z)\in S])}\\
&= 2\sqrt{2(q+1)(2\Pr[\mathcal{C}(z)\in S])}.
\end{align*}
\cref{cor:o2h} immediately follows from the above. 
\end{proof}

\subsection{Measure-and-Reprogram Lemma}
We review the measure-and-reprogram lemma of \cite{C:DonFehMaj20} with notations based on those in  \cite{YZ20}.  
We first give intuitive explanations for these notations, which are taken from \cite{YZ20}.
For a quantumly-accessible classical oracle $\oracle$, we denote by $\ora\leftarrow \reprogram(\ora,x,y)$ to mean that we reprogram $\oracle$ to output $y$ on input $x$.
For a $q$-quantum-query algorithm $\A$, function $H:\calX\rightarrow \calY$, and  $\vecy=(y_1,...,y_k)\in \calY^k$, we denote by $\tilde{\A}[H,\vecy]$  to mean an algorithm that runs $\A$ w.r.t. an oracle that computes $H$ except that randomly chosen $k$ queries are measured and the oracle is reprogrammed to output  $y_i$ on $i$-th measured query.
Formal definitions are given below:

\begin{definition}[Reprogramming Oracle]
Let $\A$ be a quantum algorithm with quantumly-accessible oracle $\ora$ that is initialized to be an oracle that  computes some classical function from $\calX$ to $\calY$. 
At some point in an execution of $\A^{\ora}$, we say that we reprogram $\ora$ to output $y\in \calY$ on $x\in \calX$ if we update the oracle to compute the function $H_{x,y}$ defined by
\begin{eqnarray*}
    H_{x,y}(x'):=
    \begin{cases} y   &\text{if~}  x'=x \\  
     H(x')  &\text{otherwise}
    \end{cases}
\end{eqnarray*}
where $H$ is the function computed by $\ora$ before the update. 
This updated oracle is used in the rest of execution of $\A$.
We denote by $\ora\leftarrow \reprogram(\ora,x,y)$ the above reprogramming procedure. 
\end{definition}

\begin{lemma}(Measure-and-Reprogram Lemma, Rephrasing of \cite[Lemma 4]{C:DonFehMaj20})\label{lem:measure_and_reprogram}
Let $\calX$, $\calY$, and $\calZ$ be sets of classical strings and $k$ be a positive integer. 
Let $\A$ be a $q$-quantum-query algorithm that is given quantum oracle access to an oracle that computes a function from $\calX$ to $\calY$ and a (possibly quantum) input $\inp$ and outputs $\vecx \in \calX^k$ and $z \in \calZ$.
For a function $H:\calX \rightarrow \calY$ and $\vecy=(y_1,...,y_k) \in \calY^k$, we define a measure-and-reprogram algorithm $\tilde{\A}[H,\vecy]$ as follows:


\begin{description}
\item[$\widetilde{A}{[}H,\vecy{]}(\inp)$:] Given a (possibly quantum) input $\inp$, it works as follows:
\begin{enumerate}
    \item 
    For each $i\in[k]$, uniformly pick $(j_i,b_i)\in ([q]\times \bit) \cup \{(\bot,\bot)\}$  conditioned on that there exists at most one $i\in [k]$ such that $j_i=j^*$  for all $j^*\in [q]$. 
    \item Run $\A^{\ora}(\inp)$ 
    where  the oracle $\ora$ is initialized to be a quantumly-accessible classical oracle that computes $H$, and when $\A$ makes its $j$-th query, the oracle is simulated as follows:
    \begin{enumerate}
      \item If $j=j_i$ for some $i\in[k]$, 
        measure $\A$'s  query register to obtain $x'_i$, and do either of the following.
        \begin{enumerate}
        \item If $b_i=0$, reprogram $\ora\leftarrow \reprogram(\ora,x'_i,y_i)$ and answer $\A$'s $j_i$-th query by using the reprogrammed oracle. 
        \item If $b_i=1$, answer  $\A$'s $j_i$-th query by using the oracle before the reprogramming and then reprogram $\ora\leftarrow \reprogram(\ora,x'_i,y_i)$.
        \end{enumerate}
    \item Otherwise, answer $\A$'s $j$-th query by just using the oracle $\ora$ without any measurement or  reprogramming.
    \end{enumerate}
    \item Let $(\vecx=(x_1,...,x_k),z)$ be $\A$'s output.
\item For all $i\in[k]$ such that $j_i=\bot$, set $x'_i:=x_i$. 
    \item Output $\vecx':=((x'_1,...,x'_k),z)$.
\end{enumerate}
\end{description}

Then for any $q$-quantum query algorithm $\A$, $\inp$, $H:\calX \rightarrow \calY$, $\vecx^*=(x^*_1,...,x^*_k) \in \calX^k$ such that $x^*_i\neq x^*_{i'}$ for all $i\neq i'$, $\vecy=(y_1,...,y_k) \in \calY^k$, and a relation $R\subseteq \calX^k \times \calY^k \times \calZ$, we have
\begin{align*}
&\Pr[ \vecx'=\vecx^* \land (\vecx',\vecy,z)\in R:(\vecx',z)\sample \widetilde{A}{[}H,\vecy{]}(\inp)]\\
 &~~~~~~~~\geq \frac{1}{(2q+1)^{2k}}\Pr[\vecx=\vecx^* \land (\vecx,\vecy,z)\in R:(\vecx,z)\sample \A^{H_{\vecx^*,\vecy}}(\inp)].
\end{align*}
where $H_{\vecx^*,\vecy}:\calX\rightarrow \calY$ is defined as 
\begin{eqnarray*}
    H_{\vecx^*,\vecy}(x'):=
    \begin{cases} y_i   &\text{if~} \,\exists i\in[k] \text{~s.t.~} x'=x^*_i \\  
     H(x')  &\text{otherwise}
    \end{cases}.
\end{eqnarray*}
\end{lemma}
\begin{remark}
The above lemma is a rephrasing of \cite[Lemma 4]{C:DonFehMaj20} (taking \cite[Remark 5]{C:DonFehMaj20} into account) given in \cite[Definition 4.5, Lemma 4.6]{YZ20}.\footnote{Note a minor notational difference from \cite{YZ20} that the roles of $i$ and $j$ are swapped.}
\end{remark}

Especially, we will rely on the following special case of \Cref{lem:measure_and_reprogram}.

\begin{lemma}[Measure-and-Reprogram Lemma, Ordered Queries]\label{lem:measure_and_reprogram_ordered}
Let $\calX$and $\calY$ be sets of classical strings and $k$ be a positive integer. 
Let $\A$ be a $q$-quantum-query algorithm that is given quantum oracle access to an oracle that computes a function from $\calX$ to $\calY$ and outputs $\vecx \in \calX^k$.
For a function $H:\calX^{\leq k} \rightarrow \calY$ and  $\vecy=(y_1,...,y_k) \in \calY^k$, 
we define an algorithm $\Aord[H,\vecy]$ as follows:
\begin{description}
\item[$\Aord{[}H,\vecy{]}$:] 
It  works as follows:
\begin{enumerate}
    \item 
    For each $i\in[k]$, uniformly pick $(j_i,b_i)\in ([q]\times \bit) \cup \{(\bot,\bot)\}$  conditioned on that 
    there exists at most one $i\in [k]$ such that $j_i=j^*$  for all $j^*\in [q]$.
    \item Run $\A^{\ora}$ 
    where  the oracle $\ora$ is initialized to be a quantumly-accessible classical oracle that computes $H$, and when $\A$ makes its $j$-th query, the oracle is simulated as follows:
    \begin{enumerate}
      \item If $j=j_i$ for some $i\in[k]$, 
        measure $\A$'s  query register to obtain $\vecx'_i=(x'_{i,1},...,x'_{i,k_i})$ where $k_i\leq k$ is determined by the measurement outcome, and do either of the following.
        \begin{enumerate}
        \item If $b_i=0$, reprogram $\ora\leftarrow \reprogram(\ora,\vecx'_i,y_i)$ and answer $\A$'s $j_i$-th query by using the reprogrammed oracle. 
        \item If $b_i=1$, answer  $\A$'s $j_i$-th query by using the oracle before the reprogramming and then reprogram $\ora\leftarrow \reprogram(\ora,\vecx'_i,y_i)$. 
        \end{enumerate}
    \item Otherwise, answer $\A$'s $j$-th query by just using the oracle $\ora$ without any measurement or  reprogramming. 
    \end{enumerate}
    \item Let $\vecx=(x_1,...,x_k)$ be $\A$'s output.
     \item For all $i \in [k]$ such that $j_i = \bot$, set $\vecx'_i = \vecx_i$ where $\vecx_i := (x_1, \cdots, x_i)$. 
     \item Output $\vecx'_k$ if for all $i \in [k]$, $\vecx'_i$ is a prefix of $\vecx'_k$, and output $\bot$ otherwise. 
     
    
        
\end{enumerate}
\end{description}
Then for any $q$-quantum query algorithm $\A$, 
$H:\calX^{\leq k}\rightarrow \calY$, 
$\vecx^*=(x^*_1,...,x^*_k) \in \calX^k$, and
$\vecy=(y_1,...,y_k)$, 
we have 
\begin{align*}
\Pr[\Aord[H,\vecy]=\vecx^*]\geq \frac{1}{(2q+1)^{2k}}\Pr\left[\A^{\Hord_{\vecx^*,\vecy}}=\vecx^*\right]
\end{align*}
where $\Hord_{\vecx^*,\vecy}: \calX^{\leq k}\rightarrow \calY$ is defined as 
\begin{eqnarray*}
    \Hord_{\vecx^*,\vecy}(\vecx'):=
    \begin{cases} y_i   &\text{if~} \exists i\in[k] \text{~s.t.~} \vecx'=(x^*_1,...,x^*_i) \\  
     H(\vecx')  &\text{otherwise}
    \end{cases}.
\end{eqnarray*}
\end{lemma}
\begin{proof}
We apply \Cref{lem:measure_and_reprogram} for the following setting, where we append $[\ref{lem:measure_and_reprogram}]$ and $[\ref{lem:measure_and_reprogram_ordered}]$ to characters to distinguish those in \Cref{lem:measure_and_reprogram} and \ref{lem:measure_and_reprogram_ordered}, respectively (e.g., $\calX[\ref{lem:measure_and_reprogram}]$ and $\calX[\ref{lem:measure_and_reprogram_ordered}]$ are $\calX$ in \Cref{lem:measure_and_reprogram} and \ref{lem:measure_and_reprogram_ordered}, respectively).
\begin{itemize}
    \item 
    $\calX[\ref{lem:measure_and_reprogram}]:=\calX[\ref{lem:measure_and_reprogram_ordered}]^{\leq k}$, $\calY[\ref{lem:measure_and_reprogram}]:=\calY[\ref{lem:measure_and_reprogram_ordered}]$, $\calZ[\ref{lem:measure_and_reprogram}]:=\emptyset$  
    \item $\inp[\ref{lem:measure_and_reprogram}]$ is a null string and $H[\ref{lem:measure_and_reprogram}]:=H[\ref{lem:measure_and_reprogram_ordered}]$.
    \item 
    $\vecx^*[\ref{lem:measure_and_reprogram}]:=
    (x_1^*[\ref{lem:measure_and_reprogram_ordered}],(x_1^*[\ref{lem:measure_and_reprogram_ordered}],x_2^*[\ref{lem:measure_and_reprogram_ordered}]),...,(x_1^*[\ref{lem:measure_and_reprogram_ordered}],...,x_k^*[\ref{lem:measure_and_reprogram_ordered}]))$.
    \item $\A[\ref{lem:measure_and_reprogram}]$ works similarly to $\A[\ref{lem:measure_and_reprogram_ordered}]$
    except that $\A[\ref{lem:measure_and_reprogram}]$ outputs \[
    (x_1[\ref{lem:measure_and_reprogram}],...,x_k[\ref{lem:measure_and_reprogram}]):=(x_1[\ref{lem:measure_and_reprogram_ordered}],(x_1[\ref{lem:measure_and_reprogram_ordered}],x_2[\ref{lem:measure_and_reprogram_ordered}]),...,(x_1[\ref{lem:measure_and_reprogram_ordered}],...,x_k[\ref{lem:measure_and_reprogram_ordered}]))
    \]
    where $(x_1[\ref{lem:measure_and_reprogram_ordered}],...,x_k[\ref{lem:measure_and_reprogram_ordered}])$ is the output of $\A[\ref{lem:measure_and_reprogram_ordered}]$. 
    \item 
    $R[\ref{lem:measure_and_reprogram}]$ is a trivial relation, i.e., 
    $R[\ref{lem:measure_and_reprogram}]:=\calX[\ref{lem:measure_and_reprogram}]^k\times \calY[\ref{lem:measure_and_reprogram}]^{k}$. 
    
    
    
    
    
\end{itemize}
Then we have $H_{\vecx^*,\vecy}[\ref{lem:measure_and_reprogram}]=\Hord_{\vecx^*,\vecy}[\ref{lem:measure_and_reprogram_ordered}]$ are defined as the same functions in \cref{lem:measure_and_reprogram,lem:measure_and_reprogram_ordered}, and thus the r.h.s. of the inequalities in \Cref{lem:measure_and_reprogram,lem:measure_and_reprogram_ordered}  are the same probability. 
The l.h.s of the inequality in \Cref{lem:measure_and_reprogram} corresponds to the probability that we have $\vecx'_i = (x^*_1,...,x^*_i)$ for all $i\in [k]$ such that $\vecx'_i$ is defined (i.e., $i=k$ or $j_i\neq \bot$) where the probability is taken over randomness of $\Aord[H,\vecy]$. 
Especially, when this event happens, 
we have $\vecx'_k=\vecx^*$ and 
$\vecx'_i$ is a prefix of $\vecx'_k$ for all $i\in [k]$ such that $j_i\neq \bot$, i.e., $\Aord[H,\vecy]=\vecx^*$.  Therefore,  \Cref{lem:measure_and_reprogram} implies \Cref{lem:measure_and_reprogram_ordered}.

\end{proof}
 \section{Impossibility of BB ZK for Constant-Round Arguments}\label{sec:impossible_constant}
In this section, we prove the following theorem.
\begin{theorem}\label{thm:main_impossibility_BB_QZK}
If there exists a constant-round post-quantum black-box  zero-knowledge argument for a language $L$, then $L\in \mathbf{BQP}$.
\end{theorem}

The rest of this section is devoted to prove the above theorem. 
\revise{
Specifically, we prove the theorem by the following steps:
\begin{enumerate}
    \item In \cref{sec:strict_poly}, we prove the impossibility for a strict-polynomial-time simulator that works for inefficient verifiers. (See the paragraph after \cref{def:post_quantum_ZK} for the meaning of that ``a simulator works for inefficient verifiers".)
    This part can be seen as a quantum version of the classical impossibility result of \cite{STOC:BarLin02} (except that we consider inefficient verifiers).
    \item In \cref{sec:impossibility_expected_poly_simulation_inefficient}, we prove the impossibility for a \emph{expected-polynomial-time simulator} that works for inefficient verifiers.
    This is proven by reducing it to the strict-polynomial-time case, which relies on a novel inherently quantum technique.   
    \item In \cref{sec:impossibility_expected_poly_simulation_efficient}, we prove \cref{thm:main_impossibility_BB_QZK}, i.e.,  the impossibility for a expected-polynomial-time simulator that only works for \emph{efficient} verifiers.
    This part is basically done by replacing a random function with  $2q$-wise independent function relying on \cref{lem:simulation_QRO}, where some delicate argument is needed due to technical reasons.
\end{enumerate}
}

First, we define notations that are used throughout this section. 
Let $\Pi=(P,V)$ be a classical constant-round interactive argument for a language $\lang$. 
Without loss of generality, we assume that 
$P$ sends the first message, and let $(P,V)$ be $(2k-1)$-round protocol where $P$ sends $k=O(1)$ messages in the protocol and $V$ sends $(k - 1)$ messages. 

We assume that all  messages sent between  $P$ and $V$ are elements of a classical set $\messpace$ (e.g., we can take $\messpace:= \bit^\ell$ for sufficiently large $\ell$).    
Let $\randspace$ be $V$'s randomness space.
For any fixed statement $x$ and randomness $r\in\randspace$, $V$'s message in $(2i+1)$-th round can be seen as a deterministic function of $(\mes_1,...,\mes_i)$ where $(\mes_1,...,\mes_i)$ are prover's first $i$ messages. 
Similarly, $V$'s final decision can be seen as a deterministic function of all prover's messages $(\mes_1,...,\mes_k)$.  
We denote this function by $F[x,r]:\messpace^{\leq  k} \rightarrow \messpace \cup \{\top,\bot\}$.
Here, $F[x,r]$ outputs an element of $\messpace$ if the input is in $\messpace^{\leq  {k-1}}$ and outputs $\top$ or $\bot$ if the input is in $\messpace^{k}$.
We denote by $\Acc[x,r]\subseteq \messpace^k$ the subset consisting of  all $(\mes_1,...,\mes_k)$ such that $F[x,r](\mes_1,...,\mes_k)=\top$, i.e, all accepting transcripts corresponding to the statement $x$ and randomness $r$. 
A quantum algorithm having black-box access to $V$ means the algorithm has quantum superposition access to the function $F[x, r]$. 
 
\subsection{Strict-Polynomial-Time Simulation for Inefficient Verifier}\label{sec:strict_poly}

Our first goal is to prove the impossibility of strict-polynomial-time black-box simulation for inefficient verifiers. Assuming $\Pi = (P, V)$ is a constant-round post-quantum black-box zero-knowledge argument for a language $L$. We will show the impossibility for \emph{random aborting verifiers}, which works similarly to the honest verifier $V$ except that it aborts with probability $1-\epsilon$ for some $\epsilon\in[0,1]$ in each round.
More precisely, for any $\epsilon\in[0,1]$, we consider a malicious verifier $V^*$ with an auxiliary input $\ket{\psi_\epsilon}_{\regR,\regH}$ as follows.

Intuitively, $V^*$ works in the same way as $V$ except on each input $(m_1, \cdots, m_i) \in \messpace^{\leq k}$, it returns $F[x, r](m_1, \cdots, m_i)$ with probability $\epsilon$ and aborts (outputs $\bot$) with probability $1 - \epsilon$. Therefore, $V^*$ prepares randomness $r$ for $V$'s decision function $F[x, r]$ and a random function $H$ that decides if it outputs $\bot$. Here $H$ is  drawn from a distribution $\mathcal{H}_\epsilon$ where $\mathcal{H}_\epsilon$ is a distribution over $H:\messpace^{\leq k}\rightarrow \bit$ such that we have  $\Pr[H(\mes_1,...,\mes_i)=1]=\epsilon$ independently for each $(\mes_1,...,\mes_i)\in \messpace^{\leq k}$. 

\begin{description}
\item $\ket{\psi_\epsilon}_{\regR,\regH}$ is a superposition of $(r,H)$ according to the distribution  $\{(r,H):r\sample \mathcal{R}, H\sample \mathcal{H}_\epsilon\}$ Formally,
\[
\ket{\psi_\epsilon}_{\regR,\regH}=\sum_{r\in \mathcal{R},H\in \func(\messpace^{\leq k},\bit)}\sqrt{\frac{D(H)}{|\mathcal{R}|}}\ket{r,H}_{\regR,\regH}
\]
where $D$ is the density function corresponding to $\mathcal{H}_{\epsilon}$.
Here, $H$ is represented as a concatenation of function values $H(\vmes)$ for all $\vmes \in \messpace^{\leq k}$. 
We denote by $\regH_{\vmes}$ the sub-register of $\regH$ that stores $H(\vmes)$.

\begin{remark}\label{rem:inefficient_malicious_verifier}
Note that the state $\ket {\psi_\epsilon}_{R, H}$ is exponentially large (of length $O(|\messpace|^k)$). Therefore our malicious verifier $V^*$ (which uses this state as its auxiliary input) is inefficient. 

\end{remark}

\item $V^*$ works over its internal register $(\regX,\regaux=(\regR,\regH), \regW = ( \regcount ,\regM_1,...,\regM_k, \regB))$ and an additional message register $\regM$. 
We define the output register as $\regout:=\regB$.  
It works as follows where $\regcount$ stores a non-negative integer smaller than  $k$ (i.e. $\{0, 1, \cdots, k - 1\}$), each register of $\regM_1,...,\regM_k$ and $\regM$ stores an element of $\messpace$ and $\regB$ stores a single bit. $\regM$ is the register to store messages from/to external prover, and $\regM_i$ is a register 
to record the $i$-th message from the prover.


We next explain the unitary $U^*$ for $V^*$. The interaction between $V^*$ and  the honest prover $P$ has been formally defined in \Cref{sec:prelim_interactive_proof}. We recall it here.  



\begin{enumerate}
\item $V^*$ takes inputs a statement $x$  and a quantum auxiliary input $\ket {\psi_\epsilon}$: $\regX$ is initialized to be $\ket{x}_{\regX}$ where $x$ is the statement to be proven, $\regaux$ is initialized to be $\ket{\psi_\epsilon}_{\regR,\regH}$,
and all other registers are initialized to be $0$.
    \item 
    \emph{Verifier $V^*$ on round $< k$}: Upon receiving the $i$-th message from $P$ for $i<\revise{k}$ in $\regM$, swap $\regM$ and  $\regM_i$ and increment the value in $\regcount$\footnote{More precisely, this maps $\ket i$ to $\ket {(i + 1)\bmod k}$. }.  
    We note that $V^*$ can know $i$ since it keeps track of which round it is playing by the value in $\regcount$.  
    Let $(\mes_1,...,\mes_i)$ be the messages sent from $P$ so far.
    Then do the following in superposition where 
    $(\mes_1,...,\mes_i)$, 
    $r$, $H$ and $m$ are values in registers 
    $(\regM_1,...,\regM_i)$,  
    $\regR$, $\regH$ and $\regM$: 
    \begin{itemize}
        \item If $H(\mes_1,...,\mes_i)=0$, then do nothing.
        \item If $H(\mes_1,...,\mes_i)=1$, then add  $F[x,r](\mes_1,...,\mes_i)$ to the message register $m$.
    \end{itemize}
    
    It then measures the register $\regM$ and sends the result to the prover $P$. 
    
\vspace{1em}
    
    \emph{Unitary $U^*$ on $\regcount < k - 1$}: It acts on registers $\regcount$ (whose value is less than $k - 1$), $\regX, (\regM_1,...,\regM_k)$,
    $\regR$, $\regH$ and $\regM$: 
        \begin{itemize}
            \item It reads the value $j$ in $\regcount$ and increments it to $i:=j+1$. It swaps $\regM$ and $\regM_{i}$ (in superposition).  

            
            \item Let $x, (m_1, \cdots, m_i), r, H$ and $m$ be the values in registers $\regX, (\regM_1, \cdots, \regM_i), \regR, \regH$ and $\regM$. Let $F^*[x, r, H]$ be the following function: on input $(m_1, \cdots, m_i) \in \messpace^{\leq k - 1}$, 
\begin{align*}
F^*[x,r,H](\mes_1,...,\mes_i):=
\begin{cases}
F[x,r](\mes_1,...,\mes_{i}) &\text{~if~} H(\mes_1,...,\mes_{i})=1 \\
0 &\text{~otherwise~}
\end{cases}.    
\end{align*}

It then applies the function in superposition. 
\begin{align*}
    \ket {x, m_1, \cdots, m_i, r, H, m}  \to \ket {x, m_1, \cdots, m_i, r, H, m + F^*[x, r, H](m_1, \cdots, m_i)}. 
\end{align*}
\end{itemize}

    \item 
    \emph{Verifier $V^*$ on round $k$}: 
    Upon receiving the $k$-th message from $P$ in $\regM$,  swap $\regM$ and $\regM_k$  and increment the value in $\regcount$. 
    Then flip the bit in $\regB$ if 
    $(\mes_1,...,\mes_k)\in \Acc[x,r]$ and 
    $H(\mes_1,...,\mes_i)=1$ for all $i\in[k]$
    where 
    $(\mes_1,...,\mes_k)$, 
    $r$, and $H$ are values in registers 
    $(\regM_1,...,\regM_k)$,  
    $\regR$, and $\regH$.

\vspace{1em}
    
\emph{Unitary $U^*$ on $\regcount= k - 1$}: It acts on registers $\regcount$ (whose value is exactly equal to $k-1$), $\regX, (\regM_1,...,\regM_k)$,
    $\regR$, $\regH$ and $\regB$: 
        \begin{itemize}
            \item  It reads the value $j = k-1$ in $\regcount$ and sets it to $0$. It swaps $\regM$ and $\regM_{k}$ (in superposition). 
            \item Let $x, (m_1, \cdots, m_k), r, H$ and $b$ be the values in registers $\regX, (\regM_1, \cdots, \regM_k), \regR, \regH$ and $\regB$. Let $F^*[x, r, H]$ be the following function: on input $(m_1, \cdots, m_k) \in \messpace^k$, 
\begin{align*}
F^*[x,r,H](\mes_1,...,\mes_k):=
\begin{cases}
1 & \substack{\text{~if~} H(\mes_1,...,\mes_{i})=1,\,\forall i \in [k] \\ \text{~and~}  F[x, r](m_1, \cdots, m_k) = \top} \\
0 &\text{~otherwise~} 
\end{cases}.    
\end{align*}

It applies the function in superposition. 
\begin{align*}
    \ket {x, m_1, \cdots, m_k, r, H, b}  \to \ket {x, m_1, \cdots, m_k, r, H, b + F^*[x, r, H](m_1, \cdots, m_k)}. 
\end{align*}
\end{itemize}
\end{enumerate}
\end{description}

With the description of $V^*$ above, we have the following observation. 
\begin{observation} \label{ob:measure_aux_input}
    Let $M_{\regR, \regH}$ be the measurement on registers $\regR, \regH$. 
    For any (inefficient) black-box simulator $\siml$ it has zero advantage of distinguishing if it has black-box access to $V^*(x; \ket {\psi}_{\regR, \regH})$ or $V^*(x;  M_{\regR, \regH} \circ \ket{\psi}_{\regR, \regH})$. 
\end{observation}
Observation \ref{ob:measure_aux_input} says that even for an unbounded simulator with black-box access to $V^*(x; \ket{\psi}_{\regR, \regH})$, it has no way to tell if the auxiliary input $\ket{\psi}_{\regR, \regH}$ gets measured at the beginning or never gets measured.  

This is because registers $\regH$ and $\regR$ are only used as control qubits throughout the execution of $\siml^{\ver^*(x;\ket{\psi_\epsilon})}(x)$, we can trace out registers $\regR, \regH$ while preserving the behavior of this simulator. 


\begin{observation} \label{ob:simulatable}
$V^*$ can be simulated giving oracle access to $F^*[x, r, H]$.
\end{observation}
This can be easily seen from the description of $U^*$ above. 

\begin{observation} \label{ob:F_to_H}
Given $x$ and $r$, a quantum oracle that computes 
$F^*[x, r, H]$ can be simulated by $2k$ quantum oracle access to $H$.
\end{observation}
This is can be easily seen from the definition of $F^*[x,r,H]$.
Note that we require $2k$ queries instead of $k$ queries since we need to compute $k$ values of $H$ to compute $F^*[x,r,H]$ (when the input is in $\messpace^k$) and then need to uncompute them.

\vspace{1em}

We then prove the following lemma.  
\begin{lemma}\label{lem:impossible_strict_poly_simulation}
For any
$q=\poly(\secpar)$ there exists noticeable 
$\epsilon^*_q$ such that the following holds for any noticeable $\epsilon\leq \epsilon^*_q$:
if there exists a quantum black-box simulator $\siml$ that makes at most $q$ quantum queries, such that we have 
\[
\Pr\left[M_\regB\circ\OUT_{\ver^*}\left(\siml^{\ver^*(x;\ket{\psi_\epsilon})}(x)\right)=1\right]\geq \frac{\epsilon^k}{4}-\negl(\secpar)
\]
for all $x\in L \cap \bit^\secpar$ where $M_\regB$ means measuring and outputting the register $M_\regB$,
then we have $L\in \BQP$.
\end{lemma}
The above lemma immediately implies the following corollary, which can be seen as the quantum generalization of the result of \cite{STOC:BarLin02}.

\begin{corollary}\label{cor:impossibility_strict_poly}
If there exists a constant-round post-quantum black-box  zero-knowledge argument for a language $L$ with a simulator that makes a fixed polynomial number of queries and works for all possibly inefficient verifiers,  then $L\in \mathbf{BQP}$.
\end{corollary}
\begin{proof}

Let $\Pi = (P, V)$ be a constant-round post-quantum black-box zero-knowledge argument for a language $L$ where $P$ sends $k=O(1)$ messages and $V^*$ and $\ket{\psi_\epsilon}$ are defined above. 
Suppose that there exists a fixed polynomial $q=\poly(\secpar)$ such that there exists a quantum black-box simulator $\siml$ for $\Pi$ that works for all possibly inefficient verifiers.
When  $V^*$ takes an auxiliary input  $\ket{\psi_\epsilon}$, it can be seen as a verifier that works similarly to the honest verifier except that it aborts with probability $1-\epsilon$ in each round, and $\regB$ takes $1$ if and only if it does not abort until the end of the protocol.
Therefore, for any $x\in L \cap \bit^{\secpar}$ and its witness $w\in R_L(x)$,  we have 
\[
\Pr[M_{\regB}\circ \OUT_{\ver^*}\left(\execution{P(w)}{V^*(\ket{\psi_\epsilon})}(x)\right)=1]=\epsilon^k.
\]
By the zero-knowledge property, the above equality implies 
\[
\Pr\left[ M_\regB\circ\OUT_{\ver^*}\left(\siml^{\ver^*(x;\ket{\psi_\epsilon})}(x)\right)=1\right]=\epsilon^k-\negl(\secpar)\geq \frac{\epsilon^k}{4}-\negl(\secpar).
\]
Since this holds for arbitrary $\epsilon$, by Lemma \ref{lem:impossible_strict_poly_simulation}, this implies $L \in \mathbf {BQP}$.

\end{proof}
\begin{remark}
The above proof assumes that the  simulator works for possibly inefficient verifiiers since $V^*$ is inefficient as noted in \cref{rem:inefficient_malicious_verifier}. 
Actually, we can generalize it to rule out a strict-polynomial-query simulator that only works for efficient verifiers by considering an efficient variant of $V^*$ that is indistinguishable from $V^*$ from the view of $\siml$ that makes at most $q$ queries.
Since this generalized version is also subsumed by Theorem \ref{thm:main_impossibility_BB_QZK}, we omit the details.
\end{remark}


Then we prove Lemma \ref{lem:impossible_strict_poly_simulation}. 



\begin{proof}[Proof of Lemma \ref{lem:impossible_strict_poly_simulation}]
By Observation \ref{ob:measure_aux_input}, we can assume $\ket {\psi_\epsilon}$ is measured at the beginning. In other words, the auxiliary state is sampled as $\ket {r}_\regR \ket {H}_\regH$ for $r\sample \mathcal{R}, H\sample \mathcal{H}_\epsilon$. 
Once $r$ and $H$ are fixed, the unitary $U^*$ (corresponding to $V^*$) and its inverse can be simulated by a single quantum access to a classical function
$F^*[x,r,H]$ (defined in the description of $V^*)$.

Moreover, if we let $\Acc^*[x,r,H]\subseteq \Acc[x,r]$ be the set of $\vmes=(\mes_1,...,\mes_k)\in  \Acc[x,r]$ such that $H(\mes_1,...,\mes_i)=1$ for all $i\in [k]$, 
after the execution of $\siml^{\ver^*(x;\ket r \ket H)}(x)$,
$\regB$ contains $1$ if and only if $(\regM_1,...,\regM_k)$ contains an element in $\Acc^*[x,r,H]$. 
Therefore, for proving Lemma \ref{lem:impossible_strict_poly_simulation}, it suffices to prove the following lemma.
\begin{lemma}\label{lem:impossible_strict_poly_classical_verf}
For any 
$q=\poly(\secpar)$ there exists noticeable $\epsilon^*_q$ such that the following holds for any noticeable $\epsilon\leq \epsilon^*_q$: if there exists an oracle-aided quantum algorithm $\mathcal{S}$ that makes at most $q$ quantum queries,\footnote{Though $\mathcal{S}$ can be seen as a quantum black-box simulator for a malicious classical verifier, we do not call it a simulator since this deviates our syntax of quantum black-box simulators defined in \Cref{sec:prelim_interactive_proof}} such that we have 
\[
\Pr_{r\sample \mathcal{R},H\sample \mathcal{H}_\epsilon}\left[\mathcal{S}^{F^*[x,r,H]}(x)\in \Acc^*[x,r,H]\right]\geq \frac{\epsilon^k}{4}-\negl(\secpar)
\]
for all $x\in L \cap \bit^\secpar$,
then we have $L\in \BQP$.
\end{lemma} 

We prove the above lemma below. Assuming Lemma \ref{lem:impossible_strict_poly_classical_verf}, we show Lemma \ref{lem:impossible_strict_poly_simulation} holds. For any quantum black-box simulator $\siml$ that makes at most $q$ quantum queries, such that 
\[
\Pr\left[M_\regB\circ\OUT_{\ver^*}\left(\siml^{\ver^*(x;\ket{\psi_\epsilon})}(x)\right)=1\right]\geq \frac{\epsilon^k}{4}-\negl(\secpar). 
\]
By Observation \ref{ob:measure_aux_input}, we have 
\begin{align*}
& \Pr_{r\sample \mathcal{R}, H\sample \mathcal{H}_\epsilon}\left[M_\regB\circ\OUT_{\ver^*}\left(\siml^{\ver^*(x;{\ket {r, H} })}(x)\right)=1\right] \\
= & \Pr\left[M_\regB\circ\OUT_{\ver^*}\left(\siml^{\ver^*(x;\ket{\psi_\epsilon})}(x)\right)=1\right]\geq \frac{\epsilon^k}{4}-\negl(\secpar). 
\end{align*}

Finally, we note that $M_\regB\circ\OUT_{\ver^*}\left(\siml^{\ver^*(x;{\ket {r, H} })}(x)\right)$ can be computed by only having black-box access to $F^*[x, r, H]$ (by Observation \ref{ob:simulatable}). It outputs $1$ (the register $\regB$ is $1$) if and only if the values $(m_1, \cdots, m_k)$ in $\regM_1, \cdots, \regM_k$ are in $\Acc^*[x, r, H]$. Thus, there is an algorithm $\mathcal{S}$ that computes $M_\regB\circ\siml^{\ver^*(x;{\ket {r, H} })}(x)$ and measures registers $\regM_1, \cdots, \regM_k$. Such an algorithm $\mathcal{S}$ satisfies the requirement in Lemma \ref{lem:impossible_strict_poly_classical_verf}. Therefore $L$ is in $\mathbf{BQP}$.

\end{proof}

The remaining part is to prove Lemma \ref{lem:impossible_strict_poly_classical_verf}.
\begin{proof}[Proof of Lemma \ref{lem:impossible_strict_poly_classical_verf}]
We let $\epsilon^*_q:={1} / {(256 k^2q^2(4k q}+1)^{2k})$ and $\epsilon\leq \epsilon^*_q$ be an arbitrary noticeable function in $\secpar$. 
In the following, we simply write $r$ and $H$ in subscripts of probabilities to mean $r\sample \mathcal{R}$ and $H\sample \mathcal{H}_\epsilon$ for notational simplicity. 
As observed in Observation \ref{ob:F_to_H}, $F^*[x,r,H]$ can be simulated by $2k$ quantum invocations of $H$ if we know $x$ and $r$.
Therefore, we can view $\mathcal{S}^{F^*[x,r,H]}(x)$ as an oracle-aided algorithm with quantum access to $H$ in which $x$ and $r$ are hardwired, and makes at most $2k q$ queries to $H$. 
We denote this algorithm by $\A[x,r]^H$. 
Then for any $x\in L\cap \bit^\secpar$, we have 
\begin{align}
    \Pr_{r,H}\left[\A[x,r]^{H}\in \Acc^*[x,r,H]\right]=\Pr_{r,H}\left[\mathcal{S}^{F^*[x,r,H]}(x)\in \Acc^*[x,r,H]\right]\geq \frac{\epsilon^k}{4}-\negl(\secpar). \label{eq:A_and_siml}
\end{align}

We apply Lemma \ref{lem:measure_and_reprogram_ordered} to $\A[x,r]$. 
For any $x,r,H$,  
$\vmes^*=(\mes^*_1,...,\mes^*_k)\in \messpace^k$, 
and $\vbeta=(\beta_1,...,\beta_k)\in \bit^k$, we have 
\begin{align}
\Pr\left[\widetilde{\A[x,r]}^{\mathsf{ord}}[H,\vbeta]=\vmes^*\right]\geq \frac{1}{(2\cdot 2k q +1)^{2k}}\Pr\left[ \A[x,r]^{\Hord_{\vmes^*,\vbeta}}=\vmes^*\right]  \label{eq:inequality_measure_and_reprogram}
\end{align}
where $\widetilde{\A[x,r]}^{\mathsf{ord}}[H,\vbeta]$ and $\Hord_{\vmes^*,\vbeta}$ are  as defined in Lemma \ref{lem:measure_and_reprogram_ordered} and $2k q$ is the number of queries to $H$ made by $\A[x, r]$. 

Let $\vone:=(1,...,1)\in \bit^k$ and $H_0:\messpace^{\leq k}\rightarrow \bit$ be the zero-function i.e., $H_0(\mes_1,...,\mes_i)=0$ for all $(\mes_1,...,\mes_i)\in \messpace^{\leq k}$. 
Then we prove the following claims.
\begin{claim}\label{cla:yes_instance}
For any $x\in L\cap \bit^\secpar$, we have 
\begin{align*}
    \Pr_{r}\left[\widetilde{\A[x,r]}^{\mathsf{ord}}[H_0,\vone]\in \Acc[x,r]\right]\geq \frac{1}{8(4k q + 1)^{2k}}-\negl(\secpar). 
\end{align*}
\end{claim}
\begin{claim}\label{cla:no_instance}
For any $x\in \bit^\secpar\setminus L$, we have 
\begin{align*}
    \Pr_{r}\left[\widetilde{\A[x,r]}^{\mathsf{ord}}[H_0,\vone]\in \Acc[x,r]\right]= \negl(\secpar). 
\end{align*}
\end{claim}
Roughly, we prove Claim \ref{cla:yes_instance} by using the ordered version of measure-and-reprogram lemma (Lemma \ref{lem:measure_and_reprogram_ordered}) and Claim \ref{cla:no_instance} by reducing to soundness of the protocol $\Pi$.
Proofs of these claims are given later.

In the rest of this proof, we prove Lemma \ref{lem:impossible_strict_poly_classical_verf} assuming that Claim \ref{cla:yes_instance} and \ref{cla:no_instance} are true.
We construct a QPT algorithm $\B$ that decides $\lang$.

\begin{description}
\item $\B$ takes $x$ as input, and its goal is to decide if $x\in \lang$. 
It randomly chooses $r\sample \randspace$, runs $\widetilde{\A[x,r]}^{\mathsf{ord}}[H_0,\vone]$, and outputs $1$ if the output is in $\Acc[x,r]$ (by computing $F[x, r]$). 
\end{description}
By Claim \ref{cla:yes_instance}, for any $x\in \lang \cap \bit^\secpar$, we have 
\begin{align*}
    \Pr[\B(x)=1]\geq \frac{1}{8(4 k q + 1)^{2k}}- \negl(\secpar).
\end{align*}
On the other hand, 
by Claim \ref{cla:no_instance}
for any $x\in \bit^\secpar\setminus\lang$, we have 
\begin{align*}
    \Pr[\B(x)=1]\leq  \negl(\secpar).
\end{align*}
This means $\lang \in \BQP$.
This completes the proof of  Lemma \ref{lem:impossible_strict_poly_classical_verf}.
\end{proof}

What are left are proofs of Claim  \ref{cla:yes_instance} and \ref{cla:no_instance}.
\begin{proof}[Proof of  Claim  \ref{cla:yes_instance}]
Noting that $\widetilde{\A[x,r]}^{\mathsf{ord}}[H,\vbeta]$ can be seen as an oracle-aided algorithm that makes at most $2k q$ quantum queries to $H$, by the indistinguishability of sparse and zero functions (\cref{lem:ind_sparse_and_zero}), 
for any $x,r,\vbeta$, we have 
\begin{align}
  \left|\Pr_H\left[\widetilde{\A[x,r]}^{\mathsf{ord}}[H,\vbeta]\in \Acc[x,r]\right]-\Pr\left[\widetilde{\A[x,r]}^{\mathsf{ord}}[H_0,\vbeta]\in \Acc[x,r]\right]\right| \leq 32k^2 q^2 \epsilon.  \label{eq:ind_H_and_zero}
\end{align}

Let    
$D_\epsilon$ be a distribution over $\bit^k$ whose each coordinate takes $1$ with probability $\epsilon$ independently 
 and $\widehat{H}(\vmes):=(H(\mes_1),H(\mes_1,\mes_2),...,H(\mes_1,...,\mes_k))$ for $\vmes=(\mes_1,...,\mes_k)$. 
Then for any $x\in \lang \cap \bit^\secpar$, we have
\begin{align*}
&\Pr_{r}\left[\widetilde{\A[x,r]}^{\mathsf{ord}}[H_0,\vone]\in \Acc[x,r]\right] \\
\geq& \Pr_{r,H}\left[\widetilde{\A[x,r]}^{\mathsf{ord}}[H,\vone]\in \Acc[x,r]\right]- 32k^2 q^2 \epsilon\\
=&\sum_{\vmes^*\in \Acc[x,r]}\Pr_{r,H}\left[\widetilde{\A[x,r]}^{\mathsf{ord}}[H,\vone]= \vmes^*\right]-  32k^2 q^2 \epsilon\\
\geq& \frac{1}{(4k q+1)^{2k}}\sum_{\vmes^*\in \Acc[x,r]}\Pr_{r,H}\left[\vmes=\vmes^*:\vmes\sample \A[x,r]^{\Hord_{\vmes^*,\vone}}\right]-  32k^2 q^2 \epsilon\\
=&\frac{\epsilon^{-k}}{(4k q +1)^{2k}}\sum_{\vmes^*\in \Acc[x,r]}\Pr_{r,H,\vbeta\sample D_\epsilon}\left[\vmes=\vmes^*\land \vbeta=\vone :\vmes\sample \A[x,r]^{\Hord_{\vmes^*,\vbeta}}\right]- 32k^2 q^2 \epsilon\\
=&\frac{\epsilon^{-k}}{(4k q+1)^{2k}}\sum_{\vmes^*\in \Acc[x,r]}\Pr_{r,H}\left[\vmes=\vmes^*\land \widehat{H}(\vmes)=\vone :\vmes\sample \A[x,r]^{H}\right]- 32k^2 q^2 \epsilon\\
=&\frac{\epsilon^{-k}}{(4k q+1)^{2k}}\Pr_{r,H}\left[\vmes\in \Acc[x,r]\land \widehat{H}(\vmes)=\vone :\vmes\sample \A[x,r]^{H}\right]-  32k^2 q^2 \epsilon\\
=&\frac{\epsilon^{-k}}{(4k q+1)^{2k}}\Pr_{r,H}\left[\A[x,r]^{H}\in \Acc^*[x,r,H]\right]- 32k^2 q^2 \epsilon\\
\geq& \frac{\epsilon^{-k}}{(4k q+1)^{2k}}\left(\epsilon^k/4-\negl(\secpar)\right)-  32k^2 q^2 \epsilon \\
=&\frac{1}{4(4k q+1)^{2k}}- 32k^2 q^2 \epsilon-\negl(\secpar)\\
=& \frac{1}{8(4k q+1)^{2k}}-\negl(\secpar), 
\end{align*}
where the first inequality follows from Eq. \ref{eq:ind_H_and_zero} for $\vbeta:=\vone$, the second inequality follows from Eq. \ref{eq:inequality_measure_and_reprogram}, 
the third inequality follows from Eq. \ref{eq:A_and_siml}, and 
the last equality follows from 
$\epsilon\leq \epsilon^*_q={1} / {(256 k^2q^2(4k q}+1)^{2k})$.
This completes the proof of Claim \ref{cla:yes_instance}.
\end{proof}

\vspace{1em}

\begin{proof}[Proof of Claim \ref{cla:no_instance}]
We consider a cheating prover $P^*$ described as follows.
Intuitively, $P^*$ just runs $\widetilde{\A[x,r]}^{\mathsf{ord}}[H_0,\vone]$ where $r$ is chosen by the external verifier.
We first look at how $\widetilde{\A[x,r]}^{\mathsf{ord}}[H_0,\vone]$ runs:  it runs $\mathcal{S}$ in the experiment and uses the oracle access to $H$ and $F[x, r]$ to simulate the oracle $F^*[x, r, H]$. 

The only difference between $\widetilde{\A[x,r]}^{\mathsf{ord}}[H_0,\vone]$ and $P^*$ is that  $\widetilde{\A[x,r]}^{\mathsf{ord}}[H_0,\vone]$ can compute $F[x, r]$ on its own because it samples and knows the random tape $r$ but $P^*$ does not know the randomness $r$ of the honest verifier. However, we show that it can still answer queries to $F[x,r ]$ because it needs $r$ only when responding to measured queries, and $P^*$ can then send such (classical) queries to the external verifier to get the response. 

More precisely, $P^*$ only makes queries to $F[x, r]$ on measured inputs $\vmes$, because only in this case the updated oracle in the game $\ora(\vmes)$ may not be $0$; in all inputs $\vmes$,  because $\ora(\vmes)$ is initialized as $0$ and never gets updated in the experiment, the output of $F[x, r]$ on that input is not needed. 

Formally, $P^*$ is described as follows. We will mark the difference between $P^*$ and $\widetilde{\A[x,r]}^{\mathsf{ord}}[H_0,\vone]$ with \ul{underline}. 
\begin{description}
\item[$P^*(x)$:] 
The cheating prover $P^*$ interacts with the external verifier as follows:
\begin{enumerate}
    \item 
    For each $i\in[k]$, uniformly pick $(j_i,b_i)\in ([2kq]\times \bit) \cup \{(\bot,\bot)\}$  conditioned on that 
    there exists at most one $i\in [k]$ such that $j_i=j^*$  for all $j^*\in [2kq]$.
    \item Run $\mathcal{S}$ 
    where the oracle $F^*[x, r, H]$ is simulated by additional oracles $\ora$ and $\mathcal{F}$ where $\ora$ and $\mathcal{F}$ play the roles of  $H$ and $F[x,r]$, respectively. 

    The oracle $\ora$ (for simulating $H$) and $\mathcal{F}$ (for simulating $F[x, r]$) are initialized to be an oracle that just return $0$.
    For the rest of the description, we can assume $\mathcal{S}$ is now making queries to both $\ora$ and $\mathcal{F}$. 
    
    When $\mathcal{S}$ makes its $j$-th query to $\ora$, 
    \begin{enumerate}
      \item If $j=j_i$ for some $i\in[k]$, 
        measure $\mathcal{S}$'s  query register to obtain $\vmes'_i=(\mes'_{i,1},...,\mes'_{i,k_i})$ for some $k_i\leq k$.
        \ul{If the transcript between $V$ at this point is inconsistent to $\vmes'_i$ (i.e., there is $\ell\in [k_i]$ such that $P^*$ already sent an $\ell$-th message different from $\mes'_{i,\ell}$ to the external verifier), then just abort. 
        Otherwise, run the protocol between the external verifier until $2k_i$-th round by using  $(\mes'_{i,1},...,\mes'_{i,k_i})$ as the first $k_i$ prover's messages.} 
        
        \ul{It then updates $\mathcal{F}$ such that for each $j \in [k_i]$, $\mathcal{F}$ on input $(m'_{i, 1}, \cdots, m'_{i, j})$ is compatible with the current transcript.} 
        
        \begin{enumerate}
        \item If $b_i=0$, reprogram $\ora\leftarrow \reprogram(\ora,\vmes'_i, 1)$ and answer $\mathcal{S}$'s $j_i$-th query by using the reprogrammed oracle. 
        \item If $b_i=1$, answer  $\mathcal{S}$'s $j_i$-th query by using the oracle before the reprogramming 
        and then reprogram $\ora\leftarrow \reprogram(\ora,\vmes'_i, 1)$. 
        \end{enumerate}
    \item Otherwise, answer $\mathcal{S}$'s $j$-th query by just using the oracle $\ora$. 
    \end{enumerate}

    When $\mathcal{S}$ makes its  query to $\mathcal{F}$, 
    \ul{it uses the current updated oracle $\mathcal{F}$.}

    \item Let $\vmes=(\mes_1,...,\mes_k)$ be $\mathcal{S}$'s output.
    \ul{If the protocol between the external verifier has not been completed yet, complete the protocol by using messages $\vmes$.} 
    Again, if $\vmes$ is inconsistent to the transcript so far, just abort.
\end{enumerate}
\end{description}

Note that since $P^*$ uses the interaction with $V$ to perfectly simulate the oracle access to $F[x, r]$, by definitions of $\widetilde{\A[x,r]}^{\mathsf{ord}}[H_0,\vone]$, it is straightforward to see that $P^*$ succeeds in letting $V$ accept with probability at least  $\Pr_r[\widetilde{\A[x,r]}^{\mathsf{ord}}[H_0,\vone]\in \Acc[x,r]]$. Noting that $\widetilde{\A[x,r]}^{\mathsf{ord}}[H_0,\vone]$ returns $\bot$ whenever any two of measured queries are inconsistent (i.e., one is not a prefix of the other), and thus when it does not return $\bot$, 
$P^*$ does not abort either 
in the corresponding execution. 
Therefore, the negligible soundness of the protocol ensures $\Pr_r[\widetilde{\A[x,r]}^{\mathsf{ord}}[H_0,\vone]\in \Acc[x,r]]=\negl(\secpar)$.

This completes the proof of Claim \ref{cla:no_instance}.
\end{proof}

\subsection{Expected-Polynomial-Time  Simulation for Inefficient Verifiers}\label{sec:impossibility_expected_poly_simulation_inefficient}

In the previous section, we proved that strict-polynomial-time black-box simulation is impossible. 
In this section, as a first step to prove Theorem \ref{thm:main_impossibility_BB_QZK}, we prove that even expected-polynomial-time black-box simulation is impossible if we require it to work for all \emph{inefficient} malicious verifiers.
\begin{theorem}\label{thm:main_impossibility_BB_QZK_inefficient}
If there exists a constant-round post-quantum black-box  zero-knowledge argument for a language $L$ with a simulator that works for all inefficient malicious verifiers,   
then $L\in \mathbf{BQP}$.
\end{theorem}
Though this theorem is subsumed by Theorem \ref{thm:main_impossibility_BB_QZK}, we first prove this since the proof is simpler and thus we believe that it is easier for readers to understand the proof of Theorem \ref{thm:main_impossibility_BB_QZK} if we first give the proof of Theorem \ref{thm:main_impossibility_BB_QZK_inefficient}.

Our main idea is to consider a malicious verifier $\widetilde{V}^*$ that runs the honest verifier and a random aborting verifier in superposition.  
Roughly, $\widetilde{V}^*$ works over the same registers as those of $V^*$ and one additional register $\regcont$ that stores $1$-qubit that plays the role of a ``control qubit".
We define an auxiliary input 
\begin{align*}
\ket{\widetilde{\psi}_\epsilon}_{\regcont,\regR,\regH}:=
\frac{1}{\sqrt{2}}\left(\ket{0}_\regcont+\ket{1}_\regcont\right)\otimes \ket{\psi_\epsilon}_{\regR,\regH}
\end{align*}
where $\ket{\psi_\epsilon}_{\regR,\regH}$ is as defined in \cref{sec:strict_poly}.
Given a statement $x$ and an auxiliary input $\ket{\widetilde{\psi}_\epsilon}_{\regcont,\regR,\regH}$, 
$\widetilde{V}^*$ runs the honest verifier $V$ if the value in $\regcont$ is $0$ 
and the random aborting verifier $V^*$ if the value in $\regcont$ is $1$ in superposition.
Then it ``adjusts $\regH$" so that the states in $\regH$ becomes the same 
in both cases of $\regcont=0$ and $\regcont=1$.
The motivation of introducing this step is to make the final state in $\regcont$ be a pure state, which is essential for our analyses to work (in particular for latter Lemma \ref{lem:final_state_real}). 
For describing this ``adjusting" procedure, we first prove the following lemma.
\begin{lemma}\label{lem:uncomputing_random}
    For any $\vmes=(\mes_1,...,\mes_k)\in \messpace^k$, let $S_{\vmes}\subseteq \func(\messpace^{\leq k},\bit)$ be the subset consisting of all $H$ such that $H(\mes_1,...,\mes_i)=1$ for all $i\in[k]$.  
    There exists a  unitary $U_{\vmes}$ such that
    \[
    U_{\vmes}\sum_{H\in \func(\messpace^{\leq k},\bit)}\sqrt{D(H)}\ket{H}_{\regH}=\sum_{H\in S_{\vmes}}\sqrt{\frac{D(H)} {\epsilon^k}}\ket{H}_{\regH}.
    \]
\end{lemma}
\begin{proof}
Recall that $H$ is encoded as a concatenation of $H(\vmes')$ for all $\vmes'\in \messpace^{\leq k}$ and $\regH_{\vmes'}$ denotes the register to store  $H(\vmes')$.
Then it is easy to see that we have 
\begin{align*}
    \sum_{H\in \func(\messpace^{\leq k},\bit)}\sqrt{D(H)}\ket{H}_{\regH}=\bigotimes_{\vmes'\in \messpace^{\leq k}}\left(\sqrt{1-\epsilon}\ket{0}_{\regH_{\vmes'}}+\sqrt{\epsilon}\ket{1}_{\regH_{\vmes'}}\right)
\end{align*}
and
\begin{align*}
   \sum_{H\in S_{\vmes}}\sqrt{\frac{D(H)} {\epsilon^k}}\ket{H}_{\regH}=
   \left(\bigotimes_{\vmes'\in \mathsf{Prefix}_{\vmes}}\ket{1}_{\regH_{\vmes'}}\right)
   \otimes
   \left(
   \bigotimes_{\vmes'\notin \mathsf{Prefix}_{\vmes}}\left(\sqrt{1-\epsilon}\ket{0}_{\regH_{\vmes'}}+\sqrt{\epsilon}\ket{1}_{\regH_{\vmes'}}\right)
   \right)
\end{align*}
where $\mathsf{Prefix}_{\vmes}\subseteq \messpace^{\leq k}$ is the set of all prefixes of $\vmes$, i.e., $\mathsf{Prefix}_{\vmes}=\{\mes_1,(\mes_1,\mes_2),...,(\mes_1,...,\mes_k)\}$.
For each $\vmes'\in \messpace^{\leq k}$, we define $U'_{\vmes'}$ as a unitary on $\regH_{\vmes'}$ that satisfies 
\[
U'_{\vmes'}\ket{1}_{\regH_{\vmes'}}=\sqrt{1-\epsilon}\ket{0}_{\regH_{\vmes'}}+\sqrt{\epsilon}\ket{1}_{\regH_{\vmes'}}.
\]
We define $U_{\vmes}$ as
\[
U_{\vmes}:=\prod_{\vmes' \in \mathsf{Prefix}_{\vmes}}{U'_{\vmes'}}^{\dagger}.
\]
Then the equation in Lemma \ref{lem:uncomputing_random} clearly holds.
\end{proof}

The formal description of $\widetilde{V}^*$ is given below.
\begin{description}
\item $\widetilde{V}^*$ works over its internal register $\regV=(\regX,\regaux=(\regcont,\regR,\regH),\regW=(\regcount,\regM_1,...,\regM_k,\regB))$ and an additional message register $\regM$ where $\regcont$ is a single-qubit register and all other registers are similar to those of $V^*$ in \cref{sec:strict_poly} except that $\regaux$ contains an additional register $\regcont$.  
The output register is designated as $\regout:=(\regcont,\regB)$.
The unitary $\widetilde{U}^*$ for $\widetilde{V}^*$ is defined as follows
\begin{align*}
   & \widetilde{U}^*\left(\ket{0}_{\regcont}\ket{{\sf other}_0}_{\regother}+\ket{1}_{\regcont}\ket{{\sf other}_1}_{\regother}\right) \\
= & \ket{0}_{\regcont} (U_{\mathsf{hon}} \ket{{\sf other}_0}_{\regother})+\ket{1}_{\regcont} (U^* \ket{{\sf other}_1}_{\regother}), 
\end{align*}
where $\regother$ denotes all registers except for $\regcont$, $U^*$ is the unitary for $V^*$ as defined in \cref{sec:strict_poly}, and
$U_{\mathsf{hon}}$ is the unitary that corresponds to the honest verifier with an additional ``adjusting unitary" $U_{\vmes}$ on $\regH$.
Formally, $U_{\mathsf{hon}}$ is defined as follows.

\begin{description}
\item[Unitary $U_{\mathsf{hon}}$:] 
It non-trivially acts on registers $\regcount$, $\regX, (\regM_1,...,\regM_k)$,
    $\regR$, $\regH$, $\regM$, and $\regB$: 
        \begin{itemize}

\item It reads the value $j$ in $\regcount$ and increments it to $i=j+1 \bmod k$. It swaps $\regM$ and $\regM_{i}$ (in superposition).

\item 
It does either of the following depending on the value in $\regcount$ in superposition.

\begin{enumerate}
\item If the value in $\regcount$ is $i\neq 0$ (i.e., it is not in the final round), 
it applies the following unitary over $\regX$, $\regM_1,...,\regM_i$, $\regR$, and $\regM$: 
\begin{align*}
    \ket {x, m_1, \cdots, m_i, r, m}  \to \ket {x, m_1, \cdots, m_i, r,  m \oplus F[x, r](m_1, \cdots, m_i)}. 
\end{align*}
\item If the value in $\regcount$ is $0$ (i.e., it is in the final round), 
it applies the following unitary over $\regX$, $\regM_1,...,\regM_k$, $\regR$, and $\regB$: 
\begin{align*}
    \ket {x, m_1, \cdots, m_k, r,b}  \to \ket {x, m_1, \cdots, m_k, r,  b \oplus F[x, r](m_1, \cdots, m_k)}.
\end{align*}
Then it applies the ``adjusting unitary" over $\regM_1,...,\regM_k$, and $\regH$:
\begin{align*}
    \ket {m_1, \cdots, m_k, H}  \to U_{\vmes} \ket {m_1, \cdots, m_k, H}.
\end{align*}
where $\vmes:=(\mes_1,...,\mes_k)$. That is, it first puts $F[x, r](\vmes)$ into $\regB$, then adjusts $\regH$ using the unitary $U_{\vmes}$. 
\end{enumerate}
\end{itemize}
\end{description}
\end{description}


\begin{lemma}\label{lem:final_state_real}
For any $x\in L\cap \bit^\secpar$ and $w\in R_L(x)$, suppose that we run $\execution{P(w)}{\widetilde{V}^*(\ket{\widetilde{\psi}_\epsilon})}(x)$ and measure $\regB$ and the outcome is $1$. 
\revise{Then the resulting state in $\regcont$ (tracing out other registers) is negligibly close to $\ket{\phi_\epsilon}_{\regcont}:=\sqrt{\frac{1}{1+\epsilon^k}}\ket{0}_{\regcont}+\sqrt{\frac{\epsilon^k}{1+\epsilon^k}}\ket{1}_{\regcont}$.
}
\end{lemma}
\begin{proof}
\revise{
By the completeness of $\Pi$ and a simple averaging argument, for an overwhelming fraction of $P$'s randomness, the completeness error is $\negl(\secpar)$ even if we fix $P$'s randomness to that value.
In the following, we fix $P$'s randomness to such a value.
}
\takashi{Without this, verifier's state depends on prover's randomness, which causes subtle issues. (This should have been mentioned even for the perfectly complete case.)}
Let $\ket{\eta}$ be the final state of the internal register of $\widetilde{V}^*$ after executing $\execution{P(w)}{\widetilde{V}^*(\ket{\widetilde{\psi}_\epsilon})}(x)$.
For $\beta\in \bit$, let $\ket{\eta_\beta}$ be the final state of the internal register of $\widetilde{V}^*$ after executing $\execution{P(w)}{\widetilde{V}^*(\ket{\widetilde{\psi}_\epsilon^{(\beta)}})}(x)$.
where $\ket{\widetilde{\psi}_\epsilon^{(\beta)}}_{\regcont,\regR,\regH}:=
\ket{\beta}_\regcont\otimes \ket{\psi_\epsilon}_{\regR,\regH}
$.
Since $\widetilde{V}^*$ only uses $\regcont$ as a control register
and $\ket{\widetilde{\psi}_\epsilon}_{\regcont,\regR,\regH}=\frac{1}{\sqrt{2}}\left(\ket{\widetilde{\psi}_\epsilon^{(0)}}_{\regcont,\regR,\regH}+\ket{\widetilde{\psi}_\epsilon^{(1)}}_{\regcont,\regR,\regH}\right)$
, it is easy to see that we have 
\begin{align}  \label{eq:eta}
    \ket{\eta}=\frac{1}{\sqrt{2}}\left(\ket{\eta_0}+\ket{\eta_1}\right).
\end{align}
In the following, when we consider summations over $r$ and $H$, they are over all $r\in \mathcal{R}$ and $H\in \func(\messpace^{\leq k},\bit)$, respectively, unless otherwise specified.

By the definition of $\widetilde{V}^*$ and  Lemma \ref{lem:uncomputing_random}, we have  
\begin{align*}
 \ket{\eta_0}
 &= U_{\vmes}  \ket{0}_{\regcont}\ket{x}_{\regX}\ket{0}_{\regcount}
 \otimes \sum_{r,H}\left(\sqrt{\frac{D(H)}{|\mathcal{R}|}}\ket{r,H}_{\regR,\regH} \otimes  \revise{\ket{\vmes_r}_{\regM_1,...,\regM_k} \otimes \ket{b_{r}}_{\regB}} \right)\\
 &= \ket{0}_{\regcont}\ket{x}_{\regX}\ket{0}_{\regcount}
\otimes \sum_{r,H\in S_{\vmes_r}}\left(\sqrt{\frac{D(H)}{\epsilon^k\cdot |\mathcal{R}|}}\ket{r,H}_{\regR,\regH} \otimes  \revise{\ket{\vmes_r}_{\regM_1,...,\regM_k}\otimes \ket{b_{r}}_{\regB}}\right)
\end{align*}
\revise{
where 
$\vmes_r$ is prover's messages when verifier's randomness is $r$ (note that we fix  $P$'s randomness) 
and $b_{r}$ is a bit such that $b_{r}=1$ if and only if $F[x,r](\vmes_r)=\top$, i.e.,  the verifier accepts when the  randomness is $r$.}
\revise{
By the completeness of $\Pi$ (with the fixed value of $P$'s randomness), we have 
$\Pr[M_\regB\circ \ket{\eta_0}=1]=1-\negl(\secpar)$. 
This implies 
\begin{align} \label{eq:eta_zero}
 \ket{\eta_0}
\approx \ket{0}_{\regcont}\ket{x}_{\regX}\ket{0}_{\regcount}
\otimes \sum_{r,H\in S_{\vmes_r}}\left(\sqrt{\frac{D(H)}{\epsilon^k\cdot |\mathcal{R}|}}\ket{r,H}_{\regR,\regH} \otimes  \ket{\vmes_r}_{\regM_1,...,\regM_k}\right)\otimes \ket{1}_{\regB}
\end{align}
where $\approx$ means that the trace distance between both sides is $\negl(\secpar)$. 
}
\takashi{Do we need more explanations for this? I think this is almost clear, but don't know how to formally prove this without seemingly complicated calculations.}

On the other hand, by the definition of $\widetilde{V}^*$, the value in $\regB$ of $\ket{\eta_1}$ \revise{can be} $1$ only if $H\in S_{\vmes}$ for the transcript $\vmes$. 
Therefore we have 
\begin{align*}
 \ket{\eta_1}=   \ket{1}_{\regcont}\ket{x}_{\regX}\ket{0}_{\regcount}
 \otimes
 \Biggl(
 \begin{array}{ll}
 &
 \sum_{r,H\in S_{\vmes_r}}\left(\sqrt{\frac{D(H)}{|\mathcal{R}|}}\ket{r,H}_{\regR,\regH} \revise{\otimes  \ket{\vmes_r}_{\regM_1,...,\regM_k}
 \otimes{\ket{b_r}_{\regB}}
 }
 \right)\\
+& 
 \ket{\sf garbage}_{\regR,\regH,\regM_1,...,\regM_k}
 \otimes \ket{0}_{\regB} 
 \end{array}
 \Biggr)
\end{align*}
for some (sub-normalized) state $\ket{\sf garbage}_{\regR,\regH,\regM_1,...,\regM_k}$. 
\revise{
By a similar argument to that for $\ket{\eta_0}$, we have 
\begin{align}  \label{eq:eta_one}
 \ket{\eta_1}\approx   \ket{1}_{\regcont}\ket{x}_{\regX}\ket{0}_{\regcount}
 \otimes
 \Biggl(
 \begin{array}{ll}
 &
 \sum_{r,H\in S_{\vmes_r}}\left(\sqrt{\frac{D(H)}{|\mathcal{R}|}}\ket{r,H}_{\regR,\regH} \otimes  \ket{\vmes_r}_{\regM_1,...,\regM_k}
 \right) \otimes{\ket{1}_{\regB}}\\
+& 
 \ket{\sf garbage}_{\regR,\regH,\regM_1,...,\regM_k}
 \otimes \ket{0}_{\regB} 
 \end{array}
 \Biggr)
\end{align}
}
By \cref{eq:eta,eq:eta_zero,eq:eta_one}, we have 
\begin{align*}
(\ket{1}\bra{1})_{\regB}\ket{\eta}\approx\frac{1}{\sqrt{2}}\left(
\sqrt{\frac{1}{\epsilon^k}} 
\ket{0}_{\regcont}+\ket{1}_{\regcont}\right)
&\otimes
\ket{x}_{\regX}\ket{0}_{\regcount}
\\
 &\otimes \sum_{r,H\in S_{\vmes_r}}\left(\sqrt{\frac{D(H)}{|\mathcal{R}|}}\ket{r,H}_{\regR,\regH}\right) \otimes
 \ket{1}_{\regB}.
\end{align*}
Here we omit the identity operator on registers other than $\regB$ and  $(\ket 1 \bra 1)_\regB$  simply means the projection onto states whose values in $\regB$ is $0$.
\revise{
By normalization, we can see that the final state in $\regcont$ conditioned on the measurement outcome of $\regB$ is $1$ is negligibly close to $\ket{\phi_\epsilon}_\regcont$.
Since this holds for overwhelming fraction of $P$'s randomness, Lemma \ref{lem:final_state_real} follows by an averaging argument.}
\end{proof}

Suppose that there is  a quantum black-box simulator $\siml_{\mathsf{exp}}$ ($\sf exp$ stands for `expected') for the protocol $\Pi$ whose expected number of queries is at most $q/2=\poly(\secpar)$ that works for all possibly inefficient verifiers.\footnote{We write $\mathsf{exp}$ in the subscript to differentiate this from strict-polynomial query simulators that appear in previous subsection. We take the expected number of queries to be $q/2$ instead of $q$ just for convenience of the proof. Since $q$ can be arbitrary polynomial, this does not lose generality.}
Especially, we assume that for any $\epsilon$, we have
\begin{align} \label{eq:simulation_for_inefficitn_verifier}
\{\OUT_{\widetilde{\ver}^*}\execution{\pro(w)}{\widetilde{V}^*(x;\ket{\widetilde{\psi}_\epsilon})}(x)\}_{\secpar,x,w}
\compind
\{\OUT_{\widetilde{\ver}^*}(\siml_{\mathsf{exp}}^{\widetilde{V}^*(x;\ket{\widetilde{\psi}_\epsilon})}(x))\}_{\secpar,x,w} 
\end{align}
where 
$\secpar\in \mathbb{N}$, 
$x\in \lang\cap \bit^\secpar$, $w\in \rel_\lang(\secpar)$.
By Lemma \ref{lem:final_state_real}, we can show that the final state in $\regcont$ after the execution $\siml_{\mathsf{exp}}^{\widetilde{V}^*(x;\ket{\widetilde{\psi}_\epsilon})}(x)$
conditioned on that $\widetilde{V}^*$ accepts (i.e., the value in $\regB$ is $1$) should be close to $\ket{\phi_\epsilon}$. (Remark that the probability that $\widetilde{V}^*$ accepts is larger than $1/2$ since it always accepts if the value in $\regcont$ is $0$.)
In the following, we show stronger claims.
Specifically, we show that 
\begin{enumerate}
    \item the probability that $\widetilde{V}^*$ accepts \emph{and} the number of $\siml_{\mathsf{exp}}$'s queries is at most $q$ is at least $1/4-\negl(\secpar)$ (Lemma \ref{lem:prob_success_and_querysmall}), and
\item the final state in $\regcont$ after the execution $\siml_{\mathsf{exp}}^{\widetilde{V}^*(x;\ket{\widetilde{\psi}_\epsilon})}(x)$
conditioned on the above event 
is close to $\ket{\phi_\epsilon}$ 
(Lemma \ref{lem:final_state_sim}).
\end{enumerate}
By combining Lemmas \ref{lem:prob_success_and_querysmall} and  \ref{lem:final_state_sim}, we can show that the ``truncated version" of $\siml_{\mathsf{exp}}$ that makes at most $q$ queries can let $V^*(x;\ket{\psi_\epsilon})$ accept with probability at least $\frac{\epsilon^k}{4}-\negl(\secpar)$ (Lemma \ref{lem:truncated_simulator}).
By Lemma \ref{lem:impossible_strict_poly_simulation}, this implies $L\in \BQP$, which completes the proof of Theorem \ref{thm:main_impossibility_BB_QZK_inefficient}. The details follow.

Similarly to Observation \ref{ob:measure_aux_input}, because registers $\regcont$ is only used as a control qubit throughout the execution of $\siml^{\ver^*(x;\ket{\tilde{\psi}_\epsilon})}(x)$, we can trace out registers $\regcont$ while preserving the behavior of this simulator. We have the following observation: 
\begin{observation} \label{ob:measure_control_bit}
    Let $M_{\regcont}$ be the measurement on the register $\regcont$. 
    For any (inefficient) black-box simulator $\siml$ it has zero advantage of distinguishing if it has black-box access to $V^*(x;\ket {\tilde{\psi}})$ or $V^*(x; M_{\regcont} \ket {\tilde{\psi}})$. 
\end{observation}

For any $x\in L\cap \bit^\secpar$, we consider an experiment 
$\experiment(x,\epsilon)$  where we run $\siml_{\mathsf{exp}}^{\widetilde{V}^*(x;\ket{\widetilde{\psi}_\epsilon})}(x)$ and then measure $\regB$.
Let $\querysmall$ be the event that the number of queries made by $\siml_{\mathsf{exp}}$ is at most $q$ and $\event$ be the event that the measurement outcome of $\regB$ is $1$. 
Then we prove the following lemmas.

\begin{lemma}\label{lem:prob_success_and_querysmall}
For any $x\in L\cap \bit^\secpar$ and $\epsilon \in [0,1]$, 
we have 
\[
\Pr_{\experiment(x,\epsilon)}[\querysmall \land \Bisone]\geq 1/4-\negl(\secpar).
\]
\end{lemma}

\begin{proof}
First, we recall that no black-box simulator can distinguish oracle access to $\widetilde{V}^*(x;\ket{\widetilde{\psi}_\epsilon})$ and  $\widetilde{V}^*(x;M_{\regcont}\ket{\widetilde{\psi}_\epsilon})$ (see Observation \ref{ob:measure_control_bit}).
Let $\experiment_0(x,\epsilon)$ be the same as $\experiment(x,\epsilon)$ except that the auxiliary input of $\widetilde{V}^*$ is replaced with $\ket{0}_\regcont\ket{\psi_\epsilon}_{\regR,\regH}$.
Remark that $\ket{0}_\regcont\ket{\psi_\epsilon}_{\regR,\regH}$ is the post-measurement state after measuring $\regcont$ of $\ket{\widetilde{\psi}_\epsilon}$ conditioned on that the measurement outcome is $0$, which happens with probability $1/2$.  
By the above observation, we have
\begin{align}
\Pr_{\experiment(x,\epsilon)}[\querysmall \land \Bisone]
\geq 
\frac{1}{2}\Pr_{\experiment_0(x,\epsilon)}[\querysmall \land \Bisone]. \label{eq:exp_to_expzero}
\end{align}
Since the expected number of queries made by $\siml$ is at most $q/2$ for any malicious verifier given as an oracle, 
by Markov's inequality, 
we have 
\begin{align}
\Pr_{\experiment_0(x,\epsilon)}[\querysmall]\geq \frac{1}{2}. \label{eq:querysmall}
\end{align}
When we run $\execution{P(w)}{\widetilde{V}^*(\ket{0}_\regcont\ket{\psi_\epsilon}_{\regR,\regH})}(x)$ for some $w\in R_L(x)$ and then measure $\regB$, the measurement outcome of $\regB$ is always $1$ noting that $\widetilde{V}^*$ just runs the honest verifier followed by an additional ``adjusting unitary" $U_{\vmes}$, which does not affect the value in $\regB$, when its auxliary input is $\ket{0}_\regcont\ket{\psi_\epsilon}_{\regR,\regH}$. 
Therefore, by our assumption that $\siml_{\mathsf{exp}}$ is a simulator for the protocol $\Pi$, we must have 
\begin{align}
\Pr_{\experiment_0(x,\epsilon)}[\Bisone]=1-\negl(\secpar). \label{eq:Bisone}
\end{align}
By combining Eq. \ref{eq:exp_to_expzero},  \ref{eq:querysmall}, and \ref{eq:Bisone},  we obtain Lemma \ref{lem:prob_success_and_querysmall}:
\begin{align*}
\Pr_{\experiment(x,\epsilon)}[\querysmall \land \Bisone]
\geq & 
\frac{1}{2}\Pr_{\experiment_0(x,\epsilon)}[\querysmall \land \Bisone] \\
\geq & \frac{1}{2}\left(\Pr_{\experiment_0(x,\epsilon)}[\querysmall] - \Pr_{\experiment_0(x,\epsilon)}[\neg \Bisone]\right) \\
\geq & \frac{1}{4} - \negl(\lambda). 
\end{align*}
\end{proof}

\begin{lemma}\label{lem:final_state_sim}
For  any $x\in L\cap \bit^\secpar$ and $\epsilon$, 
let $\sigma_{x,\epsilon}$ be the state in $\regcont$ (tracing out other registers) after executing $\experiment(x,\epsilon)$ conditioned on $\querysmall\land \Bisone$.
Then we have  
\[
\TD(\sigma_{x,\epsilon},\ket{\phi_\epsilon}\bra{\phi_\epsilon})=\negl(\secpar)
\]
where $\ket{\phi_\epsilon}$ is as defined in Lemma \ref{lem:final_state_real}.
\end{lemma}
\begin{proof}

Let $w\in R_L(x)$ be an arbitrary witness for $x$. 
Let $\rho_{x,w,\epsilon}^{\mathsf{real}}:=\OUT_{\widetilde{V}^*}\execution{P(w)}{\widetilde{V}^*(\ket{\widetilde{\psi}_\epsilon})}(x)$ and $\rho_{x,\epsilon}^{\mathsf{sim}}:=\OUT_{\widetilde{V}^*}\siml_{\mathsf{exp}}^{\widetilde{V}^*(x;\ket{\widetilde{\psi}_\epsilon})}(x)$. 
We note that they are states over $\regout=(\regcont,\regB)$. 
We consider the following distinguisher $\mathcal{D}$ that tries to distinguish  $\rho_{x,w,\epsilon}^{\mathsf{real}}$ and $\rho_{x,\epsilon}^{\mathsf{sim}}$. 
\begin{description}
\item[$\mathcal{D}(\rho)$:]
It measures the register $\regB$.
If the measurement outcome is $0$, then it outputs $0$. 
Otherwise, it generates a state $\ket{\phi_\epsilon}_{\regcont'}:=\sqrt{\frac{1}{1+\epsilon^k}}\ket{0}_{\regcont'}+\sqrt{\frac{\epsilon^k}{1+\epsilon^k}}\ket{1}_{\regcont'}$ in a new register $\regcont'$, runs the SWAP test (\cref{lem:SWAP}) between registers $\regcont$  and $\regcont'$, and outputs $0$ if the SWAP test accepts and 
$1$ otherwise.
\end{description}
$\mathcal{D}$ will output $1$ if and only if $\regB$ is measured as $1$, and the SWAP test rejects.  
\revise{When $\mathcal{D}$'s input is $\rho_{x,w,\epsilon}^{\mathsf{real}}$, if $\regB$ is measured as $1$, then the state at this point is negligibly close to $\ket{\phi_\epsilon}_{\regcont}$ by \cref{lem:final_state_real}. 
Moreover, the SWAP test accepts $\ket{\phi_\epsilon}_{\regcont}$ with probability $1$ by \cref{lem:SWAP}.
}
Therefore we have 
\begin{align*}
    \Pr[\mathcal{D}(\rho_{x,w,\epsilon}^{\mathsf{real}})=1]=\revise{\negl(\secpar)}.
\end{align*}
Since $\rho_{x,w,\epsilon}^{\mathsf{real}}$ and $\rho_{x,\epsilon}^{\mathsf{sim}}$ are computationally indistinguishable by Eq. \ref{eq:simulation_for_inefficitn_verifier}, 
the above equation implies 
\begin{align}
    \Pr[\mathcal{D}(\rho_{x,\epsilon}^{\mathsf{sim}})=1]=\negl(\secpar). \label{eq:D_sim_negl}
\end{align}
On the other hand, 
by \cref{lem:SWAP}, 
we have 
\begin{align*}
    \Pr[\mathcal{D}(\rho_{x,\epsilon}^{\mathsf{sim}})=1]\geq \Pr_{\experiment(x,\epsilon)}[\querysmall \land \Bisone]\frac{1-\bra{\phi_\epsilon}\sigma_{x,\epsilon}\ket{\phi_\epsilon}}{2}.
\end{align*}
since the r.h.s. is the probability that $\mathcal{D}(\rho_{x,\epsilon}^{\mathsf{sim}})$ returns $1$ \emph{and} $\siml_{\mathsf{exp}}$ made at most $q$ queries when generating $\rho_{x,\epsilon}^{\mathsf{sim}}$. 
By Lemma \ref{lem:prob_success_and_querysmall}, $\Pr_{\experiment(x,\epsilon)}[\querysmall \land \Bisone]\geq 1/4-\negl(\secpar)$. 
Therefore, for satisfying Eq. \ref{eq:D_sim_negl}, we must have $\bra{\phi_\epsilon}\sigma_{x,\epsilon}\ket{\phi_\epsilon}=1-\negl(\secpar)$, which implies $\TD(\sigma_{x,\epsilon},\ket{\phi_\epsilon}\bra{\phi_\epsilon})=\negl(\secpar)$.
\end{proof}

\begin{lemma}\label{lem:truncated_simulator}
Let $\siml$ be the ``truncated version" of $\siml_{\mathsf{exp}}$ that works similarly to   $\siml_{\mathsf{exp}}$  except that it immediately halts when $\siml_{\mathsf{exp}}$ tries to make $(q+1)$-th query. 
Then for any $x\in L\cap \bit^\secpar$ and $\epsilon$, we have 
\begin{align*}
    \Pr\left[M_\regB\circ\OUT_{\ver^*}\left(\siml^{\ver^*(x;\ket{\psi_\epsilon})}(x)\right)=1\right]\geq \frac{\epsilon^k}{4}-\negl(\secpar).
\end{align*}
\end{lemma}
Remark that 
in the above lemma, 
$\siml$ is given the oracle $V^*$, which is the random-aborting verifier defined in \cref{sec:strict_poly}, rather than $\widetilde{V}^*$ 

\begin{proof}[Proof of \cref{lem:truncated_simulator}]
Let $\sigma_{x,\epsilon}$ be as in Lemma \ref{lem:final_state_sim}.
By $\TD(\sigma_{x,\epsilon},\ket{\phi_\epsilon}\bra{\phi_\epsilon})=\negl(\secpar)$ as shown in Lemma \ref{lem:final_state_sim}, if we measure $\regcont$ of $\sigma_{x,\epsilon}$, then the outcome is $1$ with probability $\frac{\epsilon^k}{1+\epsilon^k}\pm \negl(\secpar)$.  
Combining this with Lemma \ref{lem:prob_success_and_querysmall}, we have 
\begin{align}
\Pr_{\experiment'(x,\epsilon)}[\querysmall \land \Bisone \land \Contisone]\geq \frac{\epsilon^k}{4(1+\epsilon^k)}-\negl(\secpar)\geq \frac{\epsilon^k}{8}-\negl(\secpar) \label{eq:smallquery_and_Bisone_and_Contisone}
\end{align}
where $\experiment'(x,\epsilon)$ is the same as $\experiment(x,\epsilon)$ except that $\regcont$ is also measured at the end, and $\Contisone$ is the event that the measurement outcome of $\regcont$ is $1$.
Since $\siml$ works similarly to $\siml_{\mathsf{exp}}$ when $\querysmall$ occurs, we have 
\begin{align*}
\Pr\left[M_\regB\circ\OUT_{\ver^*}\left(\siml^{\ver^*(x;\ket{\psi_\epsilon})}(x)\right)=1\right]
\geq  \Pr[M_\regB\circ\siml_{\mathsf{exp}}^{\ver^*(x;\ket{\psi_\epsilon})}(x)=1 \land \querysmall].
\end{align*}
By Observation \ref{ob:measure_control_bit}, there is no difference if we measure $\regcont$ at the beginning of the experiment instead of at the end, and when the measurement outcome of $\regcont$ is $1$, $\widetilde{V}^*$ with auxiliary input $\ket{\widetilde{\psi}_\epsilon}$  works similarly to $V^*$ with auxiliary input $\ket{\psi_\epsilon}$. 
Therefore, we have  
\begin{align*}
\Pr[M_\regB\circ\siml_{\mathsf{exp}}^{\ver^*(x;\ket{\psi_\epsilon})}(x)=1 \land \querysmall]&= \Pr_{\experiment'(x,\epsilon)}[\querysmall \land \Bisone|\Contisone] \\   
&=\frac{\Pr_{\experiment'(x,\epsilon)}[\querysmall \land \Bisone \land \Contisone]}{\Pr_{\experiment'(x,\epsilon)}[\Contisone]}\\
&\geq  \frac{\epsilon^k}{4}-\negl(\secpar)
\end{align*}
where  the last inequality follows from 
Eq. \ref{eq:smallquery_and_Bisone_and_Contisone} and $\Pr_{\experiment'(x,\epsilon)}[\Contisone]=1/2$, which is easy to see.
\end{proof}

Finally, we prove Theorem \ref{thm:main_impossibility_BB_QZK_inefficient} based on Lemma \ref{lem:truncated_simulator}.
\begin{proof}[Proof of Theorem \ref{thm:main_impossibility_BB_QZK_inefficient}]
Let $\epsilon^*_q$ be as in Lemma \ref{lem:impossible_strict_poly_simulation} and  $\epsilon:=\epsilon^*_q$.
Since $\siml$ makes at most $q$ queries, Lemma \ref{lem:impossible_strict_poly_simulation} and \ref{lem:truncated_simulator} immediately imply $L\in \BQP$.
\end{proof}

\subsection{Expected-Polynomial-Time Simulation for Efficient Verifier}\label{sec:impossibility_expected_poly_simulation_efficient}
In this section, we extend the proof of Theorem \ref{thm:main_impossibility_BB_QZK_inefficient} in \cref{sec:impossibility_expected_poly_simulation_inefficient} to prove Theorem \ref{thm:main_impossibility_BB_QZK}.
That is, we prove that black-box simulation is impossible even for \emph{QPT} malicious verifiers.

The reason why the malicious verifier $\widetilde{V}^*$ in \cref{sec:impossibility_expected_poly_simulation_inefficient} is inefficient is that it has to apply the unitary $U_{\vmes}$ given in Lemma \ref{lem:uncomputing_random}, which works over 
an exponential-qubit register $\regH$.
Therefore, if we have an analogue of Lemma \ref{lem:uncomputing_random} for a family of efficiently computable functions that is indistinguishable from a random function taken from $\mathcal{H}_\epsilon$ by at most $Q$  quantum queries.
We prove such a lemma in the following.

\begin{lemma}\label{lem:uncomputing_for_twoQwise}
Let $\epsilon\in [0,1]$ be a rational number expressed as $\epsilon=\frac{B}{A}$ for some $A,B\in \mathbb{N}$ such that $\log A=\poly(\secpar)$ and $\log B=\poly(\secpar)$.\footnote{Note that $\epsilon$ is also a function of $\secpar$, but we omit to explicitly write the dependence on $\epsilon$ for simplicity.} 
For any  $Q=\poly(\secpar)$,  there exists a family $\tildehashfamily_\epsilon=\{\tildeH_\kappa:\messpace^{\leq k}\ra \bit\}_{\kappa\in \keyspace}$ of classical polynomial-time computable functions that satisfies the following properties.
\begin{enumerate}
    \item For any algorithm $\A$ that 
    makes at most $Q$ quantum queries and any quantum input $\rho$, we have 
    \[
    \Pr_{H\sample \mathcal{H}_\epsilon}\left[\A^{H}(\rho)=1\right]
    =
      \Pr_{\kappa\sample \keyspace}\left[\A^{\tildeH_\kappa}(\rho)=1\right].
    \]
    \item 
    For any $\vmes=(\mes_1,...,\mes_k)\in \messpace^k$, let $S_{\vmes}\subseteq \keyspace$ be the subset consisting of all $\kappa$ such that $\tildeH_\kappa(\mes_1,...,\mes_i)=1$ for all $i\in[k]$.  
    There exists a  unitary $U_{\vmes}^{(Q)}$ such that
    \[
    U_{\vmes}^{(Q)}\sqrt{\frac{1}{|\keyspace|}}\sum_{\kappa\in \keyspace}\ket{\kappa}=\sqrt{\frac{1}{|S_{\vmes}|}}\sum_{\kappa\in S_{\vmes}}\ket{\kappa}.
    \]
    $U_{\vmes}$ can be implemented by a quantum circuit of size $\poly(\secpar)$. 
\end{enumerate}
\end{lemma}
\begin{proof}
Let $\hashfamily'_{2Q,A}=\{H'_{\kappa'}:\messpace^{\leq k}\ra [A]\}_{\kappa'\in \keyspace'}$ be a $2Q$-wise independent hash family from $\messpace^{\leq k}$ to $[A]$. 
By \cref{lem:simulation_QRO}, 
for any $\A$ that 
    makes at most $Q$ quantum queries and any quantum input $\rho$, we have 
\begin{align}
    \Pr_{H'\sample \func(\messpace^{\leq k},[A])}\left[\A^{H'}(\rho)=1\right]
    =
      \Pr_{\kappa'\sample \keyspace'}\left[\A^{H'_{\kappa'}}(\rho)=1\right]. \label{eq:twoQwise}
\end{align}
We define $\tildehashfamily_\epsilon=\{\tildeH_\kappa:\messpace^{\leq k}\ra \bit\}_{\kappa\in \keyspace}$ as follows.
\begin{itemize}
    \item $\keyspace:=\keyspace'\times [A]^{k}$. In other words, a key $\kappa'$ for $\mathcal{H}'_{2 Q, A}$ is sampled, and $k$ additional random additive terms $\{a_i\}$ for each input length is sampled, as explained below. 
    \item For 
    $\kappa=(\kappa',\{a_i\}_{i\in[K]})$ and
    $(\mes_1,...,\mes_i)\in \messpace^{\leq k}$, we define
    \[
    \tildeH_{\kappa}(\mes_1,...,\mes_i):=
    \begin{cases}
    1    &\text{~if~}(H'_{\kappa'}(\mes_1,...,\mes_i)+a_i \mod A)\leq B\\
    0 &\text{~otherwise~}
    \end{cases}.
    \]
\end{itemize}
In the above definition, if we use a uniformly random function $H'\sample  \func(\messpace^{\leq k},[A])$ instead of $H'_{\kappa'}$, $\tildeH_{\kappa}$ is distributed according to $\mathcal{H}_\epsilon$.
Therefore, Eq. \ref{eq:twoQwise} implies the first item of Lemma \ref{lem:uncomputing_for_twoQwise}.

For proving the second item, we consider unitaries
$U_{\mathsf{add},\vmes,i}$, $U_{\leq A}$, and $U_{\leq B}$ that satisfy the following.
\begin{itemize}
    \item For any 
    $\vmes=(\mes_1,...,\mes_k)\in \messpace^k$,
    $i\in [k]$, and 
    $(\kappa',a_1,...,a_k)\in \keyspace$, we have 
\begin{align*}
    U_{\mathsf{add},\vmes,i}\ket{\kappa',a_1,...,a_k}=\ket{\kappa',a_1,...,a_{i-1},(H'_{\kappa'}(\mes_1,...,\mes_i)+a_i \mod A),a_{i+1},...,a_k}.
\end{align*}
\item For any $\kappa'\in \keyspace'$, we have 
\begin{align*}
 U_{\leq A}\ket{\kappa',0,...,0}=\sqrt{\frac{1}{A^k}}\sum_{(a_1,...,a_k)\in [A]^k}\ket{\kappa',a_1,...,a_k}.
\end{align*}
\item For any $\kappa'\in \keyspace'$, we have 
\begin{align*}
    U_{\leq B}\ket{\kappa',0,...,0}=\sqrt{\frac{1}{B^k}}\sum_{(a_1,...,a_k)\in [B]^k}\ket{\kappa',a_1,...,a_k}.
\end{align*}
\end{itemize}
We can see that such unitaries exist and are implementable by polynomial-size quantum circuits.

For any $\vmes=(\mes_1,...,\mes_k)\in \messpace^{\leq k}$, we define $U_{\vmes}$ as 
\begin{align*}
    U_{\vmes}^{(Q)}:=\left(U_{\leq A}U_{\leq B}^{\dagger}\prod_{i=1}^{k}U_{\mathsf{add},\vmes,i}\right)^\dagger.
\end{align*}
Then $U_{\vmes}^{(Q)}$ is implementable by a polynomial-size quantum circuit.  
For any 
$\kappa'\in \keyspace'$ and 
$\vmes=(\mes_1,...,\mes_k)$, we let $T_{\kappa',\vmes}\subseteq [A]^k$ be the subset consisting of all $(a_1,...,a_k)$ such that we have $(H'_{\kappa'}(\mes_1,...,\mes_i)+a_i \mod A)\leq B$ for all $i\in[k]$. 
Then we have 
\begin{align*}
    {U_{\vmes}^{(Q)}}^\dagger\sqrt{\frac{1}{|S_{\vmes}|}}\sum_{\kappa\in S_{\vmes}}\ket{\kappa}   &={U_{\vmes}^{(Q)}}^\dagger\sqrt{\frac{1}{|\keyspace'|\cdot B^k}}\sum_{\kappa'\in \keyspace'}\ket{\kappa'}\sum_{
    (a_1,...,a_i)\in T_{\kappa',\vmes}}
    \ket{a_1,...,a_i}\\
    &=U_{\leq A}U_{\leq B}^{\dagger}\sqrt{\frac{1}{|\keyspace'|\cdot B^k}}\sum_{\kappa'\in \keyspace'}\ket{\kappa'}\sum_{
    (a_1,...,a_i)\in [B]^k}
    \ket{a_1,...,a_i}\\
    &=U_{\leq A}\sqrt{\frac{1}{|\keyspace'|}}\sum_{\kappa'\in \keyspace'}\ket{\kappa'}\ket{0,...,0}\\
    &=\sqrt{\frac{1}{|\keyspace'|\cdot A^k}}\sum_{\kappa'\in \keyspace'}\ket{\kappa'}\sum_{
    (a_1,...,a_i)\in [A]^k}
    \ket{a_1,...,a_i}\\
    &=\sqrt{\frac{1}{|\keyspace|}}\sum_{\kappa\in \keyspace}\ket{\kappa}.
\end{align*}
By applying $U_{\vmes}^{(Q)}$ on both sides, the second item of Lemma \ref{lem:uncomputing_for_twoQwise} follows.
\end{proof}

With Lemma \ref{lem:uncomputing_for_twoQwise} in hand, we can prove Theorem \ref{thm:main_impossibility_BB_QZK} similarly to the proof of Theorem \ref{thm:main_impossibility_BB_QZK_inefficient} except that we consider the efficient version of $\widetilde{V}^*$ using Lemma \ref{lem:uncomputing_for_twoQwise}.
Here, we remark that we only need to take sufficiently small yet noticeable $\epsilon$ (depending on $q$) in the proof of Theorem \ref{thm:main_impossibility_BB_QZK_inefficient}, and thus we can assume that $\epsilon$ satisfies the assumption of Lemma \ref{lem:uncomputing_for_twoQwise} without loss of generality. 
Specifically, we modify the auxiliary input to  
\[
\ket{\widetilde{\psi}_\epsilon^{(q)}}_{\regcont,\regR,\regH}:=
\frac{1}{\sqrt{2}}\left(\ket{0}_\regcont+\ket{1}_\regcont\right)\otimes \ket{\psi_\epsilon^{(q)}}_{\regR,\regH}
\]
where 
\begin{align*}
 \ket{\psi_\epsilon^{(q)}}_{\regR,\regH}:=\sum_{r\in \mathcal{R},\kappa\in \keyspace} \sqrt{\frac{1}{|\mathcal{R}|\cdot|\keyspace|}}\ket{r,\kappa}_{\regR,\regH}
\end{align*}
and modify $\widetilde{V^*}$ to apply $H_{\kappa}$ whenever it applies $H$ and apply $U_{\vmes}^{(Q)}$ instead of $U_{\vmes}$ for the case of $\regcont=0$ and $\regcount=k-1$ (i.e., in the final round)
where $Q:=2kq$ 
In this setting, we can prove an analogue of Lemma \ref{lem:final_state_real} by using the second item of Lemma \ref{lem:uncomputing_for_twoQwise} instead of Lemma \ref{lem:uncomputing_random}.
We can prove analogues of Lemma
\ref{lem:prob_success_and_querysmall}, 
\ref{lem:final_state_sim}, 
and \ref{lem:truncated_simulator} in the exactly the same way to the original ones since these proofs do not use anything on the state in $\regH$. 
Finally, when we can prove 
Theorem \ref{thm:main_impossibility_BB_QZK} by Lemma \ref{lem:impossible_strict_poly_simulation} and
the analogue of Lemma \ref{lem:truncated_simulator} noting that $\siml$ that makes at most $q$ queries to the verifier can be seen as an algorithm that makes at most $Q=2kq$ quantum queries to $H$ by Observations \ref{ob:simulatable} and \ref{ob:F_to_H}, and thus 
it cannot distinguish $H_\kappa$ for $\kappa \sample \keyspace$ and $H\sample \mathcal{H}_\epsilon$ by the first item of  Lemma \ref{lem:uncomputing_for_twoQwise}. 

\paragraph{Malicious verifier with fixed polynomial-time.}
In the proof of Theorem \ref{thm:main_impossibility_BB_QZK}, we consider a malicious verifier whose running time depends on $q$, which is twice of the simulator's expected number of queries. 
Though this is sufficient for proving the impossibility of quantum black-box simulation (since Definition \ref{def:post_quantum_ZK} requires a simulator to work for all QPT verifiers whose running time may be an arbitrarily large polynomial), one may think that it is ``unfair" that the running time of the malicious verifier is larger than that of the simulator.  
We can resolve this issue if we use a quantumly-accessible PRF, which exists under the existence of post-quantum one-way functions \cite{FOCS:zhandry12}.\footnote{A similar idea is used to resolve a similar issue in the context of the impossibility of strict-polynomial-time simulation in the classical setting \cite{STOC:BarLin02}.} 
Specifically, if we use a quantumly-accessible PRF instead of $4q$-wise independent function in Lemma \ref{lem:uncomputing_for_twoQwise}, we can prove a similar lemma with a unitary whose size does not depend on $q$ instead of $U_{\vmes}^{(Q)}$, which naturally yields a malicious verifier whose running time does not depend on $q$. 
 \section{Impossibility of BB $\epsilon$-ZK for Constant-Round Public-Coin Arguments}
In this section, we prove the following theorem.
\begin{theorem}\label{thm:impossibility_BB_public_coin_QZK}
If there exists a constant-round public-coin post-quantum black-box  $\epsilon$-zero-knowledge argument for a language $L$, then $L\in \mathbf{BQP}$.
\end{theorem}

First, we define notations that are used throughout this section. 
Let $\Pi=(P,V)$ be a classical constant-round public-coin interactive argument for a language $\lang$. 
Without loss of generality, we assume that
$P$ sends the first message, and let $(P,V)$ be $(2k-1)$-round protocol where $P$ sends $k=O(1)$ messages in the protocol and $V$ sends $k-1$ messages (which are public-coin). 
We use $m_i$ to mean $P$'s $i$-th message and $c_i$ to mean $V$'s $i$-th message.
We assume that all  messages sent from  $P$ are elements of a classical set $\messpace$ and $V$'s messages are uniformly chosen from a classical set $\chalspace$. 
Without loss of generality, we assume that $\chalspace$ is a subset of $\messpace$. (If it is not the case, we can simply augment $\messpace$ to include $\chalspace$.)
For any $\vchal=(c_1,...,c_{k-1})\in \chalspace^{k-1}$, 
we denote by $\Acc[x,\vchal]\subseteq \messpace^k$ the subset consisting of  all $(\mes_1,...,\mes_k)\in \messpace^k$ such that $(\mes_1,c_1,...,\mes_{k-1},c_{k-1},m_k)$ is an accepting transcript. (Note that this is well-defined as we assume that $\Pi$ is public-coin). 
For any $q$, let $\mathcal{H}_{4(k-1)q}$ be a family of $4(k-1)q$-wise independent hash functions from $\messpace^{\leq (k-1)}$ to $\chalspace$.

 A high level structure of the rest of this section is similar to \cref{sec:strict_poly}.
To prove \cref{thm:impossibility_BB_public_coin_QZK}, we consider a malicious verifier that derives its messages by applying a random function on the transcript so far.  
More precisely, 
for any $q=\poly(\secpar)$, 
we consider a malicious verifier $V^*$ with an auxiliary input $\ket{\psi_q}_{\regH}$ as follows.

\begin{description}
\item $\ket{\psi_q}_{\regH}$ is the uniform superposition over $H\in \mathcal{H}_{4(k-1)q}$.


\item $V^*$ works over its internal register $(\regX,\regaux=\regH, \regW = ( \regcount ,\regM_1,...,\regM_k, \regB))$ and an additional message register $\regM$. 
We define the output register as $\regout:=\regB$.  
It works as follows where $\regcount$ stores a non-negative integer smaller than  $k$ (i.e. $\{0, 1, \cdots, k - 1\}$), each register of $\regM_1,...,\regM_k$ and $\regM$ stores an element of $\messpace$ and $\regB$ stores a single bit. $\regM$ is the register to store messages from/to external prover, and $\regM_i$ is a register 
to record the $i$-th message from the prover.


We next explain the unitary $U^*$ for $V^*$. 


\begin{enumerate}
\item $V^*$ takes inputs a statement $x$  and a quantum auxiliary input $\ket {\psi}$: $\regX$ is initialized to be $\ket{x}_{\regX}$ where $x$ is the statement to be proven, $\regaux$ is initialized to be $\ket{\psi}_{\regH}$,
and all other registers are initialized to be $0$.
    \item 
    \emph{Verifier $V^*$ on round $< k$}: Upon receiving the $i$-th message from $P$ for $i<k$ in $\regM$, swap $\regM$ and  $\regM_i$ and increment the value in $\regcount$.  
    We note that $V^*$ can know $i$ since it keeps track of which round it is playing by the value in $\regcount$.  
    Let $(\mes_1,...,\mes_i)$ be the messages sent from $P$ so far.
    Then it returns $H(m_1,...,m_i)\in \chalspace$ to $P$ as its next message.
    
\vspace{1em}
    
    \emph{Unitary $U^*$ on $\regcount < k - 1$}: It acts on registers $\regcount$ (whose value is less than $k - 1$), $\regX, (\regM_1,...,\regM_k)$,
     $\regH$ and $\regM$: 
        \begin{itemize}
            \item It reads the value $j$ in $\regcount$ and increments it to $i:=j+1$. It swaps $\regM$ and $\regM_{i}$ (in superposition).  
\item It applies the following unitary.\footnote{Here, $H(m_1, \cdots, m_i)$ is seen as an element of $\messpace$ by using our assumption $\chalspace \subseteq \messpace$.} 
\begin{align*}
    \ket {x, m_1, \cdots, m_i, H, m}  \to \ket {x, m_1, \cdots, m_i, H, m + H(m_1, \cdots, m_i)}. 
\end{align*}
\end{itemize}

    \item 
    \emph{Verifier $V^*$ on round $k$}: Upon receiving the $k$-th message from $P$ in $\regM$,  swap $\regM$ and $\regM_k$  and increment the value in $\regcount$. 
    Then flip the bit in $\regB$ if 
    $(m_1,...,m_k)\in \Acc[x,(c_1,...,c_{k-1})]$  
    where 
    $c_i:=H(\mes_1,...,\mes_i)$ for all $i\in [k-1]$ and 
    $(\mes_1,...,\mes_k)$ and $H$ are values in registers 
    $(\regM_1,...,\regM_k)$ and $\regH$.

\vspace{1em}
    
\emph{Unitary $U^*$ on $\regcount= k - 1$}: It acts on registers $\regcount$ (whose value is exactly equal to $k-1$), $\regX, (\regM_1,...,\regM_k)$,
    $\regH$ and $\regB$: 
        \begin{itemize}
            \item  It reads the value $i = k-1$ in $\regcount$ and sets it to $0$. It swaps $\regM$ and $\regM_{k}$ (in superposition). 
            \item Let $x, (m_1, \cdots, m_k), H$ and $b$ be the values in registers $\regX, (\regM_1, \cdots, \regM_k), \regH$ and $\regB$. 
            Let $F^*[x, H]$ be the following function: on input $(m_1, \cdots, m_k) \in \messpace^k$, 
\begin{align*}
F^*[x,H](\mes_1,...,\mes_k):=
\begin{cases}
1 & \text{~if~} (m_1,...,m_k)\in \Acc[x,(c_1,...,c_{k-1})] \\
0 &\text{~otherwise~} 
\end{cases}
\end{align*}
where $c_i:=H(m_1,...,m_i)$ for $i\in[k-1]$.

It applies the function in superposition. 
\begin{align*}
    \ket {x, m_1, \cdots, m_k, r, H, b}  \to \ket {x, m_1, \cdots, m_k, r, H, b + F^*[x, H](m_1, \cdots, m_k)}. 
\end{align*}
\end{itemize}
\end{enumerate}
\end{description}

With the description of $V^*$ above, we have the following observation. 
\begin{observation} \label{ob:measure_aux_input_pc}
    Let $M_{\regH}$ be the measurement on the register $\regH$. 
    For any (inefficient) black-box simulator $\siml$ it has zero advantage of distinguishing if it has black-box access to $V^*(x; \ket {\psi}_{\regH})$ or $V^*(x;  M_{ \regH} \circ \ket {\psi}_{\regH})$. 
\end{observation}
Observation \ref{ob:measure_aux_input_pc} says that even for an unbounded simulator with black-box access to $V^*(x;  \ket {\psi_{q}}_{\regH})$, it has no way to tell if the auxiliary input $ \ket {\psi_{q}}_{\regH}$ gets measured at the beginning or never gets measured.  
This is because the register $\regH$ is only used as control qubits throughout the execution of $\siml^{\ver^*(x;\ket{\psi_q})}(x)$, we can trace out the register $\regH$ while preserving the behavior of this simulator.

\begin{observation} \label{ob:simulatable_pc}
An oracle that applies $U^*$ can be simulated by $2(k-1)$ quantum oracle access to $H$.
\end{observation}
This can be easily seen from the description of $U^*$ above. 
Note that we require $2(k-1)$ queries instead of $k-1$ queries since we need to compute $k-1$ values of $H$ to compute $F^*[x,r,H]$ 
(when the input is in $\messpace^k$)
and then need to uncompute them.

\vspace{1em}

We then prove the following lemma.  
\begin{lemma}\label{lem:impossible_constant_round_pc}
If there exists a quantum black-box simulator $\siml$ that makes at most $q=\poly(\secpar)$ queries such that 
\[
\Pr\left[M_\regB\circ\OUT_{\ver^*}\left(\siml^{\ver^*(x;\ket{\psi_q})}(x)\right)=1\right]\geq \frac{1}{\poly(\secpar)}
\]
for all $x\in L \cap \bit^\secpar$ where $M_\regB$ means measuring and outputting the register $M_\regB$,
then we have $L\in \BQP$.
\end{lemma}
The above lemma immediately implies \cref{thm:impossibility_BB_public_coin_QZK}.
\begin{proof}[Proof of \cref{thm:impossibility_BB_public_coin_QZK}]
Let $\Pi = (P, V)$ be a constant-round public-coin post-quantum black-box $\epsilon$-zero-knowledge argument for a language $L$ where $P$ sends $k=O(1)$ messages and $V^*$ and $\ket{\psi_q}$ are defined above. 
By definition, 
for any noticeable $\epsilon$, 
there exists a quantum black-box simulator $\siml$  for $\Pi$ that makes at most $q=\poly(\secpar)$ queries such that 
\[
\{\OUT_{\ver^*}\execution{\pro(w)}{\ver^*(\ket{\psi_q})}(x)\}_{\secpar,x,w}
\compind_\epsilon
\{\OUT_{\ver^*}(\siml^{\ver^*(x;\ket{\psi_q})}(x))\}_{\secpar,x,w}
\]
where $\secpar \in \mathbb{N}$, $x\in \lang\cap \bit^\secpar$, and $w\in \rel_\lang(\secpar)$. 
(Note that we can assume that the number of queries by the simulator is strict-polynomial as explained in \cref{rem:strict_poly_epsilon_ZK}.)
Especially, we take $\epsilon:=1/2$.
By completeness of $\Pi$ and the definitions of $V^*$ and $\ket{\psi_q}$, for any $x\in L \cap \bit^{\secpar}$ and its witness $w\in R_L(x)$,  we have 
\[
\Pr[M_{\regB}\circ \OUT_{\ver^*}\left(\execution{P(w)}{V^*(\ket{\psi_q})}(x)\right)=1]\revise{\geq 1-\negl(\secpar)}.
\]
By combining the above, we have 
\[
\Pr\left[ M_\regB\circ\OUT_{\ver^*}\left(\siml^{\ver^*(x;\ket{\psi_q})}(x)\right)=1\right]\revise{\geq 1-\negl(\secpar)-\epsilon=\frac{1}{2}-\negl(\secpar)}> \frac{1}{\poly(\secpar)}.
\]
By Lemma \ref{lem:impossible_constant_round_pc}, this implies $L \in \mathbf {BQP}$.
\end{proof}
\begin{remark}
As one can see from the above proof, we can actually prove a stronger statement than \cref{thm:impossibility_BB_public_coin_QZK}. 
That is, even a black-box simulation with approximation error as large as $1-\frac{1}{\poly(\secpar)}$ is still impossible for a language outside $\BQP$.
\end{remark}

Then we prove Lemma \ref{lem:impossible_constant_round_pc}. 

\begin{proof}[Proof of Lemma \ref{lem:impossible_constant_round_pc}]
By Observation \ref{ob:measure_aux_input_pc}, we can assume $\ket {\psi_q}$ is measured at the beginning. In other words, the auxiliary state is sampled as $H\sample \mathcal{H}_{4(k-1)q}$.
Once $H$ is fixed, the unitary $U^*$ (corresponding to $V^*$) and its inverse can be simulated by $2(k-1)$ quantum access to a classical function
$H$ as observed in Observation \ref{ob:simulatable_pc}. Moreover, since the simulator makes at most $q$ queries, by \cref{lem:simulation_QRO}, the simulator's behavior does not change even if  $H$ is uniformly sampled from $\func(\messpace^{\leq (k-1)},\chalspace)$. 

We let $\Acc^*[x,H]\subseteq \messpace^k$ be the set of $\vmes=(\mes_1,...,\mes_k)$ such that 
$(\mes_1,...,\mes_k)\in \Acc[x,(c_1,...,c_{k-1})]$ where $c_i:=H(\mes_1,...,\mes_i)$ for $i\in [k-1]$. 
After the execution of $\siml^{\ver^*(x;\ket H)}(x)$,
$\regB$ contains $1$ if and only if $(\regM_1,...,\regM_k)$ contains an element in $\Acc^*[x,H]$. 
Therefore, for proving Lemma \ref{lem:impossible_constant_round_pc}, it suffices to prove the following lemma.
\begin{lemma}\label{lem:impossible_constant_round_pc_classical}
If there exists an 
 oracle-aided quantum algorithm $\mathcal{S}$ that makes at most $\poly(\secpar)$ quantum queries such that  
\[
\Pr_{H\sample \func(\messpace^{\leq (k-1)},\chalspace)}\left[\mathcal{S}^{H}(x)\in \Acc^*[x,H]\right]\geq \frac{1}{\poly(\secpar)}
\]
for all $x\in L \cap \bit^\secpar$,
then we have $L\in \BQP$.
\end{lemma} 

We prove the above lemma below. Assuming Lemma \ref{lem:impossible_constant_round_pc_classical}, we show Lemma \ref{lem:impossible_constant_round_pc} holds. Let $\siml$ be a quantum black-box simulator that makes at most $q$ quantum queries, such that 
\[
\Pr\left[M_\regB\circ\OUT_{\ver^*}\left(\siml^{\ver^*(x;\ket{\psi_q})}(x)\right)=1\right]\geq \frac{1}{\poly(\secpar)}.
\]
By Observation \ref{ob:measure_aux_input_pc} and \cref{lem:simulation_QRO}, we have 
\begin{align*}
& \Pr_{H\sample \func(\messpace^{\leq (k-1)}, \chalspace)}\left[M_\regB\circ\OUT_{\ver^*}\left(\siml^{\ver^*(x;{\ket {H} })}(x)\right)=1\right] \\
= & \Pr\left[M_\regB\circ\OUT_{\ver^*}\left(\siml^{\ver^*(x;\ket{\psi_q})}(x)\right)=1\right]\geq \frac{1}{\poly(\secpar)}. 
\end{align*}

Finally, we note that $M_\regB\circ\OUT_{\ver^*}\left(\siml^{\ver^*(x;{\ket {r, H} })}(x)\right)$ can be computed by only having black-box access to $H$ (by Observation \ref{ob:simulatable_pc}). It outputs $1$ (the register $\regB$ is $1$) if and only if the values $(m_1, \cdots, m_k)$ in $\regM_1, \cdots, \regM_k$ are in $\Acc^*[x, H]$. Thus, there is an algorithm $\mathcal{S}$ that computes $M_\regB\circ\siml^{\ver^*(x;{\ket {H} })}(x)$ and measures registers $\regM_1, \cdots, \regM_k$. Such an algorithm $\mathcal{S}$ satisfies the requirement in Lemma \ref{lem:impossible_constant_round_pc_classical}. Therefore $L$ is in $\mathbf{BQP}$.
\end{proof}

The remaining part is to prove Lemma \ref{lem:impossible_constant_round_pc_classical}.

\begin{proof}[Proof of Lemma \ref{lem:impossible_constant_round_pc_classical}]
As shown in \cite{C:DonFehMaj20}, Fiat-Shamir transform for constant-round public-coin protocol preserves the soundness up to a polynomial security loss in the quantum random oracle model where an adversary may quantumly query the random oracle. 
Let $\Pi_{\mathsf{ni}}$ be the non-interactive protocol that is obtained by applying Fiat-Shamir transform to $\Pi$.  
More concretely, $\Pi_{\mathsf{ni}}$ is a non-interactive argument in the random oracle model where a proof is of the form $(\mes_1,...,\mes_k)$ and it is accepted for a statement $x$ if and only if $(\mes_1,...,\mes_k)\in \Acc^*[x,H]$.
Since we assume that $\Pi$ has a negligible soundness error, so does $\Pi_{\mathsf{ni}}$ by the result of \cite{C:DonFehMaj20}.\footnote{Though the way of applying the Fiat-Shamir is slightly different from that in \cite{C:DonFehMaj20} (where they derive a challenge by nested applications of the random oracle), their proof still works for the above way as noted in \cite[Remark 12]{C:DonFehMaj20}} 
This means that we have 
\[
\Pr_{H\sample \func(\messpace^{\leq (k-1)},\chalspace)}\left[\mathcal{S}^{H}(x)\in \Acc^*[x,H]\right]\leq \negl(\secpar)
\]
for all $x\in  \bit^\secpar\setminus L$ since otherwise we can use $\mathcal{S}$ to break the soundness of $\Pi_{\mathsf{ni}}$.
Therefore, we can use $\mathcal{S}$ to decide $L$ by simulating a random function $H$ by itself, which implies $L\in \mathsf{BQP}$. (Note that one can efficiently simulate a random function for quantum-query adversaries that makes at most $2(k-1)q$ queries by using $4(k-1)q$-wise independent hash function by \cref{lem:simulation_QRO}.)
\end{proof}

 \section{Impossibility of BB $\epsilon$-ZK for Three-Round Arguments}\label{sec:impossible_three}
In this section, we prove the following theorem.
\begin{theorem}\label{thm:impossibility_BB_QZK_three}
If there exists a three-round post-quantum black-box  $\epsilon$-zero-knowledge argument for a language $L$, then $L\in \mathbf{BQP}$.
\end{theorem}

First, we define notations that are used throughout this section. 
Let $\Pi=(P,V)$ be a classical three-round interactive argument for a language $\lang$. 
We assume that all  messages sent between  $P$ and $V$ are elements of a classical set $\messpace$ (e.g., we can take $\messpace:= \bit^\ell$ for sufficiently large $\ell$).  
Let $\randspace$ be $V$'s randomness space.
For any fixed statement $x$ and randomness $r\in\randspace$, $V$'s message in the second round can be seen as a deterministic function of the prover's first message $\mes_1$. 
We denote this function by $F[x,r]:\messpace \rightarrow \messpace$.
We denote by $\Acc[x,r]\subseteq \messpace^2$ the subset consisting of  all $(\mes_1,\mes_2)$ such that the verifier accepts when its randomness is $r$ and prover's messages in the first and third rounds are $\mes_1$ and $\mes_2$, respectively. 
For any $q$, let $\mathcal{H}_{4q}$ be a family of $4q$-wise independent hash functions from $\messpace$ to $\randspace$. 
 
A high level structure of the rest of this section is similar to \cref{sec:strict_poly}.
To prove \cref{thm:impossibility_BB_QZK_three}, we consider a malicious verifier that derives its randomness by applying a random function on the prover's first message.  
More precisely, 
for any $q=\poly(\secpar)$, 
we consider a malicious verifier $V^*$ with an auxiliary input $\ket{\psi_q}_{\regH}$ as follows.  

\begin{description}
\item $\ket{\psi_q}_{\regH}$ is the uniform superposition over $H\in\mathcal{H}_{4q}$. 

\item $V^*$ works over its internal register $(\regX,\regaux=\regH, \regW = ( \regcount ,\regM_1,\regM_2, \regB))$ and an additional message register $\regM$. 
We define the output register as $\regout:=\regB$.  
It works as follows where $\regcount$ stores an integer $0$ or $1$,\footnote{We view them as an integer rather than a bit for the consistency to the description of $V^*$ in \cref{sec:strict_poly}.} each register of $\regM_1,\regM_2$, and $\regM$ stores an element of $\messpace$ and $\regB$ stores a single bit. $\regM$ is the register to store messages from/to a external prover, and $\regM_1$ and $\regM_2$ are  registers 
to record the prover's first and second messages, respectively.\footnote{Remark that the prover's second message is a message sent in the third round.}

We next explain the unitary $U^*$ for $V^*$. The interaction between $V^*$ and  the honest prover $P$ has been formally defined in \Cref{sec:prelim_interactive_proof}. We recall it here.  

\begin{enumerate}
\item $V^*$ takes inputs a statement $x$  and a quantum auxiliary input $\ket {\psi_q}$: $\regX$ is initialized to be $\ket{x}_{\regX}$ where $x$ is the statement to be proven, $\regaux$ is initialized to be $\ket{\psi_q}_{\regH}$,
and all other registers are initialized to be $0$.
    \item 
    \emph{Verifier $V^*$ on the first message}: Upon receiving the first message from $P$ in $\regM$, swap $\regM$ and  $\regM_1$ and increment the value in $\regcount$\footnote{This is a NOT gate that maps $\ket i$ to $\ket {(i + 1)\bmod 2}$. }.  
    We note that $V^*$ can know $i$ since it keeps track of which round it is playing by the value in $\regcount$.  
    Let $\mes_1$ be the first message sent from $P$.
    Then return $F[x,H(\mes_1)](\mes_1)$ to $P$.
    
\vspace{1em}
    
    \emph{Unitary $U^*$ on $\regcount=0$}: It acts on registers $\regcount$ (whose value is exactly equal to $0$), $\regX, (\regM_1,\regM_2)$,
    $\regH$ and $\regM$: 
        \begin{itemize}
            \item It reads the value $i=0$ in $\regcount$ and increments it to $1$. It swaps $\regM$ and $\regM_{1}$ (in superposition).   
            \item Let $x, m_1, r, H$ and $m$ be the values in registers $\regX, \regM_1, \regH$ and $\regM$. 
Let $F^*[x,H]$ be the following function: on input $m_1\in \messpace$, 
\[
F^*[x,H](m_1):=F[x,H(\mes_1)](m_1).
\]
It then applies the function  in superposition:
\begin{align*}
    \ket {x, m_1,  H, m}  \to \ket {x, m_1, H, m + F^*[x,H](m_1)}. 
\end{align*}
\end{itemize}

    \item 
    \emph{Verifier $V^*$ on the second message}: Upon receiving the second message from $P$ in $\regM$,  swap $\regM$ and $\regM_2$  and increment the value in $\regcount$. 
    Then flip the bit in $\regB$ if 
    $(\mes_1,\mes_2)\in \Acc[x,H(\mes_1)]$  
    where 
    $(\mes_1,\mes_2)$  
    and $H$ are values in registers 
    $(\regM_1,\regM_2)$,  
    and $\regH$.

\vspace{1em}
    
\emph{Unitary $U^*$ on $\regcount=1$}: It acts on registers $\regcount$ (whose value is exactly equal to $1$), $\regX, (\regM_1,\regM_2)$,
     $\regH$ and $\regB$: 
        \begin{itemize}
            \item  It reads the value $i = 1$ in $\regcount$ and sets it to $0$. It swaps $\regM$ and $\regM_{2}$ (in superposition). 
            \item Let $x, (m_1, m_2), H$ and $b$ be the values in registers $\regX, (\regM_1,\regM_2), \regH$ and $\regB$. Let $F_{\Acc^*[x, H]}$ be the following function: on input $(m_1,m_2) \in \messpace^2$, 
\begin{align*}
F_{\Acc^*[x, H]}(\mes_1,\mes_2):=
\begin{cases}
1 & \text{~if~} (\mes_1,\mes_2)\in \Acc^*[x,H]\\
0 &\text{~otherwise~} 
\end{cases}
\end{align*}
           where $\Acc^*[x,H]\subseteq \messpace^2$ is the set of all $(\mes_1,\mes_2)$ such that 
$(\mes_1,\mes_2)\in \Acc[x,H(\mes_1)]$. 
            
It applies the function in superposition. 
\begin{align*}
    \ket {x, m_1, m_2, H, b}  \to \ket {x, m_1, m_2,  H, b + F_{\Acc^*[x, H]}(m_1,  m_2)}. 
\end{align*}
\end{itemize}
\end{enumerate}
\end{description}

With the description of $V^*$ above, we have the following observation. 
\begin{observation} \label{ob:measure_aux_input_three}
    Let $M_{\regH}$ be the measurement on register $\regH$. 
    For any (inefficient) black-box simulator $\siml$ it has zero advantage of distinguishing if it has black-box access to $V^*(x; \ket {\psi}_{\regH})$ or $V^*(x;  M_{\regH} \circ \ket{\psi}_{\regH})$. 
\end{observation}

Observation \ref{ob:measure_aux_input_three} says that even for an unbounded simulator with black-box access to $V^*(x; \ket{\psi}_{\regH})$, it has no way to tell if the auxiliary input $\ket{\psi}_{\regH}$ gets measured at the beginning or never gets measured.  
This is because the register $\regH$ is only used as control qubits throughout the execution of $\siml^{\ver^*(x;\ket{\psi})}(x)$, we can trace out the register $\regH$ while preserving the behavior of this simulator. 

\begin{observation} \label{ob:simulatable_three}
$V^*$ can be simulated giving oracle access to $F^*[x,H]$ and $F_{\Acc^*[x, H]}$.
\end{observation}
This can be easily seen from the description of $U^*$ above. 

\begin{observation} \label{ob:F_to_H_three}
Given $x$ and $r$, a quantum oracle that computes 
$F^*[x, H]$ or $F_{\Acc^*[x, H]}$ can be simulated by $2$ quantum oracle access to $H$.
\end{observation}
This can be easily seen from the definitions of $F^*[x,H]$ and $F_{\Acc^*[x, H]}$.
Note that we require $2$ queries instead of $1$ query since we need to compute the value of $H$ to compute $F^*[x,H]$ or $F_{\Acc^*[x, H]}$ and then need to uncompute it.

\vspace{1em}

We then prove the following lemma.  
\begin{lemma}\label{lem:impossible_three_round}
If there exists a quantum black-box simulator $\siml$ that makes at most $q=\poly(\secpar)$ queries such that
\[
\Pr\left[M_\regB\circ\OUT_{\ver^*}\left(\siml^{\ver^*(x;\ket{\psi_q})}(x)\right)=1\right]\geq \frac{1}{\poly(\secpar)}
\]
for all $x\in L \cap \bit^\secpar$ where $M_\regB$ means measuring and outputting the register $M_\regB$,
then we have $L\in \BQP$.
\end{lemma}
The above lemma immediately implies \cref{thm:impossibility_BB_QZK_three}.
\begin{proof}[Proof of \cref{thm:impossibility_BB_QZK_three}]
This is exactly the same as the proof of \cref{thm:impossibility_BB_public_coin_QZK} based on \cref{lem:impossible_constant_round_pc}
Let $\Pi = (P, V)$ be a three-round post-quantum black-box $\epsilon$-zero-knowledge argument for a language $L$  and $V^*$ and $\ket{\psi_q}$ are defined above. 
By definition, 
for any noticeable $\epsilon$, 
there exists a quantum black-box simulator $\siml$  for $\Pi$ that makes at most $q=\poly(\secpar)$ queries such that 
\[
\{\OUT_{\ver^*}\execution{\pro(w)}{\ver^*(\ket{\psi_q})}(x)\}_{\secpar,x,w}
\compind_\epsilon
\{\OUT_{\ver^*}(\siml^{\ver^*(x;\ket{\psi_q})}(x))\}_{\secpar,x,w}
\]
where $\secpar \in \mathbb{N}$, $x\in \lang\cap \bit^\secpar$, and $w\in \rel_\lang(\secpar)$. 
(Note that we can assume that the number of queries by the simulator is strict-polynomial as explained in \cref{rem:strict_poly_epsilon_ZK}.)
Especially, we take $\epsilon:=1/2$.
By completeness of $\Pi$ and the definitions of $V^*$ and $\ket{\psi_q}$, for any $x\in L \cap \bit^{\secpar}$ and its witness $w\in R_L(x)$,  we have 
\[
\Pr[M_{\regB}\circ \OUT_{\ver^*}\left(\execution{P(w)}{V^*(\ket{\psi_q})}(x)\right)=1]\revise{\geq 1-\negl(\secpar)}.
\]
By combining the above, we have 
\[
\Pr\left[ M_\regB\circ\OUT_{\ver^*}\left(\siml^{\ver^*(x;\ket{\psi_q})}(x)\right)=1\right]\revise{\geq 1-\negl(\secpar)-\epsilon=\frac{1}{2}-\negl(\secpar)}> \frac{1}{\poly(\secpar)}.
\]
By Lemma \ref{lem:impossible_three_round}, this implies $L \in \mathbf {BQP}$.
\end{proof}
\begin{remark}
As one can see from the above proof, we can actually prove a stronger statement than \cref{thm:impossibility_BB_QZK_three}. 
That is, even a black-box simulation with approximation error as large as $1-\frac{1}{\poly(\secpar)}$ is still impossible for a language outside $\BQP$.
\end{remark}

Then we prove Lemma \ref{lem:impossible_three_round}.

\begin{proof}[Proof of Lemma \ref{lem:impossible_three_round}]
By Observation \ref{ob:measure_aux_input_three}, we can assume $\ket {\psi_q}$ is measured at the beginning. In other words, the auxiliary state is sampled as $\ket {H}_\regH$ for $H\sample \mathcal{H}_{4q}$. 
Once $H$ is fixed, the unitary $U^*$ (corresponding to $V^*$) and its inverse can be simulated by a single quantum access to a classical function
$F^*[x,H]$ or $F_{\Acc^*[x, H]}$ defined in the description of $V^*$ (Observation \ref{ob:simulatable_three}). Moreover, since the simulator makes at most $q$ queries to the verifier and a single query can be simulated by two queries to $H$ as observed in Observation \ref{ob:F_to_H_three},
the simulator can be seen as an oracle-aided algorithm that makes at most $2q$ quantum queries to $H$. 
Therefore, by \cref{lem:simulation_QRO}, the simulator's behavior does not change even if  $H$ is uniformly sampled from $\func(\messpace,\randspace)$. 
After the execution of $\siml^{\ver^*(x;\ket H)}(x)$,
$\regB$ contains $1$ if and only if $(\regM_1,\regM_2)$ contains an element in $\Acc^*[x,H]$. 

Therefore, for proving Lemma \ref{lem:impossible_three_round}, it suffices to prove the following lemma.
\begin{lemma}\label{lem:impossible_three_round_classical}
If there exists an
 oracle-aided quantum algorithm $\mathcal{S}$ that makes at most $\poly(\secpar)$ quantum queries such that
\[
\Pr_{H\sample \func(\messpace,\randspace)}\left[\mathcal{S}^{F^*[x,H],F_{\Acc^*[x,H]}}(x)\in \Acc^*[x,H]\right]\geq \frac{1}{\poly(\secpar)}
\]
for all $x\in L \cap \bit^\secpar$,
then we have $L\in \BQP$.
\end{lemma} 

We prove the above lemma below. Assuming Lemma \ref{lem:impossible_three_round_classical}, we show Lemma \ref{lem:impossible_three_round} holds. 
We note that this proof is similar to the proof of  Lemma \ref{lem:impossible_constant_round_pc} based on Lemma \ref{lem:impossible_constant_round_pc_classical}. 

Let $\siml$ be a quantum black-box simulator that makes at most $q$ quantum queries, such that 
\[
\Pr\left[M_\regB\circ\OUT_{\ver^*}\left(\siml^{\ver^*(x;\ket{\psi_q})}(x)\right)=1\right]\geq \frac{1}{\poly(\secpar)}.
\]
By Observation \ref{ob:measure_aux_input_three} and \cref{lem:simulation_QRO}, we have 
\begin{align*}
& \Pr_{H\sample \func(\messpace, \randspace)}\left[M_\regB\circ\OUT_{\ver^*}\left(\siml^{\ver^*(x;{\ket {H} })}(x)\right)=1\right] \\
= & \Pr\left[M_\regB\circ\OUT_{\ver^*}\left(\siml^{\ver^*(x;\ket{\psi_q})}(x)\right)=1\right]\geq \frac{1}{\poly(\secpar)}. 
\end{align*}

Finally, we note that $M_\regB\circ\OUT_{\ver^*}\left(\siml^{\ver^*(x;{\ket {r, H} })}(x)\right)$ can be computed by only having black-box access to $F^*[x,H]$ and $F_{\Acc^*[x,H]}$ (by Observation \ref{ob:simulatable_three}). It outputs $1$ (the register $\regB$ is $1$) if and only if the values $(m_1, m_2)$ in $\regM_1, \regM_2$ are in $\Acc^*[x, H]$. Thus, there is an algorithm $\mathcal{S}$ that computes $M_\regB\circ\siml^{\ver^*(x;{\ket {H} })}(x)$ and measures registers $\regM_1, \regM_2$. Such an algorithm $\mathcal{S}$ satisfies the requirement in Lemma \ref{lem:impossible_three_round_classical}. Therefore $L$ is in $\mathbf{BQP}$.
\end{proof}

For proving Lemma \ref{lem:impossible_three_round_classical} 
we reduce  it to a simplified lemma (Lemma \ref{lem:impossible_three_round_classical_simplified} below) where $\mathcal{S}$ is not given the oracle $F_{\Acc^*[x,H]}$.
\begin{lemma}\label{lem:impossible_three_round_classical_simplified}
If there exists an
 oracle-aided quantum algorithm $\mathcal{S}$ that makes at most $\poly(\secpar)$ quantum queries such that we have 
\[
\Pr_{H\sample \func(\messpace,\randspace)}\left[\mathcal{S}^{F^*[x,H]}(x)\in \Acc^*[x,H]\right]\geq \frac{1}{\poly(\secpar)}
\]
for all $x\in L \cap \bit^\secpar$,
then we have $L\in \BQP$.
\end{lemma}
We first prove  Lemma \ref{lem:impossible_three_round_classical} assuming  Lemma \ref{lem:impossible_three_round_classical_simplified} by using \cref{cor:o2h}.
\begin{proof}[Proof of Lemma \ref{lem:impossible_three_round_classical}]
Let $\mathcal{S}$ be an algorithm that satisfies the assumption of Lemma \ref{lem:impossible_three_round_classical}. 
We apply \cref{cor:o2h} by considering $\mathcal{S}$ as $\A$ in \cref{cor:o2h}.
Then we can see that the algorithm corresponding to $\mathcal{C}$ in \cref{cor:o2h} satisfies the assumption of Lemma \ref{lem:impossible_three_round_classical_simplified}, which implies $L\in \mathsf{BQP}$. (Note that though $\mathcal{S}$ has an additional oracle $F^*[x,H]$,   \cref{cor:o2h} is still applicable by considering an augmented algorithm $\mathcal{S}'$ that takes $H$ as part of its input and simulates $F^*[x,H]$ by itself.) 
\nai{I understand the proof here. However, I was confused a bit since the S in this Lemma seems to be different from the S in the previous Lemma. }
\takashi{It may be clearer if we use different names for the algorithm for these lemmas, but I have not good idea. (For example, writing $\mathcal{S}'$) may look ugly.}
\end{proof}

The remaining part is to prove Lemma \ref{lem:impossible_three_round_classical_simplified}.
\begin{proof}[Proof of Lemma \ref{lem:impossible_three_round_classical_simplified}]
In the following, we simply write $r$ and $H$ in subscripts of probabilities to mean  
$r\sample \randspace$ and 
$H\sample \func(\messpace,\randspace)$ for notational simplicity. 


We apply Lemma \ref{lem:measure_and_reprogram} to $\mathcal{S}^{F^*[x,H]}$ where 
$k:=1$, 
$\calX:= \messpace$, $\calY:=\messpace$, $\calZ:=\messpace$, and we define a relation $R\subseteq \calX \times \calY \times \calZ$ by 
$(\mes_1,\mes_V,\mes_2)\in R$ if and only if
\revise{$(\mes_1,\mes_V,\mes_2)\in \Acc'[x,r]$ where $\Acc'[x,r]$ is the set of all accepting transcripts w.r.t. the randomness $r$ 
(i.e., $(\mes_1,\mes_V,\mes_2)\in \Acc'[x,r]$ if and only if  $F[x,r](\mes_1)=\mes_V$ and $(\mes_1,\mes_2)\in \Acc[x,r]$).}
By Lemma \ref{lem:measure_and_reprogram}, for any $x,H,r$, and  
$\mes^*_1$ (where we set $m_V:=F[x,r](\mes^*_1)$), 
 we have 
\begin{align}
\begin{split}
&\Pr\left[
\mes'_1=\mes^*_1 \land (\mes'_1,F[x,r](\mes^*_1),\mes_2)\in \Acc'[x,\revise{r}]
:(\mes'_1,\mes_2)\sample \widetilde{\mathcal{S}}[F^*[x,H],F[x,r](\mes^*_1)](x)\right]\\
&\geq \frac{1}{(2q +1)^{2}}\Pr\left[
\mes_1=\mes^*_1 \land (\mes_1,F[x,r](\mes^*_1),\mes_2)\in \Acc'[x,\revise{r}]
:(\mes_1,\mes_2) \sample \mathcal{S}^{F^*[x,H]_{\mes^*_1,F[x,r](\mes^*_1)}}(x)\right]
\end{split}
\label{eq:inequality_measure_and_reprogram_three}
\end{align}
where $\widetilde{\mathcal{S}}[F^*[x,H],F[x,r](\mes^*_1)]$ and $F^*[x,H]_{\mes^*_1,F[x,r](\mes^*_1)}$ are  as defined in Lemma \ref{lem:measure_and_reprogram}.  
Note that $F^*[x,H]_{\mes^*_1,F[x,r](\mes^*_1)}\equiv F^*[x,H_{\mes^*_1,r}]$  and
$H_{m^*_1,r}$ is uniformly distributed over $\func(\messpace,\randspace)$ if $H$ and $r$ are randomly chosen for any fixed $m^*_1$. 
Therefore, by  
taking the average over all $H$ and $r$ for Eq. \ref{eq:inequality_measure_and_reprogram_three}, for any fixed $\mes^*_1$ we have 
\begin{align}
\begin{split}
&\Pr_{H,r}\left[
\mes'_1=\mes^*_1 \land (\mes'_1,F[x,r](m_1^*),\mes_2)\in \Acc'[x,\revise{r}]
:(\mes'_1,\mes_2)\sample \widetilde{\mathcal{S}}[F^*[x,H],F[x,r](\mes^*_1)](x)\right]\\
&\geq \frac{1}{(2q +1)^{2}}\Pr_{H}\left[
\mes_1=\mes^*_1 \land (\mes_1,F[x,H(m_1^*)](m_1^*),\mes_2)\in \Acc'[x,\revise{H(m_1^*)}]
:(\mes_1,\mes_2) \sample \mathcal{S}^{F^*[x,H]}(x)\right]
\end{split}
\label{eq:inequality_measure_and_reprogram_three_average}
\end{align}
When  
$\mes'_1=\mes^*_1$ and $\mes_1=\mes^*_1$, 
$(\mes'_1,F[x,r](m_1^*),\mes_2)\in \Acc'[x,\revise{r}]$ and
$(\mes_1,F[x,H(m_1^*)](m_1^*),\mes_2)\in \Acc'[x,\revise{H(m_1^*)}]$ are equivalent to 
$(\mes'_1,\mes_2)\in \Acc[x,r]$ and
$(\mes_1,\mes_2)\in \Acc^*[x,H]$, respectively. 
Therefore, by taking a summation over all $\mes^*_1$ for Eq. \ref{eq:inequality_measure_and_reprogram_three_average}, we have 
\begin{align}
\begin{split}
&\sum_{\mes_1^*\in \messpace}\Pr_{H,r}\left[
\mes'_1=\mes^*_1 \land (\mes'_1,\mes_2)\in \Acc[x,r]
:(\mes'_1,\mes_2)\sample \widetilde{\mathcal{S}}[F^*[x,H],F[x,r](\mes^*_1)](x)\right]\\
&\geq \frac{1}{(2q +1)^{2}}\Pr_{H}\left[
\mathcal{S}^{F^*[x,H]}(x)\in \Acc^*[x,H]\right] \geq \frac{1}{\poly(\secpar)}
\end{split}
\label{eq:inequality_measure_and_reprogram_three_average_sum}
\end{align}
for all $x\in L\cap \bit^\secpar$ where the last inequality follows from the assumption of \cref{lem:impossible_three_round_classical_simplified}. 

For any $H$ and $r$, 
we consider an algorithm $\B[H,r](x)$ that ``imitates" the LHS of Eq. \ref{eq:inequality_measure_and_reprogram_three_average_sum}.
Specifically, $\B[H,r](x)$ works as follows:
\begin{description}
\item[$\B{[}H,r{]}(x)$:] 
It works as follows.
\begin{enumerate}
    \item 
    Pick $(j^*,b^*)\sample ([q]\times \bit) \cup \{(\bot,\bot)\}$. 
    \item Run $\mathcal{S}$ 
    where its oracle is simulated by  $\ora$ that is initialized to be $F^*[x,H]$.
     When $\mathcal{S}$ makes its $j$-th query to $\ora$, 
    \begin{enumerate}
      \item If $j=j^*$, 
        measure $\mathcal{S}$'s  query register to obtain $\mes'_1$.
        \begin{enumerate}
        \item If $b^*=0$, reprogram $\ora\leftarrow \reprogram(\ora,\mes'_1, F[x,r](\mes'_1))$ and answer $\mathcal{S}$'s $j$-th query by using the reprogrammed oracle. 
        \item If $b^*=1$, answer  $\mathcal{S}$'s $j$-th query by using the oracle before the reprogramming 
        and then reprogram $\ora\leftarrow \reprogram(\ora,\mes'_1, F[x,r](\mes'_1))$. 
        \end{enumerate}
    \item Otherwise, answer $\mathcal{S}$'s $j$-th query by just using the oracle $\ora$. 
    \end{enumerate}
    \item Let $(\mes_1,\mes_2)$ be $\mathcal{S}$'s output.
    If $j^*=\bot$ (in which case $\mes'_1$ has not been defined), set $\mes'_1:=\mes_1$.
    Output $(\mes'_1,\mes_2)$. 
\end{enumerate}
\end{description}

Then we prove the following claims.
\begin{claim}\label{cla:yes_instance_three}
For any $x\in L\cap \bit^\secpar$, we have 
\begin{align*}
    \Pr_{H,r}\left[\B[H,r](x)\in \Acc[x,r]\right]\geq \frac{1}{\poly(\secpar)}. 
\end{align*}
\end{claim}
\begin{proof}[Proof of \cref{cla:yes_instance_three}]
By definition, we can see that  
$\B[H,r]$ works similarly to 
$\widetilde{\mathcal{S}}[F^*[x,H],F[x,r](\mes^*_1)](x)$
conditioned on that the measured query $\mes'_1$ is equal to $\mes^*$.
Therefore we have 
\begin{align*}
&\Pr_{H,r}\left[
\mes'_1=\mes^*_1 \land (\mes'_1,\mes_2)\in \Acc[x,r]
:(\mes'_1,\mes_2)\sample \B[H,r](x)\right]\\
&=\Pr_{H,r}\left[
\mes'_1=\mes^*_1 \land (\mes'_1,\mes_2)\in \Acc[x,r]
:(\mes'_1,\mes_2)\sample \widetilde{\mathcal{S}}[F^*[x,H],F[x,r](\mes^*_1)](x)\right]
\end{align*}
By substituting this for  the LHS of Eq. \ref{eq:inequality_measure_and_reprogram_three_average_sum}, \cref{cla:yes_instance_three} follows. 
\end{proof}
\begin{claim}\label{cla:no_instance_three}
For any $x\in \bit^\secpar\setminus L$, we have 
\begin{align*}
    \Pr_{H,r}\left[\B[H,r](x)\in \Acc[x,r]\right] \leq \negl(\secpar). 
\end{align*}
\end{claim}
\begin{proof}[Proof of Claim \ref{cla:no_instance_three}]
We construct a cheating prover $P^*$ against the protocol $\Pi$ that wins with probability 
    $\Pr_{H,r}\left[\B[H,r](x)\in \Acc[x,r]\right]$, which immediately implies  Claim \ref{cla:no_instance_three} by the soundness of $\Pi$.
Intuitively, $P^*(x)$ just runs $\B[H,r](x)$ where $H$ is chosen by itself and $r$ is chosen by the external verifier.
Though $P^*$ does not know $r$, it can simulate $\B[H,r](x)$ because it needs $r$ only when responding to the measured query, and $P^*$ can then send such a (classical) query to the external verifier to get the response. 



Formally, $P^*$ is described as follows. We will mark the difference between $P^*$ and $\B[H,r]$ (for $H\sample \mathcal{H}_{4q}$ and $r\sample \randspace$) with \ul{underline}. 
\begin{description}
\item[$P^*(x)$:] 
The cheating prover $P^*$ interacts with the external verifier as follows:
\begin{enumerate}
   \item Choose a function $H\sample \mathcal{H}_{4q}$ where $\mathcal{H}_{4q}$ is a family of $4q$-wise independent hash function, and initialize an oracle $\ora$ to be a (quantumly-accessible) oracle that computes $F^*[x,H]$. 
    \item 
    Pick $(j^*,b^*)\sample ([q]\times \bit) \cup \{(\bot,\bot)\}$. 
    \item Run $\mathcal{S}$ 
    where its oracle is simulated by  $\ora$.
     When $\mathcal{S}$ makes its $j$-th query to $\ora$, 
    \begin{enumerate}
      \item If $j=j^*$, 
        measure $\mathcal{S}$'s  query register to obtain $\mes'_1$.
        \ul{Send $\mes'_1$ to the external verifier as the first message, and receives the response $m_V$.} 
        \begin{enumerate}
        \item If $b^*=0$, \ul{reprogram $\ora\leftarrow \reprogram(\ora,\mes'_1, \mes_V)$} and answer $\mathcal{S}$'s $j$-th query by using the reprogrammed oracle. 
        \item If $b^*=1$, answer  $\mathcal{S}$'s $j$-th query by using the oracle before the reprogramming 
        and then \ul{reprogram $\ora\leftarrow \reprogram(\ora,\mes'_1, \mes_V)$.} 
        \end{enumerate}
    \item Otherwise, answer $\mathcal{S}$'s $j$-th query by just using the oracle $\ora$. 
    \end{enumerate}
    \item Let $(\mes_1,\mes_2)$ be $\mathcal{S}$'s output.
    \ul{If $j^*=\bot$ (in which case $P^*$ has not sent the first message to the external verifier yet), complete the protocol by sending $\mes_1$ and $\mes_2$ as first and second messages to the external verifier (regardless of the verifier's response in the second round).  
    Otherwise, 
    $P^*$ should have already run the protocol until the second round, so it completes the protocol by sending $\mes_2$ to the external verifier as the prover's second message.  
    }
\end{enumerate}
\end{description}

By definitions, we can see that $P^*$ perfectly simulates $\B[H,r]$ for $H\sample \mathcal{H}_{4q}$ and $r\sample \randspace$ where $r$ is chosen by the external verifier. 
Moreover, $P^*$ wins (i.e., the verifier accepts) if and only if the output of $\B[H,r]$ is in $\Acc[x,r]$. 
Moreover, by \cref{lem:simulation_QRO} and that we can simulate $\B[H,r](x)$ by at most $2q$ oracle access to $H$, 
the probability that $\B[H,r](x)\in \Acc[x,r]$ does not change if we choose a completely random function $H$ instead of one from $\mathcal{H}_{4q}$. 
Therefore, the soundness of the protocol ensures $\Pr_{H,r}[\B[H,r](x)\in \Acc[x,r]]\leq \negl(\secpar)$. 
This completes the proof of Claim \ref{cla:no_instance_three}.
\end{proof}

Finally, we conclude the proof of  Lemma \ref{lem:impossible_three_round_classical_simplified} by using Claim \ref{cla:yes_instance_three} and \ref{cla:no_instance_three}.
 Since $\B[H,r](x)$ can be seen as an oracle-aided algorithm that makes at most $2q$ queries to $H$, we have 
 \[
 \Pr_{H\sample \mathcal{H}_{4q},r}\left[\B[H,r](x)\in \Acc[x,r]\right]=\Pr_{H,r}\left[\B[H,r](x)\in \Acc[x,r]\right]
 \]
by \cref{lem:simulation_QRO}.
Then we can decide if a given element $x$ is in $L$ by running $\B[H,r](x)$ for $H\sample \mathcal{H}_{4q}$ and $r\sample \randspace$ and seeing if the output is in $\Acc[x,r]$.
This means $\lang \in \BQP$.
This completes the proof of  Lemma \ref{lem:impossible_three_round_classical_simplified}.
\end{proof}




\begin{thebibliography}{JKMR09}

\bibitem[AHU19]{C:AmbHamUnr19}
Andris Ambainis, Mike Hamburg, and Dominique Unruh.
\newblock Quantum security proofs using semi-classical oracles.
\newblock In Alexandra Boldyreva and Daniele Micciancio, editors, {\em
  CRYPTO~2019, Part~II}, volume 11693 of {\em {LNCS}}, pages 269--295.
  Springer, Heidelberg, August 2019.

\bibitem[BBBV97]{BBBV}
Charles~H. Bennett, Ethan Bernstein, Gilles Brassard, and Umesh Vazirani.
\newblock Strengths and weaknesses of quantum computing.
\newblock {\em SIAM Journal on Computing}, 26(5):1510–1523, Oct 1997.

\bibitem[BCY91]{TCS:BCY91}
Gilles Brassard, Claude Crépeau, and Moti Yung.
\newblock Constant-round perfect zero-knowledge computationally convincing
  protocols.
\newblock {\em Theoretical Computer Science}, 84(1):23--52, 1991.

\bibitem[BJY97]{EC:BelJakYun97}
Mihir Bellare, Markus Jakobsson, and Moti Yung.
\newblock Round-optimal zero-knowledge arguments based on any one-way function.
\newblock In Walter Fumy, editor, {\em EUROCRYPT'97}, volume 1233 of {\em
  {LNCS}}, pages 280--305. Springer, Heidelberg, May 1997.

\bibitem[BL02]{STOC:BarLin02}
Boaz Barak and Yehuda Lindell.
\newblock Strict polynomial-time in simulation and extraction.
\newblock In {\em 34th ACM STOC}, pages 484--493. {ACM} Press, May 2002.

\bibitem[Blu86]{Blum86}
Manuel Blum.
\newblock How to prove a theorem so no one else can claim it.
\newblock In {\em Proceedings of the International Congress of Mathematicians},
  page 1444–1451, 1986.

\bibitem[BS20]{STOC:BitShm20}
Nir Bitansky and Omri Shmueli.
\newblock Post-quantum zero knowledge in constant rounds.
\newblock In Konstantin Makarychev, Yury Makarychev, Madhur Tulsiani, Gautam
  Kamath, and Julia Chuzhoy, editors, {\em 52nd ACM STOC}, pages 269--279.
  {ACM} Press, June 2020.

\bibitem[CCY21]{CCY20}
Nai-Hui Chia, Kai-Min Chung, and Takashi Yamakawa.
\newblock A black-box approach to post-quantum zero-knowledge in constant
  rounds.
\newblock CRYPTO 2021 (To appear), 2021.

\bibitem[DFM20]{C:DonFehMaj20}
Jelle Don, Serge Fehr, and Christian Majenz.
\newblock The measure-and-reprogram technique 2.0: Multi-round fiat-shamir and
  more.
\newblock In Daniele Micciancio and Thomas Ristenpart, editors, {\em
  CRYPTO~2020, Part~III}, volume 12172 of {\em {LNCS}}, pages 602--631.
  Springer, Heidelberg, August 2020.

\bibitem[DFMS19]{C:DFMS19}
Jelle Don, Serge Fehr, Christian Majenz, and Christian Schaffner.
\newblock Security of the {Fiat}-{Shamir} transformation in the quantum
  random-oracle model.
\newblock In Alexandra Boldyreva and Daniele Micciancio, editors, {\em
  CRYPTO~2019, Part~II}, volume 11693 of {\em {LNCS}}, pages 356--383.
  Springer, Heidelberg, August 2019.

\bibitem[FGJ18]{EC:FleGoyJai18}
Nils Fleischhacker, Vipul Goyal, and Abhishek Jain.
\newblock On the existence of three round zero-knowledge proofs.
\newblock In Jesper~Buus Nielsen and Vincent Rijmen, editors, {\em
  EUROCRYPT~2018, Part~III}, volume 10822 of {\em {LNCS}}, pages 3--33.
  Springer, Heidelberg, April~/~May 2018.

\bibitem[FS90]{C:FeiSha89}
Uriel Feige and Adi Shamir.
\newblock Zero knowledge proofs of knowledge in two rounds.
\newblock In Gilles Brassard, editor, {\em CRYPTO'89}, volume 435 of {\em
  {LNCS}}, pages 526--544. Springer, Heidelberg, August 1990.

\bibitem[GK96a]{JC:GolKah96}
Oded Goldreich and Ariel Kahan.
\newblock How to construct constant-round zero-knowledge proof systems for
  {NP}.
\newblock {\em Journal of Cryptology}, 9(3):167--190, June 1996.

\bibitem[GK96b]{SIAM:GK96}
Oded Goldreich and Hugo Krawczyk.
\newblock On the composition of zero-knowledge proof systems.
\newblock {\em SIAM Journal on Computing}, 25(1):169--192, 1996.

\bibitem[GMR89]{GolMicRac89}
Shafi Goldwasser, Silvio Micali, and Charles Rackoff.
\newblock The knowledge complexity of interactive proof systems.
\newblock {\em {SIAM} Journal on Computing}, 18(1):186--208, 1989.

\bibitem[GMW91]{JACM:GMW91}
Oded Goldreich, Silvio Micali, and Avi Wigderson.
\newblock Proofs that yield nothing but their validity for all languages in
  {NP} have zero-knowledge proof systems.
\newblock {\em J. {ACM}}, 38(3):691--729, 1991.

\bibitem[HRS15]{ePrint:HRS}
Andreas Hülsing, Joost Rijneveld, and Fang Song.
\newblock Mitigating multi-target attacks in hash-based signatures.
\newblock Cryptology ePrint Archive, Report 2015/1256, 2015.
\newblock \url{https://eprint.iacr.org/2015/1256}.

\bibitem[JKMR09]{JKMR09}
Rahul Jain, Alexandra Kolla, Gatis Midrijanis, and Ben~W. Reichardt.
\newblock On parallel composition of zero-knowledge proofs with black-box
  quantum simulators.
\newblock {\em Quantum Inf. Comput.}, 9(5{\&}6):513--532, 2009.

\bibitem[Kat08]{TCC:Katz08a}
Jonathan Katz.
\newblock Which languages have 4-round zero-knowledge proofs?
\newblock In Ran Canetti, editor, {\em TCC~2008}, volume 4948 of {\em {LNCS}},
  pages 73--88. Springer, Heidelberg, March 2008.

\bibitem[KRR17]{C:KalRotRot17}
Yael~Tauman Kalai, Guy~N. Rothblum, and Ron~D. Rothblum.
\newblock From obfuscation to the security of {Fiat}-{Shamir} for proofs.
\newblock In Jonathan Katz and Hovav Shacham, editors, {\em CRYPTO~2017,
  Part~II}, volume 10402 of {\em {LNCS}}, pages 224--251. Springer, Heidelberg,
  August 2017.

\bibitem[PW09]{TCC:PasWee09}
Rafael Pass and Hoeteck Wee.
\newblock Black-box constructions of two-party protocols from one-way
  functions.
\newblock In Omer Reingold, editor, {\em TCC~2009}, volume 5444 of {\em
  {LNCS}}, pages 403--418. Springer, Heidelberg, March 2009.

\bibitem[SV03]{SahVad03}
Amit Sahai and Salil~P. Vadhan.
\newblock A complete problem for statistical zero knowledge.
\newblock {\em J. {ACM}}, 50(2):196--249, 2003.

\bibitem[Unr12]{EC:Unruh12}
Dominique Unruh.
\newblock Quantum proofs of knowledge.
\newblock In David Pointcheval and Thomas Johansson, editors, {\em
  EUROCRYPT~2012}, volume 7237 of {\em {LNCS}}, pages 135--152. Springer,
  Heidelberg, April 2012.

\bibitem[Unr15]{JACM:Unruh15}
Dominique Unruh.
\newblock Revocable quantum timed-release encryption.
\newblock {\em J. {ACM}}, 62(6):49:1--49:76, 2015.

\bibitem[Unr16]{EC:Unruh16}
Dominique Unruh.
\newblock Computationally binding quantum commitments.
\newblock In Marc Fischlin and Jean-S{\'{e}}bastien Coron, editors, {\em
  EUROCRYPT~2016, Part~II}, volume 9666 of {\em {LNCS}}, pages 497--527.
  Springer, Heidelberg, May 2016.

\bibitem[Wat09]{SIAM:Watrous09}
John Watrous.
\newblock Zero-knowledge against quantum attacks.
\newblock {\em {SIAM} J. Comput.}, 39(1):25--58, 2009.

\bibitem[YZ21]{YZ20}
Takashi Yamakawa and Mark Zhandry.
\newblock Classical vs quantum random oracles.
\newblock Eurocrypt 2021 (To appear), 2021.

\bibitem[Zha12a]{FOCS:zhandry12}
Mark Zhandry.
\newblock How to construct quantum random functions.
\newblock In {\em 53rd FOCS}, pages 679--687. {IEEE} Computer Society Press,
  October 2012.

\bibitem[Zha12b]{C:Zhandry12}
Mark Zhandry.
\newblock Secure identity-based encryption in the quantum random oracle model.
\newblock In Reihaneh Safavi-Naini and Ran Canetti, editors, {\em CRYPTO~2012},
  volume 7417 of {\em {LNCS}}, pages 758--775. Springer, Heidelberg, August
  2012.

\end{thebibliography}

\appendix
\ifnum\submission=0
\else
\newpage
 	\setcounter{page}{1}
 	{
	\noindent
 	\begin{center}
	{\Large SUPPLEMENTAL MATERIALS}
	\end{center}
 	}
	\setcounter{tocdepth}{2}
\fi
\newpage
  \tableofcontents
  \thispagestyle{empty}
\end{document}